\documentclass[a4paper]{article}

\usepackage{amsmath}
\usepackage{amssymb}
\usepackage{amsthm}
\usepackage{bussproofs}
\usepackage{enumitem}
\usepackage{latexsym}
\usepackage{xspace}

\usepackage{color}

\newenvironment{tbs}{%
   \small\tt
   \begin{itemize}}{\end{itemize}}
\newcommand{\btbs}{\begin{tbs}}                                                                      
\newcommand{\etbs}{\end{tbs}}

\newcommand{\hide}[1]{}

\newtheorem{theorem}{Theorem}[section]

\newtheorem{fact}[theorem]{Fact}
\newtheorem{proposition}[theorem]{Proposition}

\newtheorem{corollary}[theorem]{Corollary}
\newtheorem{defi}[theorem]{Definition}

\newtheorem{conv}[theorem]{Convention}
\newtheorem{rema}[theorem]{Remark}
\newtheorem{exam}[theorem]{Example}

\newenvironment{definition}{\begin{defi}\rm}{\hfill $\lhd$\end{defi}}

\newenvironment{remark}{\begin{rema}\rm}{\hfill $\lhd$\end{rema}}
\newenvironment{example}{\begin{exam}\rm}{\hfill $\lhd$\end{exam}}

\newenvironment{urlist}{\begin{enumerate}[topsep=0pt,itemsep=-1ex,partopsep=1ex,parsep=1ex,%
    label={\arabic*)}]
    }{\end{enumerate}}

\newenvironment{proofof}[1]{\begin{trivlist}\item[\hskip\labelsep{\bf
Proof~of~{#1}.\ }]}{\hspace*{\fill} {\sc qed}\end{trivlist}}

\newtheorem{claim2}{\sc Claim}
\newenvironment{claim}{\begin{claim2}\rm}{\end{claim2}\rm}
\newenvironment{claimfirst}{\setcounter{claim2}{0}
               \begin{claim2}\rm}{\end{claim2}\rm}

\newenvironment{pfclaim}{\begin{trivlist}\item[]{\sc Proof of
Claim}}{\hfill {\mbox{$\blacktriangleleft$}}\end{trivlist}}

\newcommand{\bv}{\bigvee}
\newcommand{\dia}{\Diamond}

\newcommand{\ntr}{\checkmark}

\newcommand{\bdual}[1]{#1^{\partial}}

\newcommand{\atneg}[1]{\overline{#1}}
\newcommand{\sneg}[1]{\underline{#1}}

\newcommand{\isbnf}{\;::=\;}
\newcommand{\divbnf}{\;\mid\;}

\newcommand{\Prop}{\mathsf{Prop}}
  \newcommand{\PropP}{\mathsf{P}}
  \newcommand{\PropQ}{\mathsf{Q}}

\newcommand{\Lit}{\mathsf{Lit}}

\newcommand{\muML}{\mathcal{L}_{\mu}}
\newcommand{\AFMC}{\muML^{\mathit{af}}}
\newcommand{\nth}[2]{\mathrm{N}^{#1}_{#2}(\AFMC)}

\newcommand{\ol}[1]{\overline{#1}}   
\newcommand{\Clos}{\mathsf{Clos}}
   
\newcommand{\FV}{\mathit{FV}}
\newcommand{\BV}{\mathit{BV}}

\newcommand{\thin}[1]{#1^{-}}
\newcommand{\morefocus}[2]{R^{-}({#1},{#2})}
\newcommand{\uls}[1]{\widetilde{#1}}
\newcommand{\backcl}{Q}
\renewcommand{\morefocus}[2]{F({#1},{#2})}

\newcommand{\allfocus}[1]{{#1}^f}

\newcommand{\rcla}[1]{\rightarrow_{C}^{#1}}
\newcommand{\rclat}[1]{\twoheadrightarrow_{C}^{#1}}

\newcommand{\tracestep}{\to_{C}}

\newcommand{\atrail}{\mathsf{A}} 
\newcommand{\ptrail}{\mathsf{P}} 
\newcommand{\gtrail}{\mathsf{T}} 

\newcommand{\rdc}[1]{\wh{#1}}

\newcommand{\Vitp}{\mathsf{X}}
\newcommand{\ip}[1]{\theta(#1)}

\newcommand{\fitp}[1]{\theta_{#1}}

\newcommand{\mng}[1]{[\![ #1 ]\!]}

\newcommand{\mathstr}[1]{\mathbb{#1}}

\newcommand{\bbA}{\mathstr{A}}
\newcommand{\bbB}{\mathstr{B}}

\newcommand{\bbG}{\mathstr{G}}

\newcommand{\bbS}{\mathstr{S}}
\newcommand{\bbT}{\mathstr{T}}

\newcommand{\Focus}{\ensuremath{\mathsf{Focus}}\xspace}
\newcommand{\Focusinf}{\ensuremath{\mathsf{Focus}_\infty}\xspace}

\newcommand{\ann}{a}
\newcommand{\dann}{\ol{a}}

\newcommand{\AXC}[1]{\AxiomC{#1}}

\newcommand{\UIC}[1]{\UnaryInfC{#1}}
\newcommand{\BIC}[1]{\BinaryInfC{#1}}

\newcommand{\AxLit}{\ensuremath{\mathsf{Ax1}}\xspace}
\newcommand{\AxTop}{\ensuremath{\mathsf{Ax2}}\xspace}
\newcommand{\RuOr}{\ensuremath{\mathsf{R}_{\lor}}\xspace}
\newcommand{\RuAnd}{\ensuremath{\mathsf{R}_{\land}}\xspace}
\newcommand{\RuBox}{\ensuremath{\mathsf{\mathsf{R}_{\Box}}}\xspace}
\newcommand{\RuMod}{\ensuremath{\mathsf{\mathsf{M}}}\xspace}
\newcommand{\RuFp}[1]{\ensuremath{\mathsf{R}_{#1}}\xspace}
\newcommand{\RuMu}{\RuFp{\mu}}
   
\newcommand{\RuNu}{\RuFp{\nu}}

\newcommand{\RuWeak}{\ensuremath{\mathsf{W}}\xspace}
\newcommand{\RuFocus}{\ensuremath{\mathsf{F}}\xspace}
   
   \newcommand{\RuFocustot}{\ensuremath{\mathsf{F}^t}\xspace}
\newcommand{\RuUnfocus}{\ensuremath{\mathsf{U}}\xspace}
\newcommand{\RuDischarge}[1]{\ensuremath{\mathsf{D}^{#1}}\xspace}

\newcommand{\Ru}{\mathsf{R}}

\newcommand{\Tokens}{\mathcal{D}}

\newcommand{\dx}{\ensuremath{\mathsf{x}}}
\newcommand{\dy}{\ensuremath{\mathsf{y}}}
\newcommand{\dz}{\ensuremath{\mathsf{z}}}

\newcommand{\tLab}{\mathsf{Q}}
\newcommand{\pLab}{\mathsf{R}}

\newcommand{\fsr}{\mathsf{l}}

\newcommand{\game}[1]{\mathcal{G}({#1})}
   
\newcommand{\EG}{\mathcal{E}}
\newcommand{\Own}{O}
\newcommand{\WC}{W}
\newcommand{\Om}{\Omega}
\newcommand{\SM}{\mathcal{M}_{\nu}}
\newcommand{\last}{\mathsf{last}}
\newcommand{\first}{\mathsf{first}}
\newcommand{\PM}{\mathit{PM}}
\newcommand{\Eloise}{Eloise\xspace}
\newcommand{\Abelard}{Abelard\xspace}
\newcommand{\Prover}{Prover\xspace}
\newcommand{\Refuter}{Refuter\xspace}
  \newcommand{\eloi}{\exists}
  \newcommand{\abel}{\forall}
\newcommand{\Win}{\mathit{Win}}
\newcommand{\InfP}{\mathsf{InfPath}}
\newcommand{\Inf}{\mathsf{Inf}}

\newcommand{\comp}{\mathop{;}}
\newcommand{\powerset}{\mathcal{P}}

\newcommand{\itv}[2]{[#1, #2]}

\newcommand{\isdef}{\mathrel{:=}}
\newcommand{\nada}{\varnothing}
\newcommand{\sse}{\subseteq}
\newcommand{\sz}[1]{|#1|}
\newcommand{\Dom}{\mathsf{Dom}}
\newcommand{\Ran}{\mathsf{Ran}}

\newcommand{\HF}{\mathit{HF}}

\newcommand{\al}{\alpha}
\newcommand{\be}{\beta}
\newcommand{\ga}{\gamma}

\newcommand{\si}{\sigma}
\newcommand{\om}{\omega}

\newcommand{\Ga}{\Gamma}
\newcommand{\De}{\Delta}

\newcommand{\Si}{\Sigma}
\renewcommand{\phi}{\varphi}

\newcommand{\wh}[1]{\widehat{#1}}

\author{Johannes Marti and Yde Venema}

\title{Focus-style proof systems and interpolation for the
alternation-free $\mu$-calculus}

\begin{document}

\maketitle

\begin{abstract}
In this paper we introduce a cut-free sequent calculus for the
alternation-free fragment of the modal $\mu$-calculus. This system
allows for circular proofs and uses a simple focus mechanism to control
the unravelling of fixpoints along infinite branches. We show that the
proof system is sound and complete and apply it to prove that the
alternation-free fragment has the Craig interpolation property.
\end{abstract}

\section{Introduction}

In this paper we present a circular proof system for the
alternation-free fragment of the modal $\mu$-calculus and use this
system to proof Craig interpolation for the alternation-free fragment.

\subsection{The alternation-free $\mu$-calculus}

The modal $\mu$-calculus, introduced by Kozen~\cite{koze:resu83}, is a
logic for describing properties of processes that are modelled by labelled
transition systems. It extends the expressive power of propositional
modal logic by means of least and greatest fixpoint operators. This
addition permits the expression of all monadic second-order properties
of transition systems~\cite{jani:auto95}. The $\mu$-calculus is
generally regarded as a universal specification language, since it
embeds most other logics that are used for this purpose, such as
\textsc{ltl}, \textsc{ctl}, \textsc{ctl}$^{*}$ and \textsc{pdl}.

The alternation-free $\mu$-calculus is a fragment of the $\mu$-calculus
in which there is no interaction between least and greatest fixpoint
operators. It can be checked that the translations of both \textsc{ctl}
and \textsc{pdl} into the $\mu$-calculus yield alternation-free
formulas. 
Over tree structures, or when restricted to bisimulation-invariant properties, 
the expressive power of the alternation-free $\mu$-calculus corresponds to
monadic second-order logic where the quantification is restricted to sets that
are finite, or in a suitable sense well-founded \cite{niwi:fixe97,facc:char13}. 
For more restricted classes of structures, such as for instance infinite words,
it can be shown that the alternation-free fragment already has the same
expressivity as the full $\mu$-calculus \cite{kaiv:axio95,guti:muca14}.

Many theoretical results on the modal $\mu$-calculus depend on
the translation from formulas in the $\mu$-calculus to automata
\cite{jani:auto95,wilk:alte01}. The general idea is to construct for
every formula an automaton that accepts precisely the pointed structures
where the formula is true. For the alternation-free fragment the
codomain of this translation can be taken to consist of weak alternating
automata \cite{facc:char13,guti:muca14}. These are parity automata for
which the assignment of priorities to states is restricted such that all
states from the same strongly connected component have the same priority.

\subsection{A cyclic focus system for the alternation-free $\mu$-calculus}

In the theory of the modal $\mu$-calculus automata- and game-theoretic
approaches have long been at the centre of attention. Apart from the
rather straightforward tableau games by Niw\'{i}nski \& Walukiewicz
\cite{niwi:game96} there have for a long time been few successful
applications of proof-theoretic techniques. This situation has changed
with a recent breakthrough by Afshari \& Leigh \cite{afsh:cutf17}, who
obtain completeness of Kozen's axiomatization of the modal
$\mu$-calculus using purely proof-theoretic arguments. The proof of this
result can be taken to consist of a series of
proof transformations: First, it starts with a successful infinite
tableau in the sense of \cite{niwi:game96}. Second, one then adds a
mechanism for annotating formulas that was developed by
Jungteerapanich and Stirling \cite{jung:tabl10,stir:tabl14} to detect
after finitely many steps when a branch of the tableau tree may develop into
a successful infinite branch, thus obtaining a finite but cyclic tableau.
Third, Afshari \& Leigh show how to apply a series of
transformations to this finite annotated tableau to obtain a proof in a cyclic
sequent system for the model $\mu$-calculus. 
Fourth, and finally, this proof can be turned into a Hilbert-style proof in 
Kozen's axiomatization.

In this paper we present an annotated cyclic proof system for the
alternation-free $\mu$-calculus that corresponds roughly to the annotated
tableaux of Jungteerapanich and Stirling mentioned in the second
step above. But, whereas in the system for the full $\mu$-calculus these
annotations are sequences of names for fixpoint variables, for the 
alternation-free fragment it suffices to annotate formulas with just one
bit of information. We think of this bit as indicating whether a
formula is in what we call \emph{in focus} or whether it is
\emph{unfocused}. We use this terminology because our proof system for
the alternation-free $\mu$-calculus is a generalization of the focus
games for weaker fixpoint logics such as \textsc{ltl} and \textsc{ctl}
by Lange \& Stirling \cite{lang:focu01}. These are games based on a tableau
such that at every sequent of the tableau there is exactly one formula
in focus. 
In our system we generalise this so that a proof node may feature a \emph{set}
of formulas in focus.

Our system can be shown to be complete while only allowing for two kind
of manipulations of annotations. The first is the rule that unfolds
least fixpoints. Whenever one is unfolding a least fixpoint formula that
is in focus at the current sequent then its unfolding in the sequent
further away from the root needs to be unfocused. Unfolding greatest
fixpoints has no influence on the annotations. That other manipulation
of annotations is by a focus rule that puts previously unfocused
formulas into focus. It suffices to only apply this rule if the current
sequence does not contain any formula that is in focus. The rule then
simply continues the proof search with the same formulas but now they
are all in focus.

The design of the annotation mechanisms in the tableau by
Jungteerapanich \& Stirling and in the focus system from this paper are
heavily influenced by ideas from automata theory. It was already
observed by Niw\'{i}nski \& Walukiewicz \cite{niwi:game96} that a tree
automaton can be used that accepts precisely the trees that encode
successful tableaux. This automaton is the product of a tree automaton
checking for local consistency of the tableau and a deterministic
automaton over infinite words that detects whether every branch in the
tableau is successful. That a branch in a tableau is successful means
that it carries at least one trail of formulas where the most
significant fixpoint that is unravelled infinitely often is a greatest
fixpoint. It is relatively straight-forward to give a nondeterministic
automaton that detects successful branches, but the construction needs a
deterministic automaton, which is obtained using the Safra construction
\cite{safr:comp88}. The crucial insight of Jungteerapanich \& Stirling
\cite{jung:tabl10,stir:tabl14} is that this deterministic automaton that
results from the Safra construction can be encoded inside the tableau by 
using annotations of formulas.

The relation between detecting successful branches in a proof and the
determinization of automata on infinite words can also be seen more
directly. In the proof system annotations are used to detect whether a
branch of the proof carries at least one trail such that the most
significant fixpoint that is unfolded infinitely often on the trail is a
greatest fixpoint.
This is analogous to a problem that arises when one tries to use the
powerset construction to construct an equivalent deterministic automaton
from a given non-deterministic parity automaton operating on infinite
words. The problem there is to determine whether a sequence of
macrostates of the deterministic automaton carries a run of the original
non-deterministic automaton that satisfies the parity condition. It is
possible to view the annotated sequents of Jungteerapanich \& Stirling
as a representation of the Safra trees which provide the states of a
deterministic Muller automaton that one obtains when determinizing a
non-deterministic parity automaton \cite[sec.~4.3.5]{jung:tabl10}.

For alternation-free formulas it is significantly simpler to detect
successful branches, because one can show that the fixpoints that are
unravelled infinitely often on a trail of alternation-free formulas are
either all least or all greatest fixpoints. One can compare the problem
of finding such a trail to the problem of recognising a successful run
of a non-deterministic weak stream automaton in the macrostates of a
determinization of the automaton. In fact the focus mechanism from the
proof system that we develop in this paper can also be used to transform
a non-deterministic weak automaton into an equivalent deterministic
co-B\"{u}chi automaton. This relatively simple construction is a special
case of Theorem~15.2.1 in \cite{demr:temp16}, which shows that every
non-deterministic co-B\"{u}chi automaton can be transformed into an
equivalent deterministic co-B\"{u}chi automaton.


%
%
%

\subsection{Interpolation for the alternation-free $\mu$-calculus}

We apply the proof system introduced in this report to prove that the
alternation-free $\mu$-calculus has Craig's interpolation property. This
means that for any two alternation-free formulas $\varphi$ and $\psi$
such that $\varphi \rightarrow \psi$ is valid there is an interpolant
$\chi$ of $\varphi$ and $\psi$ in the alternation-free $\mu$-calculus.
An interpolant $\chi$ of $\varphi$ and $\psi$ is a formula which contains
only propositional letters that occur in both $\varphi$ and $\psi$ such
that both $\varphi \rightarrow \chi$ an $\chi \rightarrow \psi$ are
valid.

Basic modal logic \cite{gabb:ipol05} and the full $\mu$-calculus
\cite{dago:logi00} have Craig interpolation. 
In fact both formalisms enjoy a even stronger property called uniform 
interpolation, where the interpolant $\chi$ only depends on $\varphi$ and the 
set of propositional letters that occur in $\psi$ (but not on the formula 
$\psi$ itself). 
Despite these strong positive results, interpolation is certainly not
guaranteed to hold for fixpoint logics. For instance, even Craig 
interpolation fails for weak temporal logics or epistemic logics with 
a common knowledge modality \cite{maks:temp91,stud:ckbp09}.
Moreover, one can show that uniform interpolations fails for both
\textsc{pdl} and for the alternation-free $\mu$-calculus~\cite{dago:logi00}. 
The argument relies on the observation that uniform interpolation
corresponds to the definability of bisimulation quantifiers. But, adding
bisimulation quantifiers to \textsc{pdl}, or the alternation-free
fragment, allows the expression of arbitrary fixpoints and thus increases
the expressive power to the level of the full $\mu$-calculus.
It is still somewhat unclear whether \textsc{pdl} has Craig interpolation. 
Various proofs have been proposed, but they have either been retracted 
or still wait for a proper verification~\cite{borz:tabl88,borz:proo20}.

The uniform interpolation result for the modal $\mu$-calculus has been 
generalised to the wider setting of coalgebraic fixpoint 
logic~\cite{mart:unif15,enqv:disj19}, 
but the proofs known for these results are all automata-theoretic in nature.
Recently, however, Afshari \& Leigh~\cite{afsh:lyndxx} pioneered the use of
proof-theoretic methods in fixpoint logics, to prove, among other things, 
a Lyndon-style interpolation theorem for the (full) modal $\mu$-calculus.
Their proof, however, does not immediately yield interpolation results for 
fragments of the logic; in particular, for any pair of alternation-free formulas
of which the implication is valid, their approach will yield an interpolant
inside the full $\mu$-calculus, but not necessarily one that is itself
alternation free.
It is here that the simplicity of our focus-style proof system comes in.

Summarising our interpolation proof for the alternation-free $\mu$-calculus,
we base ourselves on Maehara's method, adapted to the setting of cyclic proofs.
Roughly, the idea underlying Maehara's method is that, given a proof $\Pi$ for 
an implication $\varphi \rightarrow \psi$ one defines the interpolant $\chi$ by 
an induction on the complexity of the proof $\Pi$. 
The difficulty in applying this method to cyclic proof systems is that here,
some proof leaves may not be axiomatic and thus fail to have a trivial
interpolant. 
In particular, a discharged leaf indicates an infinite continuation of the 
current branch. 
Such a leaf introduces a fixpoint variable into the interpolant, which will be 
bound later in the induction. 
The crux of our proof, then, lies in the the way that we handle the additional
complications that arise in correctly managing the annotations in our proof 
system, in order to make sure that these interpolants belong to the right 
fragment of the logic.

\subsection{Overview}

This paper is organized as follows: The preliminaries about the syntax
and semantics of the $\mu$-calculus and its alternation-free fragment
are covered in Section~\ref{s:prel}. 
In Section~\ref{sec-tab} we present our version of the tableau games by 
Niw\'{i}nski \& Walukiewicz that we use later as an intermediate step in the 
soundness and completeness proofs for our proof system. 
In Section~\ref{sec-proofsystem} we introduce our focus system for the
alternation-free $\mu$-calculus and we prove some basic results about the 
system. 
The sections \ref{s:soundness}~and~\ref{s:completeness} contain the proofs of
soundness and completeness of the focus system. 
In Section~\ref{sec-itp} we show how to use the focus system to prove 
interpolation for the alternation-free $\mu$-calculus.

\section{Preliminaries}
\label{s:prel}

We first fix some terminology related to relations and trees and then
discuss the syntax and semantics of the $\mu$-calculus and its
alternation-free fragment.

\subsection{Relations and trees}

Given a binary relation $R \sse S \times S$, we let $R^{-1}$, $R^{+}$ and $R^{*}$
denote, respectively, the converse, the transitive closure and the 
reflexive-transitive closure of $R$.
For a subset $S \sse T$, we write $R[S] \isdef \{ t \in T \mid Rst \text{ for
some } s \in S \}$; in the case of a singleton, we write $R[s]$ rather than
$R[\{s\}]$.
Elements of $R(s)$ and $R^{-1}(s)$ are called, respectively, \emph{successors} 
and \emph{predecessors} of $s$.
An \emph{$R$-path} of length $n$ is a sequence $s_{0}s_{1}\cdots s_{n}$ (with 
$n \geq 0$ such that $Rs_{i}s_{i+1}$ for all $0\leq i<n$); we say that such a 
path \emph{leads from $s_{0}$ to $s_{n}$}.
Similarly, an \emph{infinite path starting at $s$} is a sequence $(s_{n})_{n\in\om}$
 such that $Rs_{i}s_{i+1}$ for all $i<\om$.

A structure $\bbT = (T,R)$, with $R$ a binary relation on $T$, is a \emph{tree} 
if there is a node $r$ such that for every $t \in T$ there is a unique path 
leading from $r$ to $t$.
The node $r$, which is then characterized as the only node in $T$ without
predecessors, is called the \emph{root} of $\bbT$.
Every non-root node $u$ has a unique predecessor, which is called the 
\emph{parent} of $u$; conversely, the successors of a node $t$ are sometimes 
called its \emph{children}.
If $R^{*}tu$ we call $u$ a \emph{descendant} of $t$ and, conversely, $t$ an
\emph{ancestor} of $u$; in case $R^{+}tu$ we add the adjective `proper'.
If $s$ is an ancestor of $t$ we define the \emph{interval} $\itv{s}{t}$ as the
set of nodes on the (unique) path from $s$ to $t$.
A \emph{branch} of a tree is a path that starts at the root.
A \emph{leaf} of a tree is a node without successors.
For nodes of a tree we will generally use the letters $s,t,u,v,\ldots$, for 
leaves we will use $l,m, \ldots$\ .
The \emph{depth} of a node $u$ in a finite tree $\bbT = (T,R)$ is the 
maximal length of a path leading from $u$ to a leaf of $\bbT$. 
The \emph{hereditarily finite} part of a tree $\bbT = (T,R)$ is the subset
$\HF(\bbT) \isdef \{ t \in T \mid R^{*}[t] \text{ is finite} \}$.

A \emph{tree with back edges} is a structure of the form $(T,R,c)$ such that $c$
is a partial function on the collection of leaves, mapping any leaf $l \in 
\Dom(c)$ to one of its proper ancestors; this node $c(l)$ will be called the 
\emph{companion} of $l$.

\subsection{The modal $\mu$-calculus and its alternation-free fragment}

In this part we review syntax and semantics of the modal $\mu$-calculus
and discuss its alternation-free fragment.

\subsubsection{The modal $\mu$-calculus}

\paragraph{Syntax}
The \emph{formulas} in the modal $\mu$-calculus are generated by the
grammar
\[
\phi \isbnf 
   p \divbnf \atneg p 
   \divbnf \bot \divbnf \top
   \divbnf (\phi\lor\phi) \divbnf (\phi\land\phi) \divbnf
   \dia\phi \divbnf \Box\phi 
   \divbnf \mu x\, \phi \divbnf \nu x\, \phi,
\]
where $p$ and $x$ are taken from a fixed set $\Prop$ of propositional
variables and in formulas of the form $\mu x. \phi$ and $\nu x.
\phi$ there are no occurrences of $\atneg x$ in $\phi$. We write
$\muML$ for the set of formulas in the modal $\mu$-calculus.

Formulas of the form $\mu x . \phi$ ($\nu x . \phi$) are called 
\emph{$\mu$-formulas} (\emph{$\nu$-formulas}, respectively); formulas 
of either kind are called \emph{fixpoint formulas}.
The operators $\mu$ and $\nu$ are called fixpoint operators. 
We use $\eta \in \{\mu,\nu\}$ to denote an arbitrary fixpoint operator and
write $\ol{\eta} \isdef \nu$ if $\eta = \mu$ and $\ol{\eta} = \mu$ if $\eta = 
\nu$.
Formulas that are of the form $\Box \phi$ or $\Diamond \phi$ are 
called \emph{modal}. 
Formulas of the form $\phi \land \psi$ or $\phi \lor \psi$ are
called \emph{boolean}.
Formulas of the form $p$ or $\atneg p$ for some $p \in \Prop$ are called
\emph{literals} and the set of all literals is denoted by $\Lit$; a formula is 
\emph{atomic} if it is either a literal or an atomic constant, that is, $\top$ 
or $\bot$.

We use standard terminology for the binding of variables by the fixpoint
operators and for substitutions. In particular we write $\FV(\phi)$
for the set of variables that occur freely in $\phi$ and
$\BV(\phi)$ for the set of all variables that are bound by some
fixpoint operator in $\phi$. We do count occurrences of $\atneg{x}$
as free occurrences of $x$. Unless specified otherwise, we assume that
all formulas $\phi \in \muML$ are \emph{tidy} in the sense
$\FV(\phi) \cap \BV(\phi) = \nada$. Given formulas $\phi$ and
$\psi$ and a propositional variable $x$ such that there is no
occurrences of $\atneg{x}$ in $\phi$, we let $\phi[\psi / x]$
denote the formula that results from substituting all free occurrences
of $x$ in $\phi$ by the formula $\psi$. We only apply this
substitution in situations where $\FV(\psi) \cap \BV(\phi) = \nada$.
This guarantees that no variable capture will occur. If the
variable that is substituted is clear from the context we also write
$\phi(\psi)$ for $\phi[\psi/x]$.
An important use of substitutions of formulas are the unfolding of
fixpoint formulas. Given a fixpoint formula $\xi = \eta x . \chi$ its
\emph{unfolding} is the formula $\chi[\xi/x]$.

Given a formula $\phi \in \muML$ we define its \emph{negation} $\ol{\phi}$
as follows.
First, we define the \emph{boolean dual} $\bdual{\phi}$ of $\phi$ using the 
following induction.
\[\begin{array}{lllclll}
   \bdual{\bot}          & \isdef & \top       & \hspace*{1cm}
  & \bdual{\top}         & \isdef & \bot
\\ \bdual{(\atneg{p})}        & \isdef & \atneg{p}
  && \bdual{p}           & \isdef & p
\\ \bdual{(\phi\lor\psi)}  & \isdef & \bdual{\phi} \land \bdual{\psi}
  && \bdual{(\phi\land\psi)}& \isdef & \bdual{\phi} \lor \bdual{\psi}
\\ \bdual{(\dia\phi)}  & \isdef & \Box \bdual{\phi}
  && \bdual{(\Box \phi)} & \isdef & \dia\bdual{\phi}
\\ \bdual{(\mu x.\phi)}    & \isdef & \nu x.\bdual{\phi}
  && \bdual{(\nu x.\phi)}  & \isdef & \mu x.\bdual{\phi}
\end{array}\]
Based on this definition, we define the formula $\ol{\phi}$ as the
formula $\bdual{\phi}[p \leftrightharpoons \atneg{p} \mid p \in
\FV(\phi)]$ that we obtain from $\bdual{\phi}$ by replacing all
occurrences of $p$ with $\atneg{p}$, and vice versa, for all free
proposition letters $p$ in $\phi$. Observe that if $\phi$ is tidy then
so is $\ol{\phi}$. 


For every formula $\phi \in \muML$ define the set $\Clos_0(\phi)$ as follows
\[\begin{array}{lll l lll}
   \Clos_0(p) & \isdef & \nada
&& \Clos_0(\atneg{p}) & \isdef & \nada
\\ \Clos_0(\psi_0 \land \psi_1) & \isdef & \{ \psi_0, \psi_1 \}
&& \Clos_0(\psi_0 \lor \psi_1) & \isdef & \{ \psi_0, \psi_1 \}
\\ \Clos_0(\Box\psi) & \isdef & \{ \psi \}
&& \Clos_0(\Diamond\psi) & \isdef & \{ \psi \}
\\ \Clos_0(\mu x. \psi) & \isdef & \{ \psi[\mu x. \psi / x] \}
&& \Clos_0(\nu x. \psi) & \isdef & \{ \psi[\nu x. \psi / x] \}
\end{array}\]
If $\psi \in \Clos_{0}(\phi)$ we sometimes write $\phi \tracestep \psi$.
Moreover, we define the \emph{closure} $\Clos(\phi) \subseteq \muML$ of $\phi$
as the least set $\Sigma$ containing $\phi$ that is closed in the sense that 
$\Clos_0(\psi) \subseteq \Sigma$ for all $\psi \in \Sigma$.
We define $\Clos(\Phi) = \bigcup_{\phi \in \Phi} \Clos(\phi)$ for any $\Phi 
\subseteq \muML$.
It is well known that $\Clos(\Phi)$ is finite iff $\Phi$ is finite.

A \emph{trace} is a sequence $(\phi_{n})_{n<\kappa}$, with $\kappa \leq
\omega$, of formulas such that $\phi_{n} \to_{C} \phi_{n+1}$, for all
$n$ such that $n+1 < \kappa$. If $\tau = (\phi_{n})_{n<\kappa}$ is an
infinite trace, then there is a unique formula $\phi$ that occurs
infinitely often on $\tau$ and is a subformula of $\phi_{n}$ for
cofinitely many $n$. This formula is always a fixpoint formula, and
where it is of the form $\phi_{\tau} = \eta x.\psi$ we call $\tau$ an
\emph{$\eta$-trace}.
A proof that there exists a unique such fixpoint formula
$\varphi$ can be found in Proposition~6.4 of \cite{kupk:size20}, but the
observation is well-known in the literature and goes back at least to
\cite{emer:comp88}.
A formula $\phi \in \muML$ is \emph{guarded} if in every subformula
$\eta x . \psi$ of $\phi$ all free occurrences of $x$ in $\psi$
are in the scope of a modality. 
It is well known that every formula can be transformed into an equivalent 
guarded formula, and it is not hard to verify that all formulas in the closure of a 
guarded formula are also guarded.

\paragraph{Semantics}
The semantics of the modal $\mu$-calculus is given in terms of
\emph{Kripke models} $\bbS = (S,R,V)$, where $S$ is a set whose elements
are called \emph{worlds}, \emph{points} or \emph{states}, $R \subseteq S
\times S$ is a binary relation on $S$ called the \emph{accessibility
relation} and $V : \Prop \to \powerset S$ is a function called the
\emph{valuation function}. The \emph{meaning} $\mng{\phi}^\bbS
\subseteq S$ of a formula $\phi \in \muML$ relative to a Kripke model
$\bbS = (S,R,V)$ is defined by induction on the complexity of $\phi$:
\[\begin{array}{lllclll}
   \mng{p}^{\bbS} &\isdef& V(p)
 && \mng{\atneg{p}}^{\bbS} &\isdef& S \setminus V(p)
\\ \mng{\bot}^{\bbS} &\isdef& \nada
 && \mng{\top}^{\bbS} &\isdef& S 
\\ \mng{\phi\lor\psi}^{\bbS} &\isdef& \mng{\phi}^{\bbS} \cup \mng{\psi}^{\bbS} 
 && \mng{\phi\land\psi}^{\bbS} &\isdef& \mng{\phi}^{\bbS} \cap \mng{\psi}^{\bbS}
\\ \mng{\dia\phi}^{\bbS} &\isdef& 
     \{ s \in S \mid R[s] \cap \mng{\phi}^{\bbS} \neq \nada \}
 && \mng{\Box\phi}^{\bbS} &\isdef& 
     \{ s \in S \mid R[s] \sse \mng{\phi}^{\bbS} \}
\\ \mng{\mu x.\phi}^{\bbS} &\isdef& 
     \bigcap \{ U \subseteq S \mid \mng{\phi}^{\bbS[x\mapsto U]}\sse U \}
   &&  \mng{\nu x.\phi}^{\bbS} &\isdef& 
     \bigcup \{ U \subseteq S \mid \mng{\phi}^{\bbS[x\mapsto U]}\supseteq U \}.
\end{array}\]
Here, $\bbS[x \mapsto U]$ for some $U \subseteq S$ denotes the model
$(S,R,V')$, where $V'(x) = u$ and $V'(p) = V(p)$ for all $p \in \Prop$
with $p \neq x$. We say that $\phi$ \emph{is true} at $s$ if $s \in
\mng{\phi}^\bbS$. A formula $\phi \in \muML$ is valid if
$\mng{\phi}^\bbS = S$ holds in all Kripke models $\bbS = (S,R,V)$ and
two formulas $\phi,\psi \in \muML$ are \emph{equivalent} if
$\mng{\phi}^\bbS = \mng{\psi}^\bbS$ for all Kripke models $\bbS$.

Alternatively, the semantics of the $\mu$-calculus is often given in terms of 
a so-called \emph{evaluation} or \emph{model checking game}.
Let $\xi \in \muML$ be a $\mu$-calculus formula, and let $\bbS = (S,R,V)$ be 
a Kripke model.
The \emph{evaluation game} $\EG(\xi,\bbS)$ is the following infinite two-player 
game\footnote{%
   We assume familiarity with such games, see the appendix for some definitions.
   }.
Its positions are pairs of the form $(\phi,s) \in \Clos(\xi)\times S$, and its
ownership function and admissible rules are given in Table~\ref{tb:EG}.
For the winning conditions of this game, consider an infinite match of the form
$\Sigma = (\phi_{n},s_{n})_{n<\omega}$; then we define the winner of the match
to be \Eloise if the induced trace $(\phi_{n})_{n<\omega}$ is a $\nu$-trace,
and \Abelard if it is a $\mu$-trace.
It is well-known that this game can be presented as a parity game, and as such
it has positional determinacy.

\begin{table}[htb]
\begin{center}
\begin{tabular}{|ll|c|l|}
\hline
\multicolumn{2}{|l|}{Position} & Player & Admissible moves\\
\hline
     $(p,s)$        & with $p\in \FV(\xi)$ and $s \in V(p)$         
   & $\abel$ & $\nada$ 
\\   $(p,s)$        & with $p\in \FV(\xi)$ and $s \notin V(p)$      
   & $\eloi$ & $\nada$ 
\\   $(\atneg{p},s)$   & with $p\in \FV(\xi)$ and $s \in V(p)$    
   & $\eloi$ & $\nada$ 
\\   $(\atneg{p},s)$  & with $p\in \FV(\xi)$ and $s \notin V(p)$ 
   & $\abel$ & $\nada$ 
\\ \multicolumn{2}{|l|}{$(\phi \lor \psi,s)$}   & $\eloi$   
   & $\{ (\phi,s), (\psi,s) \}$ 
\\ \multicolumn{2}{|l|}{$(\phi \land \psi,s)$} & $\abel$ 
   & $\{ (\phi,s), (\psi,s) \}$ 
\\ \multicolumn{2}{|l|}{$(\dia \phi,s)$}        & $\eloi$ 
   & $\{ (\phi,t) \mid sRt \}$ 
\\ \multicolumn{2}{|l|}{$(\Box \phi,s) $}       & $\abel$ 
   & $\{ (\phi,t) \mid sRt \}$ 
\\ \multicolumn{2}{|l|}{$(\eta x . \phi,s)$}    & - 
   & $\{ (\phi[\eta x\, \phi/x],s) \}$ 
\\ \hline
\end{tabular}
\end{center}
\caption{The evaluation game $\EG(\xi,\bbS)$}
\label{tb:EG}
\end{table}

\subsubsection{The alternation-free fragment}

As mentioned in the introduction, the alternation-free fragment of the modal 
$\mu$-calculus consists of relatively simple formulas, in which the interaction
between least- and greatest fixpoint operators is restricted.
There are various ways to formalise this intuition.
Following the approach by Niwi\'nski~\cite{niwi:fixp86}, we call a formula 
$\xi$ alternation free if it satisfies the following: if $\xi$ has a subformula 
$\eta x. \phi$ then no free occurrence of $x$ in $\phi$ can be in the scope of 
an $\ol{\eta}$-operator.
An inductive definition of this set can be given as follows.

\begin{definition}
\label{d:afmc}
By a mutual induction we define the \emph{alternation-free $\mu$-calculus}
$\AFMC$, and, for a subset $\PropQ \sse \Prop$ and $\eta \in \{ \mu, \nu \}$,
its \emph{noetherian $\eta$-fragment over $\PropQ$}, $\nth{\eta}{\PropQ}$.
\[\begin{array}{rlc@{\divbnf}c@{\divbnf}c@{\divbnf}c@{\divbnf}%
c@{\divbnf}c@{\divbnf}c@{\divbnf}c@{\divbnf}l@{\divbnf}l@{\divbnf}c}
\AFMC \ni \phi &\isbnf
   & \bot & \top 
   & p  & \ol{p}
   & (\phi_{0} \land \phi_{1}) & (\phi_{0} \lor \phi_{1})
   & \dia \phi & \Box \phi 
   & \mu p. \phi^{\mu}_{p} 
   & \nu p. \phi^{\nu}_{p}
\\[2mm] \nth{\mu}{\PropQ} \ni \phi & \isbnf
   & \bot & \top 
   & q  & 
   & (\phi_{0} \land \phi_{1}) & (\phi_{0} \lor \phi_{1})
   & \dia \phi & \Box \phi 
   & \mu p. \phi^{\mu}_{\PropQ p} &
   & \psi
\\[2mm] \nth{\nu}{\PropQ} \ni \phi &\isbnf&
   \bot & \top 
   & q & 
   & (\phi_{0} \land \phi_{1}) & (\phi_{0} \lor \phi_{1})
   & \dia \phi & \Box \phi 
   && \nu p. \phi^{\nu}_{\PropQ p} 
   & \psi
\end{array}\]
where $p \in \Prop$, $q \in \PropQ$, 
$\phi^{\eta}_{\PropP} \in \nth{\eta}{\PropP}$ for $\PropP \sse \Prop$, and 
$\psi \in \AFMC$ is such that $\FV(\psi) \cap \PropQ = \nada$.
Here and in the sequel we shall write $p$ for $\{ p \}$ and $\PropQ q$ for 
$\PropQ \cup \{ q \}$.
\end{definition}

Throughout the text we shall simply refer to elements of $\AFMC$ as 
\emph{formulas}.

The intuition underlying this definition is that $\nth{\eta}{\PropQ}$ consists
of those alternation-free formulas in which free variables from $\PropQ$ may 
not occur in the scope of an $\ol{\eta}$-operator.
The name `noetherian' refers to a semantic property that characterize the 
$\nth{\mu}{\PropQ}$ formulas~\cite{font:mode18}:
if a formula $\phi\in\nth{\mu}{\PropQ}$ is satisfied at the root of a tree 
model $\bbT$, then it is also true in a variant of $\bbT$ where we restrict 
the interpretation of the proposition letters in $\PropQ$ to noetherian 
subtrees of $\bbT$, i.e., subtrees without infinite paths.

\begin{example}
For some examples of alternation-free formulas, observe that $\AFMC$ contains 
all basic modal (i.e., fixpoint-free) formulas, as well as all $\muML$-formulas
that use $\mu$-operators or $\nu$-operators, but not both, and all modal and
boolean combinations of such formulas.

For a slightly more sophisticated example, consider the formula $\xi = 
\mu x. (\nu y. p \land \Box y) \land \dia x$.
This formula does feature an alternating chain of fixpoint operators, in the
sense that the $\nu$-formula $\phi = \nu y. p \land \Box y$ is a subformula 
of the $\mu$-formula $\xi$.
However, since the variable $x$ does not occur in $\phi$, this 
formula does belong to $\AFMC$.
To see this in terms of Definition~\ref{d:afmc}, observe that $\psi \in 
\nth{\mu}{x}$ since $x \not\in \FV(\psi)$.
But then the formula $(\nu y. p \land \Box y) \land \dia x$ also belongs
to this fragment, and from this it is immediate that $\xi \in \AFMC$.
\end{example}

Below we gather some basic observations on $\AFMC$. 
First we mention some useful closure conditions, stating that $\AFMC$ is closed
under taking respectively negations, unfoldings, subformulas and guarded 
equivalents.

\begin{proposition}
\label{p:af1}
Let $\xi$ be an alternation-free formula. Then
\begin{urlist}
\item \label{it:af1-1}
   its negation $\ol{\xi}$ is alternation free;
\item \label{it:af1-2}
   if $\xi$ is a fixpoint formula, then its unfolding is alternation free;
\item \label{it:af1-3}
   every subformula of $\xi$ is alternation free;
\item \label{it:af1-4}
   every formula in $\Clos(\xi)$ is alternation free;
\item \label{it:af1-5}
   there is an alternation-free guarded formula $\xi'$ that is equivalent to $\xi$.
\end{urlist}
\end{proposition}

\begin{proof}
Item~\ref{it:af1-2} is immediate by Proposition~\ref{p:af3}(\ref{it:af3-3}
and Proposition~\ref{p:af2}(\ref{it:af2-2}.
For item~\ref{it:af1-5} a careful inspection will reveal that the standard 
procedure for guarding formulas (see 
\cite{walu:comp00,kupfer:autobranch00,brus:guar15}) transforms alternation-free
formulas to guarded alternation-free formulas.
The other items can be proved by routine arguments.
\end{proof}

\begin{proposition}
\label{p:af2}
\begin{urlist}
\item \label{it:af2-1}
If $\PropQ$ and $\PropQ'$ are sets of proposition letters with $\PropQ \sse 
\PropQ'$, then $\nth{\eta}{\PropQ'} \sse \nth{\eta}{\PropQ}$.
\item \label{it:af2-2}
$\AFMC = \nth{\eta}{\nada}$.
\end{urlist}
\end{proposition}

\begin{proof}
Item~\ref{it:af2-1} can be proved by a straightforward induction on the 
complexity of formulas in $\nth{\eta}{\PropQ'}$; we leave the details for the
reader.
A similar induction shows that $\nth{\eta}{\PropQ} \sse \AFMC$, for any set 
$\PropQ$ of variables; clearly this takes care of the inclusion $\sse$ in
item~\ref{it:af2-2}.

This leaves the statement that $\AFMC \sse \nth{\eta}{\nada}$, which we prove 
by induction on the complexity of $\AFMC$-formulas.
We confine our attention here to the case where $\phi \in \AFMC$ is a fixpoint
formula, say, $\phi = \lambda p. \phi'$.
But then it is obvious that $\FV(\phi) \cap \{ p \} = \nada$, so that $\phi \in 
\nth{\eta}{p}$ by definition of the latter set.
It follows that $\phi \in \nth{\eta}{\nada}$ by item~\ref{it:af2-1}.
\end{proof}

The following proposition states some useful closure conditions on sets of the 
form $\nth{\eta}{\PropQ}$.

\begin{proposition}
\label{p:af3}
Let $\chi$ and $\xi$ be formulas in $\AFMC$, let $x,y$ be variables, and let
$\PropQ$ be a set of variables.
Then the following hold:
\begin{urlist}
\item \label{it:af3-1}
if $\xi \in \nth{\eta}{\PropQ}$ and $y \not\in \FV(\xi)$, then $\xi \in 
\nth{\eta}{\PropQ y}$;
\item \label{it:af3-2}
if $\chi \in \nth{\eta}{\PropQ x}$, $\xi \in \nth{\eta}{\PropQ}$ and $\xi$ is
free for $x$ in $\chi$, then $\chi[\xi/x] \in \nth{\eta}{\PropQ}$;
\item \label{it:af3-3}
if $\eta x\, \chi \in \nth{\eta}{\PropQ}$ then $\chi[\eta x\, \chi/x] \in
\nth{\eta}{\PropQ}$.
\end{urlist}
\end{proposition}

\begin{proof}
We prove item~\ref{it:af3-1} of the proposition by a straightforward induction 
on the complexity of $\xi$.
We only cover the case of the induction step where $\xi$ is of the form $\xi 
= \lambda z. \xi'$.
Here we distinguish cases.
If $\FV(\xi) \cap \PropQ = \nada$ then we find $\FV(\xi) \cap (\PropQ \cup 
\{ y \}) = \nada$ since $y \not\in \FV(\xi)$ by assumption.
Here it is immediate by the definition of $\nth{\eta}{\PropQ y}$ that $\xi$
belongs to it.

If, on the other hand, we have $\FV(\xi) \cap \PropQ \neq \nada$, then we can 
only have $\xi \in \nth{\eta}{\PropQ}$ if $\lambda = \eta$.
We now make a further case distinction: if $y = z$ then we have $\xi' \in 
\nth{\eta}{\PropQ y}$ so that also $\xi \in \nth{\eta}{\PropQ y}$.
If $y$ and $z$ are distinct variables, then it must be the case that 
$\xi' \in \nth{\eta}{\PropQ z}$; since we clearly have $y \not\in \FV(\xi')$ as
well, the inductive hypothesis yields that $\xi' \in 
\nth{\eta}{\PropQ yz}$.
But then we immediately find $\xi \in \nth{\eta}{\PropQ y}$ by definition of the
latter set.
\smallskip

For the proof of item~\ref{it:af3-2} we proceed by induction on the complexity 
of $\chi$. 
Again, we only cover the inductive case where $\chi$ is a fixpoint formula, say, 
$\chi = \lambda y. \chi'$.
We make a case distinction.
First assume that $x \not\in \FV(\chi)$; then we find $\chi[\xi/x] = \chi$, so
that $\chi[\xi/x] \in \nth{\eta}{\PropQ x}$ by assumption.
It then follows that $\chi[\xi/x] \in \nth{\eta}{\PropQ}$ by 
Proposition~\ref{p:af2}(\ref{it:af2-1}.

Assume, then, that $x \in \FV(\chi)$; since $\chi \in \nth{\eta}{\PropQ x}$ this 
can only be the case if $\lambda = \eta$, and, again by definition of 
$\nth{\eta}{\PropQ x}$, we find $\chi' \in \nth{\eta}{\PropQ xy}$.
Furthermore, as $\xi$ is free for $x$ in $\chi$, the variable $y$ cannot be 
free in $\xi$, so that it follows by item~\ref{it:af3-1} and the assumption 
that $\xi \in \nth{\eta}{\PropQ}$ that $\xi \in \nth{\eta}{\PropQ y}$.
We may now use the inductive hypothesis on $\chi'$ and $\xi$, to find that 
$\chi'[\xi/x] \in \nth{\eta}{\PropQ y}$; and from this we conclude that 
$\chi[\xi/x] \in \nth{\eta}{\PropQ}$ by definition of $\nth{\eta}{\PropQ}$.
\smallskip

Finally, item~\ref{it:af3-3} is immediate by item~\ref{it:af3-2}.
\end{proof}

The next observation can be used to simplify the formulation of the winning
conditions of the evaluation game for alternation-free formulas somewhat. 
It is a direct consequence of results in~\cite{kupk:size20}, so we confine 
ourselves to a proof sketch.
\begin{proposition}
\label{p:af4}
For any infinite trace $\tau = (\phi_{n})_{n<\omega}$ of $\AFMC$-formulas the
following are equivalent:
\begin{urlist}
\item \label{it:af4-1}
   $\tau$ is an $\eta$-trace;
\item \label{it:af4-2}
   $\phi_{n}$ is an $\eta$-formula, for infinitely many $n$;
\item \label{it:af4-3}
   $\phi_{n}$ is an $\ol{\eta}$-formula, for at most finitely many $n$.
\end{urlist}
\end{proposition}

\begin{proof}[Proof (sketch)]
Let $\xi = \eta z. \xi'$ be the characteristic fixpoint formula of $\tau$, i.e., 
$\xi$ is the unique formulas that occurs infinitely often on $\tau$ and that
is a subformula of almost all formulas on $\tau$.
Clearly it suffices to prove that almost every fixpoint formula on $\tau$ is an 
$\eta$-formula as well.

To show why this is the case, it will be convenient to introduce the following
notation.
We write $\psi \rclat{\rho} \phi$ if there is a sequence 
$(\chi_{i})_{0\leq i \leq n}$ such that $\psi = \chi_{0}$, $\phi = \chi_{n}$,
$\chi_{i} \to_{C} \chi_{i+1}$ for all $i<n$, and every $\chi_{i}$ is of the form
$\chi_{i}'[\rho/x]$ for some formula $\chi_{i}'$ and some $x \in \FV(\chi_{i}')$.
Then it readily follows from the definitions that $\xi \rclat{\xi} \phi_{n}$ for 
almost every formula $\phi_{n}$ on $\tau$.
The key observation in the proof is now that if $\xi$ is alternation-free,
and $\phi$ is a fixpoint formula such that $\xi \rclat{\xi} \phi$, then
$\phi$ is an $\eta$-formula.
%
%
%
%
To be more precise we first show that
\begin{equation} \label{eq:ih}
 \mbox{for all } \phi \mbox{ with } \xi \rclat{\xi} \phi \mbox{
there is some } \phi^\circ \in \nth{\eta}{z} \mbox{ such that } z \in
\FV(\phi^\circ) \mbox{ and } \phi = \phi^\circ [\xi/z].
\end{equation}
We prove this claim by induction on the length of the path $\xi \rclat{\xi}
\phi$. 
In the base case we have $\phi = \xi$ and we let $\phi^\circ = z$.

In the inductive step there is some $\chi$ such that $\xi \rclat{\xi} \chi 
\rcla{\xi} \phi$. 
By the inductive hypothesis there is some $\chi^\circ \in \nth{\eta}{z}$ such
that $z \in \FV(\chi^\circ)$ and $\chi = \chi^\circ [\xi/z]$. 
We distinguish cases depending on the main connective of $\chi$.
Omitting the boolean and modal cases we focus on the case where $\chi$ is a 
fixpoint formula, and we further distinguish cases depending on whether $\chi 
= \xi$ or not.

If $\chi = \xi$ then $\phi = \xi'[\xi/z]$. 
Because $\xi$ is alternation free we know that $\xi' \in \nth{\eta}{x}$. 
We can thus let $\phi^{\circ} \isdef \xi[z/x]$.

If $\chi = \lambda y . \chi'$ but $\chi \neq \xi$ then we have $\phi = 
\chi'[\chi/y]$. 
From the inductive hypothesis we get that
$\chi = \chi^\circ[\xi / z]$ for some $\chi^\circ \in \nth{\eta}{z}$
with $z \in \FV{\chi^\circ}$. Because $\chi \neq \xi$ it follows from
$\chi = \lambda y . \chi'$ and $\chi = \chi^\circ[\xi / z]$ that
$\chi^\circ = \lambda y . \rho$ for some $\rho$ with $\chi' = \rho[\xi
/ z]$. Hence, $\phi = \rho[\xi / z][\chi/y]$. 
Because $z \notin \FV(\chi)$ and $y \notin \FV(\xi)$ (because $\BV(\chi) \cap 
\FV(\xi) = \nada$) we may commute these substitutions
(cf.~Proposition~3.11 in \cite{kupk:size20}).
Hence $\phi = \rho[\chi/y][\xi / z]$, and we may set $\phi^\circ \isdef 
\rho[\chi/y]$. 
Because $\chi^\circ = \lambda y . \rho$ and $z \in \FV(\chi^\circ)$ it follows
that $z \neq y$ and that $z \in \FV(\rho)$. Thus also $z \in
\FV(\rho[\chi/y])$. Lastly, it follows from $\chi^\circ = \lambda y .
\rho$, $\chi^\circ \in \nth{\eta}{z}$ and $z \in \FV(\rho)$ that $\rho
\in \nth{\eta}{z}$. It is not hard to see that $\nth{\eta}{z}$ is closed
under substitution with the alternation free formula $\chi$, where $z
\notin \FV(\chi)$ and thus $\rho[\chi/y] \in \nth{\eta}{z}$.
This finishes the proof of \eqref{eq:ih}.

The claim about fixpoint formulas $\phi$ such that $\xi \rclat{\xi} \phi$
can be derived from \eqref{eq:ih} as follows.
Assume that $\phi$ is of the form $\phi = \lambda y . \rho$, then if 
$\lambda y . \rho = \phi^\circ[\xi/z]$ with $z \in \FV(\phi^\circ)$ and 
$\phi \neq \xi$ then it must be the case that $\phi^\circ =
\lambda y . \rho^\circ$, and because $\phi^\circ \in \nth{\eta}{z}$
and $z \in \FV(\rho^\circ)$ this is only possible if $\lambda = \eta$.
That is, $\phi$ is an $\eta$-formula as required.
\end{proof}

\section{The focus system}
\label{sec-proofsystem}

In this section we introduce our annotated proof systems for the
alternation-free $\mu$-calculus. 
We consider two versions of the system, which we call \Focus and \Focusinf,
respectively.
\Focusinf is a proof system that allows proofs to be based on infinite, but 
finitely branching trees.
The focus mechanism
that is implemented by the annotations of formulas helps ensuring that all
the infinite branches in a \Focusinf proof are of the right shape.
The proof system \Focus can be seen as a finite variant of \Focusinf. 
The proof trees in this system are finite, but the system is circular in that
it contains a discharge rule that allows to discharge a leaf of the tree if 
the same sequent as the sequent at the leaf is reached again closer to the 
root of the tree. 
As we will see, the two systems are equivalent in the sense that we may
transform proofs in either variant into proofs of the other kind.

\subsection{Basic notions}

In this first part of this section we provide the definition of the
proof systems \Focus and \Focusinf.

A \emph{sequent} is a finite set of formulas.
When writing sequents we often leave out the braces, meaning that we write
for instance $\phi_1,\dots,\phi_i$ for the sequent $\{\phi_1,\dots,\phi_i\}$. 
If $\Phi$ is a sequent, we also use the notation $\phi_1,\dots,\phi_i, \Phi$ 
for the sequent $\{\phi_1,\dots,\phi_i\} \cup \Phi$. 
Given a sequent $\Phi$ we write $\dia \Phi$ for the sequent $\dia \Phi \isdef \{\dia
\phi \mid \phi \in \Phi\}$.
Intuitively, sequents are to be read \emph{disjunctively}.

An \emph{annotated formula} is a pair $(\phi,a) \in \AFMC \times \{ f,u
\}$; we usually write $\phi^{a}$ instead of $(\phi,a)$ and call $a$ the
\emph{annotation} of $\phi$. 
We define a linear order $\sqsubseteq$ on the set $\{f,u\}$ of annotations by
putting $u \sqsubset f$, and given $a \in \{ f,u \}$ we let $\ol{a}$ be its 
alternative, i.e., we define $\ol{u} \isdef f$ and $\ol{f} \isdef u$.
A formula that is annotated with $f$ is called \emph{in focus}, and one 
annotated with $u$ is \emph{out of focus}.
We use $a,b,c,\ldots$ as symbols to range over the set $\{f,u\}$.

A finite set of annotated formulas is called an \emph{annotated sequent}.
We shall use the letters $\Si, \Gamma, \Delta, \ldots$ for annotated sequents,
and $\Phi,\Psi$ for sequents.
In practice we will often be sloppy and refer to annotated sequents as 
sequents.
Given a sequent $\Phi$, we define $\Phi^a$ to be the annotated sequent $\Phi^a
\isdef \{\phi^a \mid \phi \in \Phi \}$.
Conversely, given an annotated sequent $\Si$, we define $\uls{\Si}$ as its
underlying plain sequent; that is, $\uls{\Si}$ consists of the formulas
$\phi$ such that $\phi^{a} \in \Si$, for some annotation $a$.

The proof rules of our focus proof systems $\Focus$ and $\Focusinf$ are given 
in Figure~\ref{f:proof rules}.
We use standard terminology when talking about proof rules.
Every (application of a) rule has one \emph{conclusion} and a finite 
(possibly zero) number of \emph{premises}.
\emph{Axioms} are rules without premises. 
The \emph{principal} formula of a rule application is the formula in the 
conclusion to which the rule is applied. 
As non-obvious cases we have that all formulas are principal in the conclusion 
of the rule \RuBox and that the rule \RuDischarge{\dx} has no principal 
formula.
In all cases other than for the rule \RuWeak the principal formula develops
into one or more \emph{residual} formulas in each of the premises. 
Principal and residual formulas are also called \emph{active}.

\begin{figure}[tbh]
\begin{minipage}{\textwidth}
\begin{minipage}{0.16\textwidth}
\begin{prooftree}
 \AxiomC{\phantom{X}}
 \RightLabel{\AxLit}
 \UnaryInfC{$p^a, \atneg{p}^b$}
\end{prooftree}
\end{minipage}
\begin{minipage}{0.12\textwidth}
\begin{prooftree}
 \AxiomC{\phantom{X}}
 \RightLabel{\AxTop}
 \UnaryInfC{$\top^a$}
\end{prooftree}
\end{minipage}
\begin{minipage}{0.21\textwidth}
\begin{prooftree}
 \AxiomC{$\phi^a,\psi^a,\Sigma$}
 \RightLabel{\RuOr}
 \UnaryInfC{$(\phi \lor \psi)^a,\Sigma$}
\end{prooftree}
\end{minipage}
\begin{minipage}{0.28\textwidth}
\begin{prooftree}
 \AxiomC{$\phi^a, \Sigma$}
 \AxiomC{$\psi^a,\Sigma$}
 \RightLabel{\RuAnd}
 \BinaryInfC{$(\phi \land \psi)^a,\Sigma$}
\end{prooftree}
\end{minipage}
\begin{minipage}{0.20\textwidth}
\begin{prooftree}
 \AxiomC{$\phi^a,\Sigma$}
 \RightLabel{\RuBox}
 \UnaryInfC{$\Box \phi^a, \dia \Sigma$}
\end{prooftree}
\end{minipage}
\end{minipage}

\bigskip

\begin{minipage}{\textwidth}
\begin{minipage}{0.24\textwidth}
\begin{prooftree}
 \AxiomC{$\phi[\mu x . \phi / x]^u, \Sigma$}
 \RightLabel{\RuMu}
 \UnaryInfC{$\mu x . \phi^a, \Sigma$}
\end{prooftree}
\end{minipage}
\begin{minipage}{0.24\textwidth}
\begin{prooftree}
 \AxiomC{$\phi[\nu x . \phi / x]^a, \Sigma$}
 \RightLabel{\RuNu}
 \UnaryInfC{$\nu x . \phi^a, \Sigma$}
\end{prooftree}
\end{minipage}
\begin{minipage}{0.16\textwidth}
\begin{prooftree}
 \AxiomC{$\Sigma$}
 \RightLabel{\RuWeak}
 \UnaryInfC{$\phi^a, \Sigma$}
\end{prooftree}
\end{minipage}
\begin{minipage}{0.16\textwidth}
\begin{prooftree}
 \AxiomC{$\phi^f,\Sigma$}
 \RightLabel{\RuFocus}
 \UnaryInfC{$\phi^u,\Sigma$}
\end{prooftree}
\end{minipage}
\begin{minipage}{0.16\textwidth}
\begin{prooftree}
 \AxiomC{$\phi^u,\Sigma$}
 \RightLabel{\RuUnfocus}
 \UnaryInfC{$\phi^f,\Sigma$}
\end{prooftree}
\end{minipage}
\end{minipage}

\begin{prooftree}
 \AxiomC{$[\Sigma]^\dx$}
 \noLine
 \UnaryInfC{$\vdots$}
 \noLine
 \UnaryInfC{$\Sigma$}
 \RightLabel{\RuDischarge{\dx}}
 \UnaryInfC{$\Sigma$}
\end{prooftree}
\caption{Proof rules of the focus system}
\label{f:proof rules}
\end{figure}

Here are some more specific comments about the individual proof rules.
The boolean rules ($\RuAnd$ and $\RuOr$) are fairly standard; observe
that the annotation of the active formula is simply inherited by its
subformulas. The fixpoint rules (\RuMu and \RuNu) simply unfold the
fixpoint formulas; note, however, the difference between \RuMu and \RuNu
when it comes to the annotations: in \RuNu the annotation of the active
$\nu$-formula remains the same under unfolding, while in \RuMu, the
active $\mu$-formula \emph{loses focus} when it gets unfolded.
The box rule \RuBox is the standard modal rule in one-sided sequent systems;
the annotation of any formula in the consequent and its derived formula in the
antecedent are the same.

The rule \RuWeak is a standard \emph{weakening rule}. Next to \RuMu, the
\emph{focus rules} \RuFocus and \RuUnfocus are the only rules that change the
annotations of formulas. 
Finally, the \emph{discharge rule} \RuDischarge{} is a special proof rule
that allows us to discharge an assumption if it is repeating a sequent that 
occurs further down in the proof. 
Every application \RuDischarge{\dx} of this rule is marked by a so-called 
\emph{discharge token} $\dx$ that is taken from some fixed infinite set 
$\Tokens = \{\dx,\dy,\dz,\dots\}$.
In Figure~\ref{f:proof rules} this is suggested by the notation
$[\Sigma]^\dx$.
The precise conditions under which \RuDischarge{\dx} can be employed are 
explained in Definition~\ref{d:proof} below.

\begin{definition} 
\label{d:proof}
A \emph{pre-proof} $\Pi = (T,P,\Si,\pLab)$ is a quadruple such that
$(T,P)$ is a, possibly infinite, tree with nodes $T$ and parent relation $P$;
$\Si$ is a function that maps every node $u \in T$ to a non-empty annotated 
sequent $\Sigma_u$;
and 
\[
\pLab:\; T \;\to\;
\big\{
\AxLit,\AxTop,\RuOr,\RuAnd,\RuBox,\RuMu,\RuNu,\RuWeak,\RuFocus,\RuUnfocus \big\} 
\cup
\big\{\RuDischarge{\dx} \mid \dx \in \Tokens \big\} \cup \Tokens \cup
\{ \star \},
\]
is a map that assigns to every node $u$ of $T$ its \emph{label} $\pLab(u)$, 
which is either (i) the name of a proof rule, (ii) a discharge token or 
(iii) the symbol $\star$.

To qualify as a pre-proof, such a quadruple is required to satisfy the following
conditions:
\begin{enumerate}
\item \label{i:local condition}
If a node is labelled with the name of a proof rule then it has as many children 
as the proof rule has premises, and the annotated sequents at the node and its 
children match the specification of the proof rules in Figure~\ref{f:proof 
rules}.

%

\item \label{i:leaf condition}
If a node is labelled with a discharge token or with $\star$ then it is a leaf.
We call such nodes \emph{non-axiomatic leaves} as opposed to the \emph{axiomatic 
leaves} that are labelled with one of the axioms, \AxLit or \AxTop.

\item \label{i:discharge condition}
For every leaf $l$ that is labelled with a discharge token $\dx \in \Tokens$ 
there is exactly one node $u$ in $\Pi$ that is labelled with \RuDischarge{\dx}.
This node $u$, as well as its (unique) child, is a proper ancestor of $l$
and satisfies $\Sigma_u = \Sigma_l$. 
In this situation we call $l$ a \emph{discharged leaf}, and $u$ its
\emph{companion}; we write $c$ for the function that maps a discharged
leaf $l$ to its companion $c(l)$.

\item \label{i:path condition} \label{i:pc}
If $l$ is a discharged leaf with companion $c(l)$ then the path from $c(l)$ to
$l$ contains (\ref{i:pc}a) no application of the focus rules, (\ref{i:pc}b) at
least one application of \RuBox, while (\ref{i:pc}c) every node on this path 
features a formula in focus.
\end{enumerate}

Non-axiomatic leaves that are not discharged, are called \emph{open};
the sequent at an open leaf is an \emph{open assumption} of the
pre-proof. We call a pre-proof a \emph{proof in \Focus} if it is finite
and does not have any open assumptions.

A infinite branch $\beta = (v_{n})_{n\in\om}$ is \emph{successful}
if there are infinitely many applications of \RuBox on $\beta$ and there
is some $i$ such that for all $j \geq i$ the annotated sequent at $v_j$ 
contains at least one formula that is in focus and none of the focus
rules \RuFocus and \RuUnfocus is applied at $v_j$. 
A pre-proof is a \emph{\Focusinf-proof} if it does not have any non-axiomatic
leaves and all its infinite branches are successful.

An unannotated sequent $\Phi$ is \emph{derivable} in $\Focusinf$ (in $\Focus$)
if there is a \Focusinf proof (a $\Focus$ proof, respectively) such that $\Phi^f$
is the annotated sequent at the root of the proof.
\end{definition}


For future reference we make some first observations about (pre-)proofs in 
this system.

\begin{proposition} \label{p:proof in closure}
Let $\Phi$ be the set of formulas that occur in the annotated sequent
$\Sigma_r$ at the root of some pre-proof $\Pi = (T,P,\Si,\pLab)$. Then
all formulas that occur annotated in $\Sigma_t$ for any $t \in T$ are in
$\Clos(\Phi)$.
\end{proposition}
\begin{proof}
 This is an easy induction on the depth of $t$ in the tree $(T,P)$. It
amounts to checking that if the formulas in the conclusion of any of the
rules from Figure~\ref{f:proof rules} are in $\Clos(\Phi)$ then so are
the formulas at any of the premises.
\end{proof}

\begin{proposition}
\label{p:lr1}
Let $u$ and $v$ be two nodes in a proof $\Pi = (T,P,\Si,\pLab)$ such that 
$Puv$ and $\Ru_{u} \neq \RuFocus$.
Then the following holds:
\begin{equation}
\label{eq:lr2}
\text{if } \Si_{v} \text{ contains a formula in focus, then so does }
\Si_{u}.
\end{equation}
\end{proposition}
This claim is proved by straightforward inspection in a case distinction as 
to the proof rule $\Ru_{u}$.

\subsection{Circular and infinite proofs}

We first show that \Focusinf and \Focus are the infinitary and 
circular version of the same proof system, and derive the same annotated 
sequents.

\begin{theorem}
\label{t:same}
An annotated sequent is provable in $\Focus$ iff it is provable
in \Focusinf.
\end{theorem}

The two directions of this theorem are proved in Propositions
\ref{p:fintoinf}~and~\ref{p:ppp}.

\begin{proposition}
\label{p:fintoinf}
 If an annotated sequent $\Gamma$ is provable in $\Focus$ then it is
provable in \Focusinf.
\end{proposition}

\begin{proof}
Let $\Pi = (T,P,\Si,\pLab)$ be a proof of $\Gamma$ in \Focus. 
We define a proof $\Pi' = (T',P',\Si',\pLab')$ of $\Gamma$ in \Focusinf. 
Basically, the idea is to unravel the proof $\Pi$ at discharged leaves; the 
result of this, however, would contain some redundant nodes, corresponding 
to the discharged leaves in $\Pi$ and their companions.
In our construction we will take care to remove these nodes from the paths that 
provide the nodes of the unravelled proof.

Going into the technicalities, we first define the relation $L$ on $T$ such 
that $L u v$ holds iff either $P u v$ or $u$ is a discharged leaf and $v = 
c(u)$. 
Let $A$ be the set of all finite paths $\pi$ in $L$ that start at the root $r$
of $(T,P)$.
Formally, $\pi = v_0,\cdots,v_n$ is in $A$ iff $v_0 = r$ and $Lv_iv_{i+1}$ for
all $i \in \{0,\dots,n-1\}$. 
For any path $\pi = v_0,\cdots,v_n \in A$ define $\last(\pi) = v_n$.

Consider the set $S \isdef \Tokens \cup \{\RuDischarge{\dx} \mid \dx \in
\Tokens\}$; 
these are the ones that we need to get rid of in $\Pi'$. 
We then define $T' = \{\pi \in A \mid \pLab(\last(\pi)) \notin
S\}$ and set $P' \pi \rho$ for $\pi,\rho \in T'$ iff $\rho = \pi \cdot
u_1 \cdots u_n$ with $n \geq 1$ and $u_i \in S$ for all $i
\in \{1,\dots,n-1\}$. 
Moreover, we set $\Si'_\pi = \Si_{\last(\pi)}$
and $\pLab'(\pi) = \pLab(\last(\pi))$.

Note that for every node $v \in T$ we can define a unique $L$-path $\pi^v =
t^v_0 \cdots t^v_n$ with $t^v_0 = v$, $\pLab(t^v_n) \notin S$ and 
$\pLab(t^v_i) \in S$ for all $i \in \{0,\dots,n-1\}$.
This path is unique because every node $w$ with $\pLab(w) \in S$ has a unique
$L$-successor, and there cannot be an infinite $L$-path through $S$.
(To see this assume for contradiction that there would be such an infinite 
$L$-path $(t^{v}_{n})_{n\in\om}$ through $S$.
Because $T$ is finite it would follow that from some moment on the path visits 
only nodes that it visits infinitely often. 
Hence, there must then be some discharged leaf such that the infinite path 
visits all the nodes that are between mentioned leaf and its companion. 
But then by definition
the path passes a node $w$ with $\pLab(w) = \RuBox \notin S$.)
Finally, observe that by the definition of the rules in $S$ we have $\Sigma_v 
= \Sigma_{t^v_n}$ for every such path $\pi^v$.

It is not hard to see that $\Pi'$ is a pre-proof, and that it does not use the
detachment rule.
It thus remains to verify that all infinite branches are successful. 
Let  $\beta = (\pi_{n})_{n\in\om}$ be such a branch; by construction we may
associate with $\beta$ a unique $L$-path $\al = (v_{n})_{n\in\om}$ such that
the sequence $(\last(\beta_n))_{n\in\om}$ corresponds to the subsequence we
obtain from $\al$ by removing all nodes from $S$.
Because $T$ is finite, from some point on $\alpha$ only passes nodes that are 
situated on a path to some discharged leaf from its companion node. 
By condition~\ref{i:pc} from Definition~\ref{d:proof} it then follows that 
$\be$ must be successful.
\end{proof}

The converse direction of Theorem~\ref{t:same} requires some preparations.

%

\begin{definition}
A node $u$ in a pre-proof $\Pi = (T,P,\Si,\pLab)$ is called a \emph{successful 
repeat} if it has a proper ancestor $t$ such that $\Si_{t} = \Si_{u}$, $\pLab(t)
\neq \RuDischarge{}$, and the path $[t,u]$ in $\Pi$ satisfies 
condition~\ref{i:pc} of Definition~\ref{d:proof}.
Any node $t$ with this property is called a \emph{witness} to the 
successful-repeat status of $u$.
\end{definition}

The following is then obvious.
\begin{proposition} \label{p:successful repeat}
 Every successful branch $\beta = \beta_0 \beta_1 \cdots$ in a \Focusinf-proof
$\Pi = (T,P,\Si,\pLab)$ contains a successful repeat.
\end{proposition}

\begin{proposition}
\label{p:ppp}
If an annotated sequent $\Ga$ is provable in \Focusinf then it is
provable in \Focus.
\end{proposition}
\begin{proof}
Assume that $\Pi = (T,P,\Si,\pLab)$ is a proof for the annotated sequent
$\Ga$ in \Focusinf.
If $\Pi$ is finite we are done, so assume otherwise; then by K\"{o}nig's Lemma
the set $B^{\infty}$ of infinite branches of $\Pi$ is nonempty.

Because of Proposition~\ref{p:successful repeat} we may define for every
infinite branch $\tau \in B^{\infty}$ the number $\fsr(\tau) \in \omega$
as the least number $n \in \omega$ such that $\tau(n)$ is a successful
repeat. This means that $\tau(\fsr(\tau))$ is the first successful
repeat on $\tau$. Our first claim is the following: 
\begin{equation}
\label{eq:itp01}
\text{there is no pair } \si,\tau \text{ of infinite branches such that }
\si(\fsr(\si)) \text{ is a proper ancestor of } \tau(\fsr(\tau)).
\end{equation}
To see this, suppose for contradiction that $\si(\fsr(\si))$ is a proper ancestor
of $\tau(\fsr(\tau))$, then $\si(\fsr(\si))$ actually lies on the branch $\tau$.
But this would mean that $\si(\fsr(\si))$ is a successful repeat on $\tau$, 
contradicting the fact that $\tau(\fsr(\tau))$ is the \emph{first} successful 
repeat on $\tau$.

Our second claim is that 
\begin{equation}
\label{eq:itp02}
\text{ the set } Y \isdef \{ t \in T \mid t \text{ has a descendant } 
   \tau(\fsr(\tau)), 
   \text{ for some } \tau \in B^{\infty} \} 
\text{ is finite}.
\end{equation}
For a proof of \eqref{eq:itp02}, assume for contradiction that $Y$ is infinite.
Observe that $Y$ is in fact (the carrier of) a subtree of $(T,P)$, and as such 
a finitely branching tree.
It thus follows by K\"onig's Lemma that $Y$ has an infinite branch $\si$, which
is then clearly also an infinite branch of $\Pi$.
Consider the node $s \isdef \si(\fsr(\si))$.
Since $\si$ is infinite, it passes through some proper descendant $t$ of $s$.
This node $t$, lying on $\si$, then belongs to the set $Y$, so that by 
definition it has a descendant of the form $\tau(\fsr(\tau))$ for some $\tau 
\in B^{\infty}$.
But then $\si(\fsr(\si))$ is a proper ancestor of $\tau(\fsr(\tau))$, 
which contradicts our earlier claim \eqref{eq:itp01}.
It follows that the set $Y$ is finite indeed.

Note that it obviously follows from \eqref{eq:itp02} that the set
\[
\wh{Y} \isdef \{ \tau(\fsr(\tau)) \mid \tau \in B^{\infty} \}
\]
is finite as well. Recall that every element $l \in \wh{Y}$ is a
successful repeat; we may thus define a map $c: \wh{Y} \to T$ by setting
$c(l)$ to be the \emph{first} ancestor $t$ of $l$ witnessing that $l$ is
a successful repeat. Finally, let $\Ran(c)$ denote the range of $c$.
\medskip

We are almost ready for the definition of the finite tree $(T',P')$ that will 
support the proof $\Pi'$ of $\Ga$; the only thing left to care of is the 
well-founded part of $\Pi'$.
For this we first define $Z$ to consist of those successors of nodes in $Y$ 
that generate a finite subtree;
then it is easy to show that the collection $P^{*}[Z]$ of descendants of nodes
in $Z$ is finite.

With the above definitions we have all the material in hands to define a
\Focus-proof $\Pi' = (T',P',\Si',\pLab')$ of $\Ga$.
The basic idea is that $\Pi'$ will be based on the set $Y \cup P^{*}[Z]$, with
the nodes in $\wh{Y}$ providing the discharged assumptions of $\Pi'$.
Note however, that for a correct presentation of the discharge rule, every
companion node $u$ of such a leaf in $\wh{Y}$ needs to be provided with a 
successor $u^{+}$ that is labelled with the same annotated sequent as the 
companion node and the leaf.

First of all we set
\[
T' \isdef Y \cup P^{*}[Z] \cup \{ u^{+} \mid u \in \Ran(c) \}
\]
and 
\begin{align*}
P' \isdef &  \quad \{ (u,v) \in P \mid 
   u \in T' \setminus \Ran(c) \text{ and } v \in T' \}
\\ & \cup \{ (u,u^{+}) \mid u \in \Ran(c) \}
\\ & \cup \{ (u^{+},v) \mid u \in \Ran(c), (u,v) \in P \}
\end{align*}
The point of adding the nodes $u^+$ is to make space for applications of
the rule \RuDischarge{\dx} at companion nodes.
Furthermore, we put
\[
\Si'(u) \isdef \left\{ \begin{array}{ll}
   \Si(u) & \text{if } u \in T'
\\ \Si(t) & \text{if } u = t^{+} \text{ for some } t \in \Ran(c).
\end{array} \right.
\]
Finally, for the definition of the rule labelling $\pLab'$, we introduce a set $A
\isdef \{ \dx_{u} \mid u \in \Ran(c) \}$ of discharge tokens, and we define
\[
\pLab'(u) \isdef \left\{ \begin{array}{lll}
   \pLab(u)       & \text{if } & u \in T' \setminus (\wh{Y} \cup \Ran(c))
\\ \dx_{c(l)} & \text{if } & u = l \in \wh{Y},
\\ \RuDischarge{\dx_{u}}   & \text{if } & u \in \Ran(c)
\\ \pLab(t)       & \text{if } & 
      u = t^{+} \text{ for some } t \in \Ran(c).
\end{array} \right.
\]
It is straightforward to verify that with this definition, $\Pi'$ is indeed
a \Focus-proof of the sequent $\Ga$.
\end{proof}

\subsection{Thin and progressive proofs}

When we prove the soundness of our proof system it will be convenient to work
with (infinite) proofs that are in a certain normal form.
The idea here is that we restrict (as much as possible) attention to sequents
that are \emph{thin} in the sense that they do not feature formulas that are
both in and out of focus, and to proofs that are \emph{progressive} in the sense
that when (from the perspective of proof search) we move from the conclusion of
a boolean or fixpoint rule to its premise(s), we drop the principal formula.
Theorem~\ref{t:tpp} below states that we can make these assumptions without loss
of generality.

\begin{definition}
 An annotated sequent $\Sigma$ is \emph{thin} if there is no formula
$\phi \in \AFMC$ such that $\phi^f \in \Sigma$ and $\phi^u \in
\Sigma$.
Given an annotated sequent $\Si$, we define its \emph{thinning}
\[
\thin{\Si} \isdef \{ \phi^{f} \mid \phi^{f} \in \Si \} \cup \{ \phi^{u}
\mid \phi^{u} \in \Si, \phi^{f} \not\in \Si \}.
\]
A pre-proof $\Pi = (T,P,\Si,\pLab)$ is \emph{thin} if for all $v \in T$
with $\phi^f,\phi^u \in \Si_v$ we have that $\pLab_v =
\RuWeak$ and $\phi^u \notin \Si_u$ for the unique $u$ with $P v
u$.
\end{definition}

Note that one may obtain the thinning $\thin{\Si}$ from an annotated sequent
$\Si$ by removing the \emph{unfocused} versions of the formulas with a double
occurrence in $\Si$.

The definition of a thin proof implies that whenever a thin
proof contains a sequent that is not thin then this sequent is followed
by applications of the weakening rule until all the duplicate formulas
are weakened away. For example if the sequent $\Sigma_v =
p^u,p^f,q^u,q^f,r$ occurs in a thin proof then at $v$ and all of its
immediate successors there need to be applications of weakening until
only one annotated version of $p$ and one annotated version of $q$ is
left. This might look for instance as follows:
\begin{center}
\begin{prooftree}
\AxiomC{$\vdots$}
\noLine
\UnaryInfC{$p^f,q^u,r$}
\RightLabel{\RuWeak}
\UnaryInfC{$p^f,q^u,q^f,r$}
\RightLabel{\RuWeak}
\UnaryInfC{$p^u,p^f,q^u,q^f,r$}
\noLine
\UnaryInfC{$\vdots$}
\end{prooftree}
\end{center}


\begin{definition}
An application of a boolean or fixpoint rule  at a node $u$ in a pre-proof
$\Pi = (T,P,\Si,\pLab)$ is \emph{progressive} if for the principal formula 
$\phi^a \in \Si_u$ it holds that $\phi^a \notin \Si_v$ for all $v$ with 
$Puv$.\footnote{%
   Note that since we assume guardedness, the principal formula is different
   from its residuals.
   } 
The proof $\Pi$ is \emph{progressive} if all applications of the boolean rules
and the fixpoint rules in $\Pi$ are progressive.
\end{definition}

Our main result is the following.
\begin{theorem}
\label{t:tpp}
Every \Focusinf-derivable sequent $\Phi$ has a thin and progressive
\Focusinf-proof.
\end{theorem}

For the proof of Theorem~\ref{t:tpp} we need some preparations.
Recall that we defined the linear order $\sqsubseteq$ on annotations such that 
$u \sqsubset f$.

\begin{definition}
\label{d:mf}
Let $\Sigma$ and $\Gamma$ be annotated sequents. 
We define $\morefocus{\Gamma}{\Sigma}$ to hold if for all $\phi^a \in \Gamma$
there is a $b \sqsupseteq a$ such that $\phi^b \in \Sigma$.
\end{definition}

\begin{definition} \label{d:backcl}
Let $\Si$ be a set of annotated formulas.
We define $\backcl_0(\Sigma)$ as the set of all annotated formulas $\phi^a$ 
such that either
\begin{enumerate}
\item $\phi^b \in \Sigma$ for some $b \sqsupseteq a$;
\item 
 $\phi = \phi_0 \lor \phi_1$, and $\phi_0^a \in \Sigma$ and $\phi_1^a \in \Sigma$;
\item 
 $\phi = \phi_0 \land \phi_1$, and $\phi_0^a \in \Sigma$ or $\phi_1^a \in \Sigma$;
\item $\phi = \mu x . \phi_0$ and $\phi_0(\phi)^u \in \Sigma$; or 
\item $\phi = \nu x . \phi_0$ and $\phi_0(\phi)^a \in \Sigma$.
\end{enumerate}
The map $\backcl_0$ clearly being a monotone operator on the sets of annotated
formulas, we define the \emph{backwards closure} of $\Si$ as the least fixpoint
$\backcl(\Sigma)$ of the operator $\Gamma \mapsto \Sigma \cup 
\backcl_0(\Gamma)$.
\end{definition}

In words, $\backcl(\Sigma)$ is the least set of annotated formulas such that
$\Sigma \subseteq \backcl(\Sigma)$ and $\backcl_0(\backcl(\Sigma)) \subseteq 
\backcl(\Sigma)$. 
The following proposition collects some basic properties of $\backcl$; recall 
that we abbreviate $\allfocus{\Sigma} = \allfocus{\uls{\Si}}$, that is, 
$\allfocus{\Sigma}$ consists of the annotated formulas $\phi^{f}$ such that 
$\phi^a \in \Sigma$ for some $a$.

\begin{proposition} 
\label{p:progressive facts}\label{p:pf}
The map $\backcl$ is a closure operator on the collection of sets of annotated 
formulas. 
Furthermore, the following hold for any pair of annotated sequents $\Ga,\Si$.
\begin{enumerate}
\item 
If $\morefocus{\Gamma}{\Sigma}$ then $\Gamma \subseteq \backcl(\Sigma)$.
\item 
\label{i:pf:2} 
If $\Gamma \subseteq \backcl(\Sigma)$ and $\Gamma$ contains only atomic or 
modal formulas, then $\morefocus{\Gamma}{\Sigma}$.
\item \label{i:proof rules} 
If $\Gamma$ is the conclusion and $\Sigma$ is one of the premises of an 
application of one of the rules \RuOr, \RuAnd, \RuMu, or \RuNu, then 
$\Gamma \subseteq \backcl(\Sigma)$.
\item \label{i:thinning} 
$\{\phi^u,\phi^f\} \cup \Sigma \subseteq
\backcl(\{\phi^f\} \cup \Sigma)$.
\item \label{i:more focus} 
If $\phi^a \in \backcl(\Sigma)$ for some $a$ then $\phi^u, \phi^f \in 
\backcl(\allfocus{\Sigma})$.
\end{enumerate}
\end{proposition}

\begin{proof}
These statements are straightforward consequences of
the Definitions~\ref{d:mf} and~\ref{d:backcl}.

For instance, in order to establish part \eqref{i:more focus} it
suffices to prove the following:
\begin{equation}
\label{eq:974}
\phi^{a} \in \backcl_{0}(\Si) \text{ only if } \phi^{f} \in \backcl(\Si^{f}).
\end{equation}
To see this, take an arbitary annotated formula $\phi^{a} \in \backcl_{0}(\Si)$
and make a case distinction as to the reason why $\phi^{a} \in \backcl_{0}(\Si)$.
(1) If $\phi^b \in \Sigma$ for some $b \sqsupseteq a$, then $\Phi^{f} \in \Si^{f}$,
and so $\phi^{f} \in \backcl_{0}(\Si) \sse \backcl(\Si)$.
(2) If $\phi = \phi_0 \lor \phi_1$, and $\phi_0^a, \phi_1^a \in \Sigma$
then $\phi_0^f, \phi_1^f \in \Si^{f}$, so that $\phi^{f} \in \Si^{f}$.
(3) If $\phi = \phi_0 \land \phi_1$, and $\phi_i^a \in \Sigma$ for some $i \in
\{0,1\}$, then $\phi_i^f \in \Sigma^{f}$ so that $\phi^{f}\in\Si^{f}$.
(4) If $\phi = \mu x . \phi_0$ and $\phi_0(\phi)^u \in \Sigma$, then clearly
also $\phi_0(\phi)^u \in \backcl(\Si)$, and so $\phi^{a} \in 
\backcl(\backcl(\Si)) \sse \backcl(\Si)$.
Finally, (5) if $\phi = \nu x . \phi_0$ and $\phi_0(\phi)^a \in \Sigma$, then
$\phi_0(\phi)^f \in \Si^{f}$, so that $\phi^{f} \in \backcl_{0}(\Si) \sse 
\backcl(\Si)$ indeed.
\end{proof}

%
\begin{definition}
A pre-proof $\Pi'$ of $\Gamma'$ is a \emph{simulation}
of a pre-proof $\Pi$ of $\Gamma$ if
$\Gamma \subseteq \backcl(\Gamma')$, and
for every open assumption $\Delta'$ of $\Pi'$ there is an open assumption 
$\Delta$ of $\Pi$ such that $\Delta \subseteq \backcl(\Delta')$.
\end{definition}

In the proof below we will frequently use the following proposition, the proof
of which is straightforward.

\begin{proposition} 
	\label{p:thinning}
Let $\Gamma$ and $\Delta$ be two sequents such that $\Gamma \subseteq
\backcl(\Delta)$. 
Then $\thin{\De}$ is thin and satisfies $\Gamma \subseteq \backcl(\thin{\De})$,
and there is a thin, progressive proof $\Pi$ of $\Delta$, which has $\thin{\De}$
as its only open assumption and uses only the weakening rule.
\end{proposition}

\begin{proof}
It is clear that $\thin{\De}$ is thin and that we may write $\De = \{\phi_1^u,
\dots,\phi_n^u\} \cup \thin{\De}$, where $\phi_{1},\ldots,\phi_{n}$ are the 
formulas that occur both focused and unfocused in $\De$. 
We then let $\Pi'$ be the proof that weakens the formulas $\phi_1^u,\dots,
\phi_n^u$ one by one.
By item~\ref{i:thinning} of Proposition~\ref{p:progressive facts} it
follows that $\Delta \subseteq \backcl(\thin{\De})$. 
Thus, $\Gamma
\subseteq \backcl(\Delta)$ implies $\Gamma \subseteq \backcl(\thin{\De})$
because $\backcl$ is a closure operator.
\end{proof}

The key technical observation in the proof of Theorem~\ref{t:tpp} is 
Proposition~\ref{p:ps} below.

\begin{definition}
\label{p:baspr}
A pre-proof $\Pi = (T,P,\Si,\pLab)$ is \emph{basic} if $T$ consists
of the root $r$ and its successors, $\pLab_{r} \neq \RuDischarge{}$ and 
$\pLab_{u} = \star$ for every successor of $r$.
\end{definition}

A basic derivation is thus a pre-proof $\Pi = (T,P,\Si,\pLab)$ of $\Si_{r}$
(where $r$ is the root of $\Pi$) with open assumptions $\{ \Si_{u} \mid u \neq
r \}$.

\begin{proposition} 
\label{p:progressive simulation} \label{p:ps}
Let $\Pi$ be a basic pre-proof of $\Gamma$ with root $r$ and let $\Gamma'$ 
be a sequent such that $\Gamma \subseteq \backcl(\Gamma')$. 
Then there is a thin and progressive simulation $\Pi'$ of $\Pi$ that proves the 
sequent $\Gamma'$.
Moreover, if $\pLab_{r} \neq \RuFocus, \RuUnfocus$ then $\Pi'$ does not use
$\RuFocus$ or $\RuUnfocus$, and if $\pLab_{r} = \RuBox$ then $\RuBox$ is
also the rule applied at the root of $\Pi'$.
\end{proposition}

Before we prove this proposition, we first show how our main theorem follows
from it.

\begin{proofof}{Theorem~\ref{t:tpp}}
Let $\Pi = (T,P,\Si,\Ru)$ be a \Focusinf-proof of the sequent $\Phi$, then by
definition we have $\Si_{r} = \Phi^{f}$, where $r$ is the root of $\Pi$.
Obviously we have $\Si_{r} \sse \backcl(\Si_{r})$.

We will transform $\Pi$ into a thin proof of $\Phi$ as follows.
On the basis of Proposition~\ref{p:ps} it is straightforward to define a map 
$\Xi$ which assigns a thin sequent $\Xi_{t}$ to each node $t \in T$, in such
a way that $\Xi_{r} \isdef \Si_{r}$, 
and for every $t \in T$ we find $\Si_{t} \sse \backcl(\Xi_{t})$,
while we also have a thin and progressive pre-proof $\Pi_{t}$ of the sequent 
$\Xi_{t}$ from the assumptions $\{ \Xi_{u} \mid Ptu \}$.
In addition we know that if $\Ru_{t} \neq \RuFocus, \RuUnfocus$, then the 
derivation $\Pi_{t}$ does not involve the focus rules, and that if $\pLab_{t}
= \RuBox$ then $\RuBox$ is also the rule applied at the root of $\Pi_{t}$.
We obtain a thin and progressive proof $\Pi'$ from this by simply adding all 
these thin and progressive derivations $\Pi_{t}$ to the `skeleton structure'
$(T,P,\Xi)$, in the obvious way.

It is easy to show that $\Pi'$ is a pre-proof, and the additional conditions 
on the focus rules and $\RuBox$ guarantee that every infinite branch of $\Pi'$
witnesses infinitely many applicatinos of $\RuBox$, but only finitely
many applications of the focus rules.
To prove the remaining condition on focused formulas, consider an infinite branch
$\al = (v_{n})_{n\in\om}$ of $\Pi'$.
It is easy to see that by construction we may associate an infinite branch 
$\be = (t_{n})_{n\in\om}$ of $\Pi$ with $\al$, together with a map $f: \om \to 
\om$ such that $\Si_{t_{n}} \sse \backcl(\Xi_{v_{f(n)}})$.
This path $\be$ is successful since $\Pi$ is a proof, and so there is a $k \in
\om$ such that for all $n \geq k$ the sequent $\Si_{t_{n}}$ contains a formula
in focus, and $\pLab(t_{n}) \neq \RuFocus$.
But by Proposition~\ref{p:pf}(\ref{i:pf:2}) for any $n\geq k$ such that 
$\pLab(t_{n}) = \RuBox$, the sequent $\Xi_{v_{f(n)}}$ must contain a focused
formula as well.
Since $\al$ features infinitely many applications of $\RuBox$, this implies 
the existence of infinitely many nodes $v_{m}$ on $\al$ such that $\Xi_{v_{m}}$
contains a focused formula.
And since the focus rule is applied only finitely often on $\al$, by 
Proposition~\ref{p:lr1} it follows from this that $\al$ actually contains 
cofinitely many such nodes, as required.

Furthermore it is obvious that, being constructed by glueing together thin and
progressive proofs, $\Pi'$ has these properties as well.
Finally, since $\Xi_{r} = \Si_{r} = \Phi^{f}$, we have indeed obtained a proof 
for the plain sequent $\Phi$.
\end{proofof}

\begin{proofof}{Proposition~\ref{p:ps}}
By definition of a basic proof, $\Pi = (T,P,\Si,\Ru)$ consists of nothing more
than a single application of the rule $\Ru \isdef \pLab_{r}$ to the annotated 
sequent $\Gamma = \Si_{r}$, where $r$ is the root of $\Pi$.
Because of Proposition~\ref{p:thinning} we can assume without loss of generality
that $\Gamma'$ is thin. 
We then make a case distinction depending on the rule $\Ru$.

Recall that we use $\RuWeak^*$ to denote a finite (potentially zero) number of 
successive applications of weakening.

\begin{description}

\item[\it Case for \AxLit:]
In this case $\Pi$ is of the form
\begin{center}
\begin{prooftree}
 \AxiomC{\phantom{X}}
 \RightLabel{\AxLit}
 \UnaryInfC{$p^a, \atneg{p}^b$}
\end{prooftree}
\end{center}
The assumption is that $\{p^a,\atneg{p}^b\} \subseteq \backcl(\Gamma')$.
By item~\ref{i:pf:2} in Proposition~\ref{p:progressive facts} it
follows that $p^a,\atneg{p}^b \in \Gamma'$.
We can thus define $\Pi'$ to be the proof
\begin{center}
\begin{prooftree}
 \AxiomC{\phantom{X}}
 \RightLabel{\AxLit}
 \UnaryInfC{$p^a, \atneg{p}^b$}
 \RightLabel{$\RuWeak^*$}
 \UnaryInfC{$\Gamma'$}
\end{prooftree}
\end{center}

\item[\it Case for \AxTop:]
In this case $\Pi$ is of the form
\begin{center}
\begin{prooftree}
 \AxiomC{\phantom{X}}
 \RightLabel{\AxTop}
 \UnaryInfC{$\top^a$}
\end{prooftree}
\end{center}
From the assumption that $\top^a \subseteq \backcl(\Gamma')$ it follows
with item~\ref{i:pf:2} of Proposition~\ref{p:progressive facts}
that that $\top^a \in \Gamma'$. We define $\Pi'$ to be the proof.
\begin{center}
\begin{prooftree}
 \AxiomC{\phantom{X}}
 \RightLabel{\AxLit}
 \UnaryInfC{$\top^a$}
 \RightLabel{$\RuWeak^*$}
 \UnaryInfC{$\Gamma'$}
\end{prooftree}
\end{center}

\item[\it Case for \RuOr:]
In this case $\Gamma = \phi_0 \lor \phi_1,\Sigma$ and $\Pi$ is of
the form
\begin{center}
\begin{prooftree}
 \AxiomC{$\phi_0^a,\phi_1^a,\Sigma$}
 \RightLabel{\RuOr}
 \UnaryInfC{$(\phi_0 \lor \phi_1)^a,\Sigma$}
\end{prooftree}
\end{center}
Let $\phi \isdef \phi_0 \lor \phi_1$. Because $\Gamma \subseteq
\backcl(\Gamma')$ it follows that $\phi^a \in \backcl(\Gamma')$.
By definition of $\backcl$ there are two cases for why this might hold,
either $\phi^b \in \Gamma'$ for $b \sqsupseteq a$ or $\phi_0^a \in
\backcl(\Gamma')$ and $\phi_1 \in \backcl(\Gamma')$.

In the latter case where $\phi_0^a \in \backcl(\Gamma')$ and
$\phi_1 \in \backcl(\Gamma')$ we can let $\Pi'$ consist of just the
sequent $\Gamma'$. This proof is thin and progressive and it clear
follows that $\phi_0^a,\phi_1^a,\Sigma \subseteq \backcl(\Gamma')$
because $\Sigma \subseteq \Gamma \subseteq \backcl(\Gamma')$.

In the former case, where $\phi^b \in \Gamma'$ for some $b
\sqsupseteq a$, consider the proof
\begin{center}
\begin{prooftree}
 \AxiomC{$\phi_0^b,\phi_1^b,\Gamma' \setminus \{\phi^b\}$}
 \RightLabel{\RuOr}
 \UnaryInfC{$(\phi_0 \lor \phi_1)^b,\Gamma' \setminus \{\phi^b\}$}
\end{prooftree}
\end{center}
We let $\Pi'$ be this proof. Clearly, this is a proof of $\Gamma' =
(\phi_0 \lor \phi_1)^b,\Gamma' \setminus \{\phi^b\}$ and it is
progressive. Moreover, we have from the definition of $\backcl$ that
$\phi_0^a,\phi_1^a \subseteq \backcl(\phi_0^b, \phi_1^b)$,
as $b \sqsupseteq a$. By item~\ref{i:proof rules} of
Proposition~\ref{p:progressive facts} it holds that $\Gamma' \subseteq
\backcl(\phi_0^b,\phi_1^b,\Gamma' \setminus \{\phi^b\})$. By
assumption we have that $\Gamma \subseteq \backcl(\Gamma')$ and hence
$\Sigma \subseteq \Gamma \subseteq \backcl(\Gamma') \subseteq
\backcl(\phi_0^b,\phi_1^b,\Gamma' \setminus \{\phi^b\})$.
Putting all of these together it follows that
\[
 \phi_0^a,\phi_1^a,\Sigma \subseteq
\backcl(\phi_0^b,\phi_1^b,\Gamma' \setminus \{\phi^b\}).
\]
It remains to be seen that $\Pi$ can be made thin. For the sequent
$\Gamma'$ at the root of $\Pi$ we have already established that it is
thin. It might be, however, that the open assumption
$\phi_0^b,\phi_1^b,\Gamma' \setminus \{\phi^b\}$ is not thin.
If this is the case we can simply apply Proposition~\ref{p:thinning} and
obtain the required proof.

\item[\it Case for \RuAnd:]
In this case $\Gamma = \phi_0 \land \phi_1,\Sigma$ and $\Pi$ is of
the form
\begin{center}
\begin{prooftree}
 \AxiomC{$\phi_0^a,\Sigma$}
 \AxiomC{$\phi_1^a,\Sigma$}
 \RightLabel{\RuAnd}
 \BinaryInfC{$(\phi_0 \land \phi_1)^a,\Sigma$}
\end{prooftree}
\end{center}
Let $\phi \isdef \phi_0 \land \phi_1$. Because $\Gamma \subseteq
\backcl(\Gamma')$ it follows that $\phi^a \in \backcl(\Gamma')$.
By the definition $\backcl$ we may split into two cases: either $\phi^b \in 
\Gamma'$ for $b \sqsupseteq a$ or $\phi_i^a \in \backcl(\Gamma')$ for some 
$i \in \{0,1\}$.

In the subcase where $\phi_i^a \in \backcl(\Gamma')$ for some $i \in
\{0,1\}$ we let $\Pi'$ just be the sequent $\Gamma'$. This sequent is
thin and the proof is trivially progressive. We need to show that there
is some open assumption $\Delta_i$ of $\Pi$ such that $\Delta_i
\subseteq \backcl(\Gamma')$. Let this be the assumption $\phi_i^a,
\Sigma$. We already know that $\phi_i^a \in \backcl(\Gamma')$, so we
it only remains to be seen that $\Sigma \subseteq \backcl(\Gamma')$. But
this follows because $\Sigma \subseteq \Gamma$ and $\Gamma \subseteq
\backcl(\Gamma')$.

In the other subcase we have that $\phi^b \in \Gamma'$ for some $b
\sqsupseteq a$. We let $\Pi'$ be the proof
\begin{center}
\begin{prooftree}
 \AxiomC{$\phi_0^b,\Gamma' \setminus \{\phi^b\}$}
 \AxiomC{$\phi_1^b,\Gamma' \setminus \{\phi^b\}$}
 \RightLabel{\RuAnd}
 \BinaryInfC{$(\phi_0 \land \phi_1)^b,\Gamma' \setminus
\{\phi^b\}$}
\end{prooftree}
\end{center}
By definition this proof is progressive and it is a proof of $\Gamma' =
(\phi_0 \land \phi_1)^b,\Gamma' \setminus \{\phi^b\}$. We then
show that for each open assumption $\phi_i^b, \Gamma' \setminus
\{\phi^b\}$ of $\Pi$, where $i \in \{0,1\}$, there is the open
assumption $\phi_i^a,\Sigma$ of $\Pi$ such that
\begin{equation*}
 \phi_i^a,\Sigma \subseteq \backcl(\phi_i^b, \Gamma' \setminus
\{\phi^b\}).
\end{equation*}
Because $a \sqsubseteq b$ it is clear that $\phi_i^a \in
\backcl(\{\phi_i^b\})$. So we only need $\Sigma \subseteq
\backcl(\phi_i^b, \Gamma' \setminus \{\phi^b\})$. But this follows
from $\Sigma \subseteq \Gamma \subseteq \backcl(\Gamma')$ and the fact
that $\Gamma' \subseteq \backcl(\phi_i^b,\Gamma' \setminus
\{\phi^b\})$, which is item~\ref{i:proof rules} in
Proposition~\ref{p:progressive facts}. 
Finally, as before, we use Proposition~\ref{p:thinning} to deal with non-thin
open assumptions of $\Pi'$, if any.

\item[\it Case for \RuMu:]
In this case $\Gamma = \mu x . \phi_0(x),\Sigma$ and $\Pi$ is of the
form
\begin{center}
\begin{prooftree}
 \AxiomC{$\phi_0(\phi)^u,\Sigma$}
 \RightLabel{\RuMu}
 \UnaryInfC{$(\mu x . \phi_0(x))^a,\Sigma$}
\end{prooftree}
\end{center}
Here we write $\phi = \mu x . \phi_0(x)$. Because $\Gamma \subseteq
\backcl(\Gamma')$ it follows that $\phi^u \in \backcl(\Gamma')$. 
By definition of $\backcl$ this gives us the cases that either $\phi^b
\in \Gamma'$ for some $b \sqsupseteq a$ or $\phi_0(\phi)^u \in
\backcl(\Gamma')$.

In the subcase where $\phi_0(\phi)^u \in \backcl(\Gamma')$ we let
$\Pi'$ just be the sequent $\Gamma'$. This sequent is thin and the proof
is trivially progressive. We need to show $\phi_0(\phi)^u,\Sigma
\subseteq \backcl(\Gamma')$. Because we are in the subcase for
$\phi_0(\phi)^u \in \backcl(\Gamma')$ it suffice to show that
$\Sigma \subseteq \backcl(\Gamma')$. But this follows because $\Sigma
\subseteq \Gamma$ and $\Gamma \subseteq \backcl(\Gamma')$.

In the other subcase we have that $\phi^b \in \Gamma'$ for some $b
\sqsupseteq a$. We let $\Pi'$ be the proof
\begin{center}
\begin{prooftree}
 \AxiomC{$\phi_0(\phi)^u,\Gamma' \setminus \{\phi^b\}$}
 \RightLabel{\RuMu}
 \UnaryInfC{$(\mu x . \phi_0(x))^b,\Gamma' \setminus
\{\phi^b\}$}
\end{prooftree}
\end{center}
Clearly, this proof is progressive and it is a proof of $\Gamma' =
(\mu x . \phi_0(x))^b,\Gamma' \setminus
\{\phi^b\}$. We can also show that
\[
 \phi_0(\phi)^u,\Sigma \subseteq \backcl(\phi_0(\phi)^u,
\Gamma' \setminus \{\phi^b\}).
\]
For this it clearly suffices to show that $\Sigma \subseteq
\backcl(\phi_0(\phi)^u, \Gamma' \setminus \{\phi^b\})$. This
follows from $\Sigma \subseteq \Gamma \subseteq \backcl(\Gamma')$ and
the fact that $\Gamma' \subseteq \backcl(\phi_0(\phi)^u, \Gamma'
\setminus \{\phi^b\})$, which comes from item~\ref{i:proof rules} in
Proposition~\ref{p:progressive facts}. 
Finally, as before, we use Proposition~\ref{p:thinning} to deal with non-thin
open assumptions of $\Pi'$, if any.

\item[\it Case for \RuNu:]
In this case $\Gamma = \nu x . \phi_0(x),\Sigma$ and $\Pi$ is of the
form
\begin{center}
\begin{prooftree}
 \AxiomC{$\phi_0(\phi)^a,\Sigma$}
 \RightLabel{\RuNu}
 \UnaryInfC{$(\nu x . \phi_0(x))^a,\Sigma$}
\end{prooftree}
\end{center}
Here, we write $\phi = \nu x . \phi_0(x)$. Because $\Gamma \subseteq
\backcl(\Gamma')$ it follows that $\phi^u \in \backcl(\Gamma')$. By
the definition $\backcl$ this gives us the cases that either $\phi^b
\in \Gamma'$ for some $b \sqsupseteq a$ or $\phi_0(\phi)^u \in
\backcl(\Gamma')$.

In the subcase where $\phi_0(\phi)^u \in \backcl(\Gamma')$ we let
$\Pi'$ just be the sequent $\Gamma'$. This sequent is thin and the proof
is trivially progressive. We need to show $\phi_0(\phi)^u,\Sigma
\subseteq \backcl(\Gamma')$. Because we are in the subcase for
$\phi_0(\phi)^u \in \backcl(\Gamma')$ it suffice to show that
$\Sigma \subseteq \backcl(\Gamma')$. But this follows because $\Sigma
\subseteq \Gamma$ and $\Gamma \subseteq \backcl(\Gamma')$.

In the other subcase we have that $\phi^b \in \Gamma'$ for some $b
\sqsupseteq a$. We let $\Pi'$ be the proof
\begin{center}
\begin{prooftree}
 \AxiomC{$\phi_0(\phi)^b,\Gamma' \setminus \{\phi^b\}$}
 \RightLabel{\RuNu}
 \UnaryInfC{$(\nu x . \phi_0(x))^b,\Gamma' \setminus
\{\phi^b\}$}
\end{prooftree}
\end{center}
Clearly, this proof is progressive and it is a proof of $\Gamma' = (\mu
x . \phi_0(x))^b,\Gamma' \setminus \{\phi^b\}$. We can also show
that
\[
 \phi_0(\phi)^a,\Sigma \subseteq \backcl(\phi_0(\phi)^b,
\Gamma' \setminus \{\phi^b\}).
\]
Because $a \sqsubseteq b$ it is clear that $\phi_0(\phi)^a \in
\backcl(\{\phi_0(\phi)^b\})$. So it clearly suffices to show that $\Sigma \subseteq
\backcl(\phi_0(\phi)^b, \Gamma' \setminus \{\phi^b\})$. This
follows from $\Sigma \subseteq \Gamma \subseteq \backcl(\Gamma')$ and
the fact that $\Gamma' \subseteq \backcl(\phi_0(\phi)^b, \Gamma'
\setminus \{\phi^b\})$, which comes from item~\ref{i:proof rules} in
Proposition~\ref{p:progressive facts}. 
Any remaining non-thin open assumptions are dealt with using 
Proposition~\ref{p:thinning}.

\item[\it Case for \RuBox:]
In this case $\Ga$ must be of the form $\Ga = \Box\phi^{a},\dia\Si$, and $\Pi$
is the derivation
\begin{center}
\begin{prooftree}
 \AxiomC{$\phi^{a},\Sigma$}
 \RightLabel{\RuBox}
 \UnaryInfC{$\Box\phi^a,\dia\Sigma$}
\end{prooftree}
\end{center}
Because $\Ga \sse \backcl(\Ga')$ it follows from 
Proposition~\ref{p:pf}\eqref{i:pf:2} that $\morefocus{\Ga}{\Ga'}$.
But then $\Ga'$ must contain a subset of the form $\Box\phi^{b},\dia\Si'$, with 
$a \sqsubseteq b$ and $\morefocus{\Si}{\Si'}$.
Consider the following derivation $\Pi'$:
\begin{center}
\begin{prooftree}
 \AxiomC{$\phi^{b},\Sigma'$}
 \RightLabel{\RuBox}
 \UnaryInfC{$\Box\phi^b,\dia\Sigma'$}
\RightLabel{$\RuWeak^*$}
 \UnaryInfC{$\Gamma'$}
\end{prooftree}
\end{center}
It is easy to see that we have $\morefocus{\De}{\De'}$, where $\De \isdef
\phi^{a},\Si$ and $\De' \isdef \phi^{b},\Si'$ are the assumptions of the 
pre-proofs $\Pi$ and $\Pi'$, respectively.
Furthermore, the proof $\Pi'$ is obviously progressive, and if not thin already,
can be made so by applying Proposition~\ref{p:thinning}.

\item[\it Case for \RuWeak:]
In this case $\Gamma = \phi^a,\Sigma$ and $\Pi$ is of the form
\begin{center}
\begin{prooftree}
 \AxiomC{$\Sigma$}
 \RightLabel{\RuWeak}
 \UnaryInfC{$\phi^a,\Sigma$}
\end{prooftree}
\end{center}
We can let $\Pi'$ consist of just the sequent $\Gamma'$. This sequent
is thin and the proof is trivially progressive. We need to show that
$\Sigma \subseteq \backcl(\Gamma')$. Clearly $\Sigma \subseteq \Gamma$,
and $\Gamma \subseteq \backcl(\Gamma')$ holds by assumption.

\item[\it Case for \RuFocus:]
In this case $\Gamma = \phi^a,\Sigma$ and $\Pi$ is of the form
\begin{center}
\begin{prooftree}
 \AxiomC{$\phi^f,\Sigma$}
 \RightLabel{\RuFocus}
 \UnaryInfC{$\phi^u,\Sigma$}
\end{prooftree}
\end{center}
We let $\Pi'$ be the proof
\begin{center}
\begin{prooftree}
 \AxiomC{$\allfocus{(\Gamma')}$}
 \RightLabel{$\RuFocus^*$}
 \UnaryInfC{$\Gamma'$}
\end{prooftree}
\end{center}
Here, $\allfocus{(\Gamma')} = \{\phi^f \mid \phi^a \in \Gamma'
\mbox{ for some } a \in \{u,f\}\}$, as in Proposition~\ref{p:progressive
facts}, and $\RuFocus^*$ are as many applications of the focus rule as
we need to put every formula in $\Gamma'$ in focus. This proof $\Pi'$ is
trivially progressive and it is thin because $\Gamma'$ is thin and hence
changing the annotations of some formulas in $\Gamma'$ in this way still
yields a thin sequent. 
From item~\ref{i:more focus} of
Proposition~\ref{p:progressive facts} it is clear that $\phi^f,\Sigma
\subseteq \backcl{\allfocus{(\Gamma')}}$ is implied by $\phi^u,\Sigma
\subseteq \backcl(\Gamma')$.

\item[\it Case for \RuUnfocus:]
In this case $\Gamma = \phi^f,\Sigma$ and $\Pi$ is of the form
\begin{center}
\begin{prooftree}
 \AxiomC{$\phi^u,\Sigma$}
 \RightLabel{\RuUnfocus}
 \UnaryInfC{$\phi^f,\Sigma$}
\end{prooftree}
\end{center}
We can let $\Pi'$ consist of just the sequent $\Gamma'$. This sequent
is thin and the proof is trivially progressive. We need to show that
$\phi^u,\Sigma \subseteq \backcl(\Gamma')$. By the definition of
$\backcl$ we have that $\phi^u \in \backcl(\phi^f)$. 
Thus $\phi^u, \Sigma \subseteq \backcl(\phi^f,\Sigma)$. 
Moreover, we have by assumption that $\phi^f,\Sigma = \Gamma \subseteq
\backcl(\Gamma')$. 
Putting this together, and using that $\backcl$ is a
closure operator, we get $\phi^u,\Sigma \subseteq \backcl(\Gamma')$.
\end{description}

Since we have covered all the cases in the above case distinction, this proves
the main part of the proposition.
The additional statements about the focus rules and the rule $\RuBox$ can 
easily be verified from the definition of $\Pi'$ given above.
\end{proofof}

\section{Tableaux and tableau games}
\label{s:nw tableaux}
\label{sec-tab}

In this section we define a tableau game for the alternation-free
$\mu$-calculus that is a adaptation of the tableau game by Niwi\'{n}ski
and Walukiewicz \cite{niwi:game96}. We also show that the tableau game
is adequate with respect to the semantics in Kripke frames, meaning that
\Prover has a winning strategy in the tableau game for some tableau of
some formula iff the formula is valid. The soundness and completeness
proofs for the focus system of this paper rely on this result. There we
will exploit that proofs in the focus system closely correspond to
winning strategies for one of the two players in the tableau game.

\subsection{Tableaux}

We first introduce tableaux, which are the graph over which the tableau game 
is played. 
The nodes of a tableau for some formula $\varphi$ are labelled with sequents
(as defined in the previous section) consisting of formulas taken from the
closure of $\varphi$.

Our tableaux are defined from the perspective that sequents are read 
disjunctively. 
We show below that \Prover has a
winning strategy in the tableau for some sequent if the disjunction of
its formulas are valid. This is is different from the satisfiability
tableaux in \cite{niwi:game96}, where sequents are read conjunctively.

The tableau system is based on the rules in Figure~\ref{f:tableaux rules}. 
We use the same terminology here as we did for rules in the focus system.
The tableau rules \AxLit, \AxTop, \RuOr, \RuAnd, \RuMu and \RuNu are direct
counterparts of the focus proof rules with the same name, the only difference
being that the tableau rules are simpler since they do not involve the 
annotations.


\begin{figure}[thb]
\begin{minipage}{\textwidth}
\begin{minipage}{0.20\textwidth}
\begin{prooftree}
 \AxiomC{\phantom{X}}
 \RightLabel{\AxLit}
 \UnaryInfC{$p, \atneg{p}, \Phi$}
\end{prooftree}
\end{minipage}
\begin{minipage}{0.20\textwidth}
\begin{prooftree}
 \AxiomC{\phantom{X}}
 \RightLabel{\AxTop}
 \UnaryInfC{$\top, \Phi$}
\end{prooftree}
\end{minipage}
\begin{minipage}{0.24\textwidth}
\begin{prooftree}
 \AxiomC{$\varphi,\psi,\Phi$}
 \RightLabel{\RuOr}
 \UnaryInfC{$\varphi \lor \psi,\Phi$}
\end{prooftree}
\end{minipage}
\begin{minipage}{0.30\textwidth}
\begin{prooftree}
 \AxiomC{$\varphi, \Phi$}
 \AxiomC{$\psi,\Phi$}
 \RightLabel{\RuAnd}
 \BinaryInfC{$\varphi \land \psi,\Phi$}
\end{prooftree}
\end{minipage}
\end{minipage}

\bigskip

\begin{minipage}{\textwidth}
\begin{minipage}{0.45\textwidth}
\begin{prooftree}
 \AxiomC{$\varphi_1,\Phi$}
 \AxiomC{$\dots$}
 \AxiomC{$\varphi_n,\Phi$}
 \LeftLabel{(\dag)}
 \RightLabel{\RuMod}
 \TrinaryInfC{$\Psi,\Box \varphi_1,
\dots, \Box \varphi_n, \Diamond \Phi$}
\end{prooftree}
\end{minipage}
\begin{minipage}{0.22\textwidth}
\begin{prooftree}
 \AxiomC{$\varphi[\mu x . \varphi / x], \Phi$}
 \RightLabel{\RuMu}
 \UnaryInfC{$\mu x . \varphi, \Phi$}
\end{prooftree}
\end{minipage}
\begin{minipage}{0.22\textwidth}
\begin{prooftree}
 \AxiomC{$\varphi[\nu x . \varphi / x], \Phi$}
 \RightLabel{\RuNu}
 \UnaryInfC{$\nu x . \varphi, \Phi$}
\end{prooftree}
\end{minipage}
\end{minipage}

\caption{Rules of the tableaux system}
\label{f:tableaux rules}
\end{figure}

The \emph{modal rule} \RuMod can be seen as a game-theoretic version of the 
box rule \RuBox from the focus system, differing from it in two ways.
First of all, the number of premises of \RuMod is not fixed, but depends on the
number of box formulas in the conclusion; as a special case, if the conclusion 
contains no box formula at all, then the rule has an empty set of premises, 
similar to an axiom.
Second, the rule \RuMod does allow side formulas in the consequent that are not
modal; note however, that \RuMod has as its \emph{side condition} (\dag) that 
this set $\Psi$ contains atomic formulas only, and that it is \emph{locally
falsifiable}, i.e., $\Psi$ does not contain $\top$ and there is no proposition 
letter $p$ such that both $p$ and $\atneg{p}$ belong to $\Psi$.
This side condition guarantees that \RuMod is only applicable if no other
tableau rule is.

\begin{definition}
\label{d:tableau}
A \emph{tableau} is a quintuple $\bbT = (V,E,\Phi,\tLab,v_I)$, where $V$ is a 
set of \emph{nodes}, $E$ is a binary relation on $V$, $v_I$ is the \emph{initial
node} or \emph{root} of the tableau, and both $\Phi$ and $\tLab$ are labelling
functions.
Here $\Phi$ maps every node $v$ to a non-empty sequent $\Phi_v$, and 
\[
\tLab: V \to \{\AxLit,\AxTop,\RuOr,\RuAnd,\RuMod,\RuMu,\RuNu\}
\]
associates a proof rule $\tLab_{v}$ with each node $v$ in $V$.
Tableaux are required to satisfy the following \emph{coherence} conditions:

\begin{enumerate}[resume]

\item \label{i:tr1}
If a node is labelled with the name of a proof rule then it has as many 
successors as the proof rule has premises, and the sequents at the 
node and its successors match the specification of the proof rules in 
Figure~\ref{f:tableaux rules}.

\item \label{i:tr2}
A node can only be labelled with the modal rule $\RuMod$ if its side condition
(\dag) is met.

\item \label{i:tr3}
In any application of the rules $\RuOr, \RuAnd, \RuMu$ and $\RuNu$,
the principal formula is not an element of the context $\Phi$.

\end{enumerate}
A tableau is a \emph{tableau for a sequent} $\Phi$ if $\Phi$ is the sequent of 
the root of the tableau.
\end{definition}

Observe that it follows from condition~\ref{i:tr2} in Definition~\ref{d:tableau}
that if a node $u$ is labelled with \RuMod, then no other rule is applicable.

\begin{proposition} \label{p:tableau exists}
There is a tree-based tableau for every sequent $\Phi$.
\end{proposition}

\begin{proof}
This can be proved in a straightforward step-wise procedure in which we 
construct the tree underlying $\bbT$ by repeatedly extending it at non-axiomatic
leaves using any of the proof rules that are applicable at that leaf. 
This generates a possibly infinite tree that is a tableau because in every
sequent there is at least one rule applicable. 
Note that \RuMod can be applied in sequents without modal formulas, in which 
case it has no premises and thus creates a leaf of the tableau.
\end{proof}

A crucial aspect of tableaux for the $\mu$-calculus is that one has to
keep track of the development of individual formulas along infinite
paths in the tableau. For this purpose we define the notion of a trail
in a path of the tableau.

\begin{definition}
\label{d:tableaux trails}
Let $\bbT = (V,E,\Phi,\tLab,v_I)$ be a tableau. 
For all nodes $u,v \in V$ such that $E u v$ we define the \emph{active trail
relation} $\atrail_{u,v} \subseteq \Phi_u \times \Phi_v$ and the \emph{passive
trail relation} $\ptrail_{u,v} \subseteq \Phi_u \times \Phi_v$, both of which
relate formulas in the sequents at $u$ and $v$.
The idea is that $\atrail$ connects the active formulas in the premise and 
conclusion, whereas $\ptrail$ connects the side formulas. 
Both relations are defined via a case distinction depending on the rule that is 
applied at $u$:

\textit{Case $\tLab_{u} = \RuOr$:} Then $\Phi_u = \{\varphi \lor \psi\}
\cup \Psi$ and $\Phi_v = \{\varphi,\psi\} \cup \Psi$ for some sequent
$\Psi$. We define $\atrail_{u,v} = \{(\varphi \lor
\psi,\varphi),(\varphi \lor \psi,\psi)\}$ and $\ptrail_{u,v} =
\Delta_\Psi$, where $\Delta_\Psi = \{(\varphi,\varphi) \mid \varphi \in
\Psi\}$.

\textit{Case $\tLab_{u} = \RuAnd$:} In this case $\Phi_u = \{\varphi \land
\psi\} \cup \Psi$ and $\Phi_v = \{\chi\} \cup \Psi$ for some sequent
$\Psi$ and $\chi$ such that $\chi = \varphi$ if $v$ corresponds to the
left premise of \RuAnd and $\chi = \psi$ if $v$ corresponds to the right
premise. In both cases we set $\atrail_{u,v} = \{(\varphi \land
\psi,\chi)\}$ and $\ptrail_{u,v} = \Delta_\Psi$.

\textit{Case $\tLab_{u} = \RuMod$:} Then $\Phi_u = \Psi \cup \{\Box
\varphi_1,\dots,\Box \varphi_n\} \cup \Diamond \Phi$ and $\Phi_v =
\{\phi_{v}\} \cup \Phi$ for some sequent $\Phi$ and locally falsifiable
set of literals $\Psi \subseteq \Lit$. We can thus define $\atrail_{u,v}
= \{(\Box \phi_{v},\phi_{v})\} \cup \{(\Diamond \phi, \phi) \mid \phi \in
\Phi\}$ and $\ptrail_{u,v} = \nada$.

\textit{Case $\tLab_{u} = \RuMu$:} Then $\Phi_u = \{\mu x . \varphi\}
\cup \Psi$ and $\Phi_v = \{\varphi[\mu x . \varphi / x]\} \cup \Psi$ for
some sequent $\Psi$. We define $\atrail_{u,v} = \{(\mu x .
\varphi,\varphi[\mu x . \varphi / x])\}$ and $\ptrail_{u,v} =
\Delta_\Psi$.

\textit{Case $\tLab_{u} = \RuNu$:} Then $\Phi_u = \{\nu x . \varphi\} \cup
\Psi$ and $\Phi_v = \{\varphi[\nu x . \varphi / x]\} \cup \Psi$ for some
sequent $\Psi$. We define $\atrail_{u,v} = \{(\nu x . \varphi,\varphi[\nu x .
\varphi / x])\}$ and $\ptrail_{u,v} = \Delta_\Psi$.

Note that it is not possible that $\tLab_{u} = \AxLit$ or $\tLab_{u} =
\AxTop$ because $u$ is assumed to have a successor.

Finally, for all nodes $u$ and $v$ with $Euv$, the \emph{general trail relation} 
$\gtrail_{u,v}$ is defined as $\gtrail_{u,v} \isdef \atrail_{u,v} \cup 
\ptrail_{u,v}$.
\end{definition}

Note that for any two nodes $u,v$ wth $Euv$ and $(\phi,\psi) \in \gtrail_{u,v}$,
we have either $(\phi,\psi) \in \atrail_{{u,v}}$ and $\psi \in 
\Clos_{0}(\phi)$, or else $(\phi,\psi) \in \atrail_{{u,v}}$ and $\phi = \psi$.

\begin{definition}
Let $\bbT = (V,E,\Phi,\tLab,v_I)$ be a tableau. 
A \emph{path} in $\bbT$ is simply a path in the underlying graph $(V,E)$ of 
$\bbT$, that is, a sequence $\pi = (v_{n})_{n<\kappa}$, for some ordinal 
$\kappa$ with $0 < \kappa \leq \omega$, such that $Ev_{i}v_{i+1}$ for every $i$
such that $i+1 < \kappa$.
A \emph{trail} on such a path $\pi$ is a sequence $(\phi_{n})_{n<\kappa}$ of 
formulas such that $(\phi_{i},\phi_{i+1}) \in \gtrail_{v_i,v_{i+1}}$, 
whenever $i+1 < \kappa$.
\end{definition}

\begin{remark}
Although our tableaux are very much inspired by the ones introduced by 
Niwi\'{n}ski and Walukiewicz~\cite{niwi:game96}, there are some notable
differences in the actual definitions.
In particular, the fixpoint rules in our tableaux simply unfold fixpoint 
formulas; that is, we omit the mechanism of definition lists.
Some minor differences are that we always decompose formulas until we reach 
literals, and that our tableaux are not necessarily tree-based.
\end{remark}

It is easy to see that because of guardedness, we have the following.

\begin{proposition} 
\label{p:prog}
Let $\pi$ be an infinite path in a tableau $\bbT$, and let 
$(\phi_{n})_{n<\omega}$ be a trail on $\pi$.
Then
\begin{urlist}
\item $\pi$ witnesses infinitely many applications of the rule \RuMod;
\item there are infinitely many $i$ such that $(\phi_{i},\phi_{i+1}) 
\in \atrail_{v_i,v_{i+1}}$.
\end{urlist}
\end{proposition}

Before we move on to the definition of tableau games, we need to have a closer
look at trails.
Recall that for any two nodes $u,v \in V$, the trail relation $\gtrail_{u,v}$
is the union of an active and a passive trail relation, and that the passive
relation is always a subset of the diagonal relation on formulas.
As a consequence, we may \emph{tighten} any trail $(\phi_{n})_{n<\kappa}$ on a
path $\pi = (v_{n})_{n<\kappa}$ simply by omitting all $\phi_{i+1}$ from the
sequence for which $(\phi_{i},\phi_{i+1})$ belongs to the passive trail relation
$\ptrail_{v_{i},v_{i+1}}$.

\begin{definition}
Let $\tau = (\phi_{n})_{n<\kappa}$ be a trail on the path $\pi = 
(v_{n})_{n<\kappa}$ in some tableau $\bbT$.
Then the \emph{tightened} trail $\rdc{\tau}$ is obtained from $\tau$ by 
omitting all $\phi_{i+1}$ from $\tau$ for which $(\phi_{i},\phi_{i+1})$ belongs 
to the passive trail relation $\ptrail_{v_{i},v_{i+1}}$.
\end{definition}

It is not difficult to see that tightened trails are \emph{traces}, and that it
follows from Proposition~\ref{p:prog} that the tightening of an infinite 
trail is infinite.

\begin{definition}
Let $\tau = (\phi_{n})_{n<\om}$ be an infinite trail on the path $\pi = 
(v_{n})_{n<\om}$ in some tableau $\bbT$.
Then we call $\tau$ a \emph{$\nu$-trail} if its tightening $\rdc{\tau}$ is 
a $\nu$-trace.
\end{definition}

\subsection{Tableau games}

We are now ready to introduce the \emph{tableau game} $\game{\bbT}$ that we
associate with a tableau $\bbT$.
We will first give the formal definition of this game, and then provide an 
intuitive explanation; Appendix~\ref{sec:games} contains more information on
infinite games.
We shall refer to the two players of tableau games as \emph{\Prover} (female)
and \emph{\Refuter} (male).

\begin{definition}
\label{d:tg}
Given a tableau $\bbT = (V,E,\Phi,\tLab,v_I)$, the \emph{tableau game} 
$\game{\bbT}$ is the (initialised) board game $\game{\bbT} = 
(V,E,\Own,\SM,v_{I})$ defined as follows.
$\Own$ is a partial map that assigns a player to some positions in $V$; the
player $\Own(v)$ will then be called the \emph{owner} of the position $v$.
More specifically, \Refuter owns all positions that are labelled with 
one of the axioms, \AxLit or \AxTop, or with the rule \RuAnd; \Prover owns all 
position labelled with \RuMod; 
$O$ is undefined on all other positions.
In this context $v_{I}$ will be called the \emph{initial} or \emph{starting}
position of the game.

The set $\SM$ is the \emph{winning condition} of the game (for \Prover); it is 
defined as the set of infinite paths through the graph that carry a $\nu$-trail.
\end{definition}

A \emph{match} of the game consists of the two players moving a token from one 
position to another, starting at the initial position, and following the edge
relation $E$.
The owner of a position is responsible for moving the token from that position 
to an adjacent one (that is, an $E$-successor); in case this is impossible 
because the node has no $E$-successors, the player \emph{gets stuck} and 
immediately loses the match.
For instance, \Refuter loses as soon as the token reaches an axiomatic leaf 
labelled \AxLit or \AxTop; similarly, \Prover loses at any modal node without
successors.
If the token reaches a position that is not owned by a player, that is, a node
of $\bbT$ that is labelled with the proof rule \RuOr, \RuMu or \RuNu, the token
automatically moves to the unique successor of the position.
If neither player gets stuck, the resulting match is infinite; we declare 
\Prover to be its winner if the match, as an $E$-path, belongs to the set $\SM$,
that is, if it carries a $\nu$-trail.

Finally, we say that a position $v$ is a \emph{winning} position for a player 
$P$ if $P$ has a way of playing the game that guarantees they win the resulting 
match, no matter how $P$'s opponent plays. 
For a formalisation of these concepts we refer to the Appendix.

%

\begin{remark} \label{r:treestrat}
If $\bbT$ is \emph{tree-based} the notion of a strategy can be simplified. 
The point is that in this case finite matches can always be identified with
their last position, since any node in a tree corresponds to a unique path from 
the root to that node. 
It follows that any strategy in such a game is \emph{positional} (that is, the
move suggested to the player only depends on the current position). 
Moreover, we may identify a strategy for either player with a \emph{subtree} 
$S$ of $\bbT$ that contains the root of $\bbT$ and, for any node $s$ in $S$,
(1) it contains all successors of $s$ in case the player owns the position $s$,
while
(2) it contains exactly one successor of $s$ in case the player's opponent
owns the position $s$.
\end{remark}

The observations below are basically due to Niwi\'nski \& 
Walukiewicz~\cite{niwi:game96}.

\begin{theorem}[Determinacy]
Let $\bbT$ be a tableau for a sequent $\Phi$.
Then at any position of the tableaux game for $\bbT$ precisely one of the
players has a winning strategy.
\end{theorem}

\begin{proof}
The key observation underlying this theorem is that tableau games are 
\emph{regular}.
That is, using the labelling maps $\tLab$ and $\Sigma$ of a tableau $\bbT$, we
can find a finite set $C$, a colouring map $\ga: V \to C$, and an 
$\omega$-regular subset $L \sse C^{*}$ such that $\SM = \{ (v_{n})_{n\in\omega}
\in \InfP(\bbT) \mid (\ga(v_{n}))_{n\in\omega} \in L \}$.
The determinacy of $\game{\bbT}$ then follows by the classic result by B\"uchi
\& Landweber~\cite{buch:solv69} on the determinacy of regular games.
We skip further details of the proof, since it is rather similar to the analogous proof
in~\cite{niwi:game96}.
\end{proof}

For the Adequacy Theorem below we do provide a proof, since our proof is
somewhat different from the one by Niwi\'nski and Walukiewicz.


\begin{theorem}[Adequacy]
\label{t:adequacy}
Let $\bbT$ be a tableau for a sequent $\Phi$.
Then \Refuter (\Prover, respectively) has a winning strategy in $\game{\bbT}$ 
iff the formula $\bv\Phi$ is refutable (valid, respectively).
\end{theorem}

\begin{proof} 
Fix a sequent $\Phi$ and a tableau $\bbT$ for $\Phi$.
We will prove the following statement:
\begin{equation}
\label{eq:admain}
\text{\Refuter has a winning strategy in } \game{\bbT} 
\text{ iff }
\Phi \text{ is refutable}.
\end{equation}
The theorem follows from this by the determinacy of $\game{\bbT}$.

For the left to right implication of \eqref{eq:admain}, 
fix a tableau $\bbT = (V,E,\Phi,\tLab,v_I)$; it will be convenient to assume that
$\bbT$ is \emph{tree based}. 
This is without loss of generality: if the graph underlying $\bbT$ does not have
the shape of a tree, we may simply continue with its unravelling.

Let $f$ be a winning strategy for \Refuter in the game $\game{\bbT}$; recall that
we may think of $f$ as a subtree $\bbT_{f}$ of $\bbT$.
We will first define the pointed model in which the sequent $\Phi$ can be
refuted.
We define a \emph{state} to be a maximal path in $\bbT_{f}$ which does not 
contain any modal node, with the possible exception of its final node
$\last(\pi)$.
Note that by maximality, the first node of a state is either the root of $\bbT$
or else a successor of a modal node. 
Given a state $\pi = v_{0}\cdots v_{k}$ and a formula $\phi$, we say that $\phi$ 
\emph{occurs at} $\pi$, if $\phi \in \bigcup_{i} \Phi_{v_{i}}$.
We let $S_{f}$ denote the collection of all states, and define an accessibility 
relation $R_{f}$ on this set by putting $R_{f}\pi\rho$ iff the first node of 
$\rho$ is an $E$-successor of the last node of $\pi$.
Note that this can only happen if $\last(\pi)$ is modal.
Finally, we define the valuation $V_{f}$ by putting $V_{f}(p) \isdef \{ \pi \mid 
p \not\in \Phi_{\last(\pi)} \}$, and we set $\bbS_{f} \isdef (S_{f},R_{f},V_{f})$.

In the sequel we will need the following observation; we leave its proof as an
exercise.
\begin{claimfirst}
\label{cl:fair}
Let $\phi \in \Phi_{v_{j}}$ be a non-atomic formula, where $v_{j}$ is some node
on a finite path $\pi = (v_{i})_{i<k}$.
If $\pi$ is a state, then the formula is active at some node $v_{m}$ on $\pi$,
with $j \leq m < k$.
\end{claimfirst}

Now let $\pi_{0}$ be any state of which $\first(\pi_{0})$ is the root of $\bbT$.
We will prove that the pointed model $\bbS_{f},\pi_{0}$ refutes $\Phi$ by 
showing that 
\begin{equation}
\label{eq:adkey}
\text{for every }\phi \in \Phi, 
\text{ the position } (\phi,\pi_{0}) \text{ is winning for $\abel$ in }
\EG(\textstyle{\bv}\Phi,\bbS_{f}).
\end{equation}
To prove this, we will provide $\abel$ with a winning strategy in the 
evaluation game $\EG(\bv\Phi,\bbS_{f})@(\phi,\pi_{0})$, for each $\phi\in\Phi$.
Fix such a $\phi$, and abbreviate $\EG \isdef 
\EG(\bv\Phi,\bbS_{f})@(\phi,\pi_{0})$.
The key idea is that, while playing $\EG$, $\abel$ maintains a private match of
the tableau game $\game{\bbT}$, which is guided by \Refuter's winning
strategy $f$ and 
such that the current match of $\EG$ corresponds to a trail on this 
$\game{\bbT}$-match.

For some more detail on this link between the two games, let $\Sigma =
(\phi_{0},\pi_{0})(\phi_{1},\pi_{1})\cdots (\phi_{n},\pi_{n})$ be a partial
match of $\EG$.
We will say that a $\game{\bbT}$-match $\pi$ is \emph{linked to} $\Sigma$ if 
the following holds.
First, let $i_{1},\ldots,i_{k}$ be such that $0 < i_{1} < \cdots < i_{k} \leq n$ 
and $\phi_{i_{1}-1},\ldots,\phi_{i_{k}-1}$ is the sequence of all 
\emph{modal} formulas among $\phi_{0},\ldots,\phi_{n-1}$.
Then we require that $\pi$ is the concatenation 
$\pi = \pi_{0}\circ\cdots \circ \pi_{i_{k}-1}\circ \rho$, where 
each $\pi_{i}$ is a state and $\rho \sqsubseteq \pi_{n}$, 
and that the sequence $\phi_{0} \cdots\phi_{n}$ is the active tightening
of some trail on $\pi$.

Clearly then the matches that just consist of the initial positions of $\EG$ 
and $\game{\bbT}$, respectively, are linked.
Our proof of \eqref{eq:adkey} is based on the fact that $\abel$ has a strategy
that keeps such a link throughout the play of $\EG$.
As the crucial observation underlying this strategy, the following claim states
that $\abel$ can always maintain the link for one more round of the evaluation 
game.

\begin{claim}
\label{cl:tbcp1}
Let $\Sigma = (\phi_{0},\pi_{0})(\phi_{1},\pi_{1})\cdots (\phi_{n},\pi_{n})$ be
some $\EG$-match and let $\pi$ be an $f$-guided $\game{\bbT}$-match that is 
linked to $\Sigma$.
Then the following hold.
\begin{urlist}
\item
If $(\phi_{n},\pi_{n})$ is a position for $\abel$ in $\EG$, then he has a move
$(\phi_{n+1},\pi_{n+1})$ such that some $f$-guided extension $\pi'$ of $\pi$ is 
linked to $\Sigma\cdot (\phi_{n+1},\pi_{n+1})$.
\item 
If $(\phi_{n},\pi_{n})$ is not a position for $\abel$ in $\EG$, then for any
move $(\phi_{n+1},\pi_{n+1})$ there is some $f$-guided extension $\pi'$ of $\pi$ 
that is linked to $\Sigma\cdot (\phi_{n+1},\pi_{n+1})$.
\end{urlist}
\end{claim}

\begin{pfclaim}
Let $\Sigma$ and $\pi$ be as in the formulation of the claim.
Then $\pi = \pi_{0}\circ\cdots \circ \pi_{i_{k}-1}\circ \rho$, where
$\rho \sqsubseteq \pi_{n}$ and
$i_{1},\ldots,i_{k}$ are such that $0 < i_{1} < \cdots < i_{k} \leq n$ 
and $\phi_{i_{1}-1},\ldots,\phi_{i_{k}-1}$ is the sequence of all 
\emph{modal} formulas among $\phi_{0},\ldots,\phi_{n-1}$.
Furthermore $(\phi_{i})_{i\leq n} = \rdc{\tau}$ for some trail $\tau$ on $\pi$.
Write $\rho = v_{i_{k}}\cdots v_{l}$, then $\rho = \pi_{n}$ iff $v_{l}$ is
modal.

We prove the claim by a case distinction on the nature of $\phi_{n}$.
Note that $\phi_{n} \in \Phi_{v_{l}}$, and that by Claim~\ref{cl:fair} there is
a node $v_{i}$ on the path $\pi_{n}$ such that $i_{k} \leq i$ and $\phi_{n}$ is
active at $v_{i}$.

\begin{description}

\item[Case $\phi_{n} = \psi_{0} \land \psi_{1}$] for some formulas $\psi_{0},
\psi_{1}$.
The position $(\phi_{n},\pi_{n})$ in $\EG$ then belongs to $\abel$. 
As $\psi_{0} \land \psi_{1}$ is the active formula at the node $v_{i}$ in $\bbT$, 
this means that $\tLab_{v_{i}} = \RuAnd$, so that $v_{i}$, as a position of 
$\game{\bbT}$, belongs to \Refuter.
This means that in $\EG$, $\abel$ may pick the formula $\psi_{j}$ which is 
associated with the successor $v_{i+1}$ of $v_{i}$ on $\pi_{n}$.
Note that, since $\pi_{n}$ is part of the $f$-guided match $\pi$, this successor
is the one that is picked by \Refuter in $\game{\bbT}$ at the position $v_{i}$
in the match $\pi$.

We define $\Sigma' \isdef \Sigma \cdot (\psi_{j},\pi_{n})$, 
$\pi' \isdef \pi \cdot v_{l+1}\cdots v_{i}v_{i+1}$, and 
$\tau' \isdef \tau \cdot \phi_{n} \cdots \phi_{n} \cdot \psi_{j}$. 
It is then immediate by the definitions that 
$\pi' = \pi_{0}\circ\cdots \circ \pi_{i_{k}-1}\circ \rho'$, where 
$\rho' \isdef \rho \cdot v_{l+1}\cdots v_{i}\cdot v_{i+1}$; 
Observe that since $v_{i+1}$ lies on the path $\pi_{n}$, 
we still have $\rho' \sqsubseteq \pi_{n}$.
Furthermore, it is obvious that $\tau'$ extends $\tau$ via a number of passive
trail steps, i.e., where $\phi_{n}$ is not active, until $\phi_{n}$ is the active
formula at $v_{i}$; from this it easily follows that $\rdc{\tau'} = \rdc{\tau} 
\cdot \psi_{j} = \phi_{0}\cdots\phi_{n}\cdot\psi_{j}$.
Furthermore, since the position $v_{i+1}$ of $v_{i}$ lies on the path $\pi_{n}$, 
it was picked by \Refuter's winning strategy in $\game{\bbT}$ at the position 
$v_{i}$ in the match $\pi$; this means that the match $\pi'$ is still $f$-guided.

\item[Case $\phi_{n} = \psi_{0} \lor \psi_{1}$] for some formulas $\psi_{0},
\psi_{1}$.
The position $(\phi_{n},\pi_{n})$ in $\EG$ then belongs to $\eloi$, so suppose
that she continues the match $\Sigma$ by picking the formula $\psi_{j}$.
In this case we have $\tLab_{v_{i}} = \RuOr$, so that $v_{i}$ has a unique 
successor $v_{i+1}$ which features both $\psi_{0}$ and $\psi_{1}$ in its label 
set.

This means that if we define $\Sigma' \isdef \Sigma \cdot (\psi_{j},\pi_{n})$,
$\pi' \isdef \pi \cdot v_{l+1}\cdots v_{i}v_{i+1}$ and 
$\tau' \isdef \tau \cdot \phi_{n} \cdots \phi_{n} \cdot \psi_{j}$,
it is not hard to see that $\Sigma'$ and $\pi'$ are linked, with $\tau'$ the 
witnessing trail on $\pi'$.

\item[Case $\phi_{n} = \eta x. \psi$] for some binder $\eta$, 
variable $x$ and formula $\psi$.
The match $\Sigma$ is then continued with the automatic move 
$(\psi[\eta x\, \psi/x],\pi_{n+1})$.
This case is in fact very similar to the one where $\phi$ is a disjunction, so
we omit the details.

\item[Case $\phi_{n} = \Box\psi$] for some formula $\psi$.
Then the position $(\phi_{n},\pi_{n})$ belongs to $\abel$: he has to come up with
an $R_{f}$-successor of the state $\pi_{n}$.
Since $\Box\psi$ is active in it, the node $v_{i}$ must be modal, in the sense
that $\tLab_{v_{i}} = \RuMod$.
By the definition of a state this can only be the case if $v_{i}$ is the last
node on the path/state $\pi_{n}$; recall that in this case we have $\rho =
\pi_{n}$.
Let $u \in E[v_{i}]$ be the successor of $v$ associated with $\psi$,
and let $\pi_{n+1}$ be any state with $\first(\pi_{n+1}) = u$.
It follows by definition of $R_{f}$ that $\pi_{n+1}$ is a successor of 
$\pi_{n}$ in the model $\bbS_{f}$.
This $\pi_{n+1}$ will then be $\abel$'s (legitimate) pick in $\EG$ at the 
position $(\Box\psi,\pi_{n+1})$.

Define $\Sigma' \isdef \Sigma \cdot (\psi,\pi_{n+1})$, 
$\pi' \isdef \pi \cdot v_{l+1} \cdots v_{i} u$
and $\tau' \isdef \tau \cdot \phi_{n} \cdots \phi_{n} \cdot \psi$.
Then we find that 
$\pi' = \pi_{0}\circ\cdots \circ \pi_{i_{k}-1}\circ \rho \circ \rho'$,
where $\rho'$ is the one-position path $u$.
Clearly then $\rho' \sqsubseteq \pi_{n+1}$.
Furthermore, it is easy to verify that $\rdc{\tau'} = \rdc{\tau} \cdot \psi 
= \phi_{0}\cdots\phi_{n}\psi$.
This means that $\Si'$ and $\pi'$ are linked, as required.

\item[Case $\phi_{n} = \dia\psi$] for some formula $\psi$.
As in the previous case this means that $v_{i}$ is a modal node, and $v_{i} = 
\last(\pi_{n})$.
However, the position $(\phi_{n},\pi_{n})$ now belongs to $\eloi$; suppose that 
she picks an $R_{f}$-successor $\pi_{n+1}$ of $\pi_{n}$.
Let $u \isdef \first(\pi_{n+1})$, then it follows from the definition of $R_{f}$
that $u$ is an $E$-successor of $v_{i}$.
As such, $u$ is a legitimate move for \Prover in the tableau game.

It then follows, exactly as in the previous case, that 
$\pi' \isdef \pi \cdot v_{l+1} \cdots v_{i} u$ is linked to 
$\Sigma' \isdef \Sigma \cdot (\psi,\pi_{n+1})$.
\end{description}
This finishes the proof of the claim.
\end{pfclaim}

On the basis of Claim~\ref{cl:tbcp1}, we may assume that $\abel$ indeed uses a 
strategy $f'$ that keeps a link between the $\EG$-match and his privately played 
$f$-guided $\game{\bbT}$-match.
We claim that $f'$ is actually a winning strategy for him.
To prove this, consider a \emph{full} $f'$-guided match $\Sigma$; we claim that
$\abel$ must be the winner of $\Sigma$.
This is easy to see if $\Sigma$ is finite, since it follows by the first item of
the Claim that playing $f'$, $\abel$ will never get stuck.

This leaves the case where $\Sigma$ is infinite.
Let $\Sigma = (\phi_{n},s_{n})_{n<\omega}$; it easily follows from
Claim~\ref{cl:tbcp1} that there must be an infinite $f$-guided 
$\game{\bbT}$-match $\pi$, such that the sequence $(\phi_{n})_{n<\omega}$ is the
tightening of some trail on $\pi$.
Since $\pi$ is guided by \Refuter's winning strategy $f$ this means that all 
of its trails are $\mu$-trails; but then obviously $(\phi_{n})_{n<\omega}$ is 
a $\mu$-trace, meaning that $\abel$ is the winner of $\Sigma$ indeed.
\medskip

The implication from left to right in \eqref{eq:admain} is proved along similar
lines, so we permit ourselves to be a bit more sketchy.
Assume that $\Phi$ is refuted in some pointed model $(\bbS,s)$.
Then by the adequacy of the game semantics for the modal $\mu$-calculus, $\abel$
has a winning strategy $f$ in the evaluation game $\EG(\bv\phi,\bbS)$ initialised
at position $(\bv\Phi,s)$.
Without loss of generality we may assume $f$ to be \emph{positional}, i.e., it 
only depends on the current position of the match.

The idea of the proof is now simple: while playing $\game{\bbT}$, \Refuter will
make sure that, where $\pi = v_{0}\cdots v_{k}$ is the current match,
every formula in $\Phi_{v_{k}}$ is the endpoint of some trail, and 
every trail $\tau$ on $\pi$ is such that its tightened trace $\rdc{\tau}$ is 
the projection of an $f$-guided match of $\EG(\bv\phi,\bbS)$ initialised at 
position $(\phi,s)$ for some $\phi \in \Phi$.
To show that \Refuter can maintain this condition for the full duration of the
match, it suffices to prove that he can keep it during one single round.
For this proof we make a case distinction, as to the rule applied at the last
node $v_{k}$ of the partial $\game{\bbT}$-match $\pi = v_{0}\cdots v_{k}$.
The proof details are fairly routine, so we confine ourselves to one case,
leaving the other cases as an exercise.

Assume, then, that $v_{k}$ is a conjunctive node, that is, $\tLab_{v_{k}} = 
\RuAnd$.
This node belongs to \Refuter, so as his move he has to pick an $E$-successor of
$v_{k}$.
The active formula at $v_{k}$ is some conjunction, say, $\psi_{0}\land\psi_{1}
\in \Phi_{v_{k}}$.
By the inductive assumption there is some trail $\tau = \phi_{0}\cdots \phi_{k}$
on $\pi$ such that $\phi_{k} = \psi_{0}\land\psi_{1}$, and there is some 
$f$-guided $\EG$-match of which $\rdc{\tau}$ is the projection, i.e., it is of 
the form $\Sigma = (\phi_{0},s_{0})\cdots(\phi_{k},s_{k})$.
Now observe that in $\EG$, the last position of this match, viz., $(\phi_{k},
s_{k}) = (\psi_{0}\land \psi_{1},s_{k})$, belongs to $\abel$.
Assume that his winning strategy $f$ tells him to pick the formula $\psi_{j}$ 
at this position, then in the tableau game, at the position $v_{k}$, \Refuter 
will pick the $E$-successor $u_{j}$ of $v_{k}$ that is associated with the 
conjunct $\psi_{j}$.
That is, he extends the match $\pi$ to $\pi' \isdef \pi\cdot u_{j}$.

To see that \Refuter has maintained the invariant, consider an arbitrary trail 
on $\pi'$; clearly such a trail is of the form $\sigma' = \sigma \cdot \psi$, 
for some trail $\sigma$ on $\pi$, and some formula $\psi \in \Phi_{u_{j}}$.
It is not hard to see that either $\last(\sigma) = \psi_{0}\land\psi_{1}$ 
and $\psi=\psi_{j}$, or else $\last(\sigma) = \psi$.
In the first case $\rdc{\sigma'}$ is the match $(\phi_{0},s_{0})\cdots
(\phi_{k},s_{k})\cdot (\psi_{j},s_{k})$; in the second case we find that 
$\rdc{\sigma'} = \rdc{\sigma}$ so that for the associated $f$-guided 
$\EG$-match we can take any such match that we inductively know to exist for
$\sigma$.
\end{proof}

\begin{corollary} \label{cor:invariant}
Let $\bbT$ and $\bbT'$ be two tableaux for the same sequent. 
Then \Prover has a winning strategy in $\game{\bbT}$ iff she has a winning 
strategy in $\game{\bbT'}$.
\end{corollary}

\section{Soundness}
\label{s:soundness}

In this section we show that our proof systems are sound, meaning that
any provable formula is valid. Because of the adequacy of the tableau game
that was established in Theorem~\ref{t:adequacy} it suffices to show that 
for every provable formula \Prover has a winning strategy in some tableau
for this formula. Moreover, we only need to consider proofs in \Focusinf
because by Theorem~\ref{t:same} every formula that is provable in \Focus
is also provable in \Focusinf.

\begin{theorem} 
\label{t:soundness}
Let $\Phi$ be some sequent.
If $\Phi$ is provable in \Focusinf then there is some tableau $\bbT$ for
$\Phi$ such that \Prover has a winning strategy in $\game{\bbT}$.
\end{theorem}

We will prove the soundness theorem by transforming a thin and progressive 
$\Focusinf$-proof of $\Phi$ into a winning strategy for \Prover in the tableau 
game associated with some tableau for $\Phi$.
To make a connection between proofs and tableaux more tight, we first consider
the notion of an (annotated) trail in the setting of $\Focusinf$-proofs.

\begin{definition} \label{d:trails}
Let $\Pi = (T,P,\Si,\pLab)$ be a thin and progressive proof in \Focusinf. 
For all nodes $u,v \in T$ such that $P u v$ we define the \emph{active trail 
relation} $\atrail_{u,v} \subseteq \Si_u \times \Si_v$ and the \emph{passive 
trail relation} $\ptrail_{u,v} \subseteq \Si_u \times \Si_v$ by a case 
distinction depending on the rule that is applied at $u$.
Here we use the notation $\De_{S} \isdef \{(s,s) \mid s \in S\}$, for 
any set $S$.

\textit{Case $\pLab(u) = \RuOr$:} 
Then $\Sigma_u = \{(\phi \lor \psi)^a\} \uplus \Gamma$ and $\Sigma_v = 
\{\phi^a,\psi^a\} \cup \Gamma$, for some annotated sequent $\Gamma$. 
We define 
$\atrail_{u,v} \isdef \{((\phi \lor \psi)^a,\phi^a),((\phi \lor \psi)^a, \psi^a)\}$
and $\ptrail_{u,v} \isdef \Delta_\Gamma$.

\textit{Case $\pLab(u) = \RuAnd$:} 
In this case $\Sigma_u = \{(\phi_{0} \land \phi_{1})^{a}\} \uplus \Gamma$ and 
$\Sigma_v = \{\phi_{i}^a\} \cup \Gamma$ for some $i \in \{0,1\}$ and some
annotated sequent $\Gamma$.
We set $\atrail_{u,v} \isdef \{((\phi_{0} \land 
  \phi_{1})^a,\phi_{i}^a)\}$ and 
$\ptrail_{u,v} \isdef \Delta_\Gamma$.

\textit{Case $\pLab(u) = \RuMu$:} 
Then $\Sigma_u = \{\mu x . \phi^a\} \uplus \Gamma$ and 
$\Sigma_v = \{\phi[\mu x . \phi / x]^u\} \cup \Gamma$ for some sequent $\Gamma$.
We define $\atrail_{u,v} \isdef \{(\mu x . \phi^a,\phi[\mu x . \phi / x]^f)\}$
and $\ptrail_{u,v} \isdef \Delta_\Gamma$.

\textit{Case $\pLab(u) = \RuNu$:} 
Then $\Sigma_u = \{\nu x . \phi^a\} \uplus \Gamma$ and $\Sigma_v = 
\{\phi[\nu x . \phi / x]^a\} \cup
\Gamma$ for some sequent $\Gamma$. 
We define $\atrail_{u,v} \isdef \{(\nu x . \phi^a,\phi[\nu x . \phi / x]^a)\}$
and $\ptrail_{u,v} \isdef \Delta_\Gamma$.

\textit{Case $\pLab(u) = \RuBox$:} Then $\Sigma_u = \{\Box \phi^a\}
\cup \dia \Gamma$ and $\Sigma_v = \{\phi^a\} \cup \Gamma$ for
some annotated sequent $\Gamma$. 
We define $\atrail_{u,v} =
\{(\Box \phi^a,\phi^a)\} \cup \{(\dia \psi^b,\psi^b) \mid
\psi^b \in \Sigma\}$ and $\ptrail_{u,v} = \nada$.

\textit{Case $\pLab(u) = \RuWeak$:} 
In this case $\Sigma_u = \Sigma_v \uplus \{ \phi^{a} \}$ and we set
$\atrail_{u,v} \isdef \nada$ and $\ptrail_{u,v} \isdef\Delta_{\Sigma_v}$.

\textit{Case $\pLab(u) = \RuFocus$:} 
Then $\Sigma_u = \{\phi^u\} \cup \Ga$ and $\Sigma_v = \{\phi^f\} \cup \Ga$
for some annotated sequent $\Gamma$. 
We define $\atrail_{u,v} = \nada$ and
$\ptrail_{u,v} = \{(\phi^u,\phi^f)\} \cup \De_{\Ga}$.

\textit{Case $\pLab(u) = \RuUnfocus$:} 
Then $\Sigma_u = \{\phi^f\} \cup \Ga$ and $\Sigma_v = \{\phi^u\} \cup \Ga$
for some annotated sequent $\Gamma$. 
We define $\atrail_{u,v} = \nada$ and
$\ptrail_{u,v} = \{(\phi^f,\phi^u)\} \cup \De_{\Ga}$.

We also define the \emph{general trail relation} $\gtrail_{u,v} \isdef
\atrail_{u,v} \cup \ptrail_{u,v}$ for all nodes $u$ and $v$ with $Puv$.
\end{definition}

Note that in the case distinction of Definition~\ref{d:trails}, it is not 
possible that $u$ is an axiomatic leaf since it has a successor, and it is not
possible that $\pLab(u) \in \Tokens \cup \{\RuDischarge{\dx} \mid \dx \in 
\Tokens\}$ since $\Pi$ is a proof in \Focusinf.

We extend the trail relation $\gtrail_{u,v}$ to any two nodes such that $u$ is 
an ancestor of $v$ in the underlying proof tree.

\begin{definition}
Let $u,v$ be nodes of a proof tree $\Pi = (T,P,\Si,\pLab)$ such that $P^{*}uv$.
The relation $\gtrail_{u,v}$ is defined inductively such that $\gtrail_{u,u}
\isdef \De_{\Si_{u}}$, and if $Puw$ and $P^{*}wv$ then $\gtrail_{u,v} \isdef
\gtrail_{u,w} \comp \gtrail_{w,v}$, where $\comp$ denotes relational composition. 
\end{definition}

As in the case of tableaux, we will be specifically interested in infinite 
trails.

\begin{definition}
An \emph{(annotated) trail} on an infinite path $\alpha = (v_{n})_{n\in\om}$ in 
a \Focusinf-proof $\Pi$ is an infinite sequence $\tau = (\phi_n^{a_n})_{n\in\om}$ 
of annotated formulas such that $(\phi_i^{a_i} \phi_{i + 1}^{a_{i+1}}) \in 
\gtrail_{v_i,v_{i+1}}$ for all $i \in \omega$.
The tightening of such an annotated trail is defined exactly as in the case of 
plain trails.
An infinite trail $\tau$ is an \emph{$\eta$-trail}, for $\eta \in \{ \mu,\nu \}$ 
if its tightening $\rdc{\tau}$ is an $\eta$-trace.
\end{definition}

The central observation about the focus mechanism is that it enforces every 
infinite branch in a thin and progressive \Focusinf-proofs to contain a 
$\nu$-trail.

\begin{proposition} \label{p:nu-trail}
Every infinite branch in a thin and progressive \Focusinf-proof carries 
a $\nu$-trail.
\end{proposition}

\begin{proof}
Consider an infinite branch $\alpha = (v_{n})_{n\in\om}$ in some 
$\Focusinf$-proof $\Pi = (T,P,\Si,\pLab)$.
Then $\al$ is successful by assumption, so that we may fix a $k$ such that
for every $j \geq k$, the sequent $\Si_{v_{j}}$ contains a formula in focus,
and $\pLab(v_{j})$ is not a focus rule.

We claim that 
\begin{equation} \label{e:predecessor}
\text{for every $j \geq k$ and $\psi^f \in \Sigma_{v_{j+1}}$ 
there is some $\chi^f \in \Si_{v_{j}}$ such that $(\chi^{f},\psi^{f}) \in 
\gtrail_{v_{j},v_{j+1}}$}.
\end{equation}
To see this, let $j \geq k$ and  $\psi^f \in \Sigma_{v_{j+1}}$.
It is obvious that there is some annotated formula $\chi^a \in \Si_{v_{j}}$
with $(\chi^{a},\psi^{f}) \in \gtrail_{v_{j},v_{j+1}}$.
The key observation is now that in fact $a = f$, and this holds because 
the only way that we could have $(\chi^{u},\psi^{f}) \in \gtrail_{v_{j},v_{j+1}}$
is if we applied the focus rule at $v_{j}$, which would contradict 
our assumption on the nodes $v_{j}$ for $j \geq k$.

Now consider the graph $(V,E)$ where 
\[
V \isdef \{ (j,\phi) \mid k \leq j < \omega \text{ and } 
\phi^{f} \in \Si_{\al_{j}} \},
\] 
and 
\[
E \isdef \big\{ \big( (j,\phi),(j+1,\psi) \big) \mid (\phi^{f},\psi^{f}) \in 
\gtrail_{v_{j},v_{j+1}} \big\}
\]
This graph is directed, acyclic, infinite and finitely branching.
Furthermore, it follows by \eqref{e:predecessor} that every node $(j,\phi)$ is 
reachable in $(V,E)$ from some node $(k,\psi)$.
Then by a (variation of) K\"{o}nig's Lemma there is an infinite path 
$(n,\phi_{n}^f)_{n\in\om}$ in this graph. 
The induced sequence $\tau \isdef (\phi_{n}^f)_{n\in\om}$ is a trail on $\al$
because the formulas are related by the trail relation. 
By guardedness, $\tau$ must be either a $\mu$-trail or a $\nu$-trail.
But $\tau$ cannot feature infinitely many $\mu$-formulas, since it is not 
possible to unravel a $\mu$-formula $\phi_{j}^{f}$ and end up with 
a formula of the form $\phi_{j + 1}^f$, simply because the rule \RuMu attaches the label $u$ to the 
unravelling of $\phi_{j}$.
This means that $\tau$ cannot be a $\mu$-trail, and hence it must be 
a $\nu$-trail.
\end{proof}

\begin{proofof}{Theorem~\ref{t:soundness}}
Let $\Pi = (T,P,\Si,\pLab)$ be a \Focusinf-proof for $\Phi^f$. 
By Theorem~\ref{t:tpp} we may assume without loss of generality that $\Pi$ is 
thin and progressive.
We are going to construct a tableau $\bbT = (V,E,\Phi,\tLab,v_I)$ and a winning
strategy for \Prover in $\game{\bbT}$.
Our construction will be such that $(V,E)$ is a potentially infinite tree, of 
which the winning strategy $S \subseteq V$ for \Prover is a 
subtree, as in Remark~\ref{r:treestrat}.

The construction of $\bbT$ and $S$ proceeds via an induction that starts from 
the root and in every step adds children to one of the nodes in the subtree
$S$ that is not yet an axiom. 
Nodes of $\bbT$ that are not in $S$ are always immediately completely extended
using Proposition~\ref{p:tableau exists}. 
Thus, they do not have to be treated in the inductive construction. 
The construction of $S$ is guided by the structure of $\Pi$.

In addition to the tableau $\bbT$ we will construct a function $g : S \to T$ 
mapping those nodes of $\bbT$ that belong to the strategy $S$ to nodes of $\Pi$.
This function will satisfy the following three conditions, which will allow us 
to lift the $\nu$-trails from $\Pi$ to $S$:
\begin{enumerate}
\item \label{i:order preserving} If $Euv$ then $P^* g(u) g(v)$.
\item \label{i:tra} 
The sequent $\Si_{g(l)}$ is thin, and $\uls{\Sigma}_{g(u)} \sse \Phi_{u}$.
\item \label{i:trb}
If $Euv$ and $(\psi^b,\phi^a) \in \gtrail^{\Pi}_{g(u),g(v)}$ then 
$(\psi,\phi) \in \gtrail^{\bbT}_{u,v}$.
\end{enumerate}

We now describe the iterative construction of the approximating objects $\bbT_i$,
$S_i$ and $g_i$ for all $i \in \omega$, which in the limit will yield $\bbT$, 
$S$ and $g$. 
Each $\bbT_i$ will be a \emph{pre-tableau}, that is, an object as defined in 
Definition~\ref{d:tableau}, except that we do not require the rule labelling to 
be defined for every leaf of the tree.
Leaves without labels will be called \emph{undetermined}, and the basic idea 
underlying the construction is that each step will take care of one undetermined leaf.
We will make sure that in each step $i$ of the construction, the entities
$\bbT_i$, $S_i$ and $g_i$ satisfy 
the conditions \ref{i:order preserving}, \ref{i:tra} and \ref{i:trb},
and moreover ensure that all undetermined leaves of $\bbT_i$ belong to $S_i$. 
It is easy to see that then also $S$ and $g$ satisfy these conditions.
\medskip

In the base case we let $\bbT_0$ be the node $v_I$ labelled with just $\Phi$
at the root of the tableau. We let $g_0(v_0)$ be the root of the proof $\Pi$.
The strategy $S_0$ just contains the node $v_0$.
\medskip

In the inductive step we assume that we have already constructed a pre-tableau 
$\bbT_i$, a subtree $S_i$ corresponding to \Prover's strategy and a function 
$g_i : S_i \to T$ satisfying the above conditions \ref{i:order preserving} 
-- \ref{i:trb}.

To extend these objects further we fix an undetermined leaf $l$ of $S_i$.
We may choose $l$ such that its distance to the root of $\bbT_i$ is minimal 
among all the undetermined leaves of $\bbT_i$. 
This will guarantee that every undetermined leaf gets treated eventually and 
thus ensure that the trees $S$ and $T$ in the limit do not contain any 
undetermined leaves. 
We  distinguish cases depending on the rule that is applied in $\Pi$ at 
$g_i(l)$.
\smallskip

\textit{Case $\pLab(g_i(l)) = \AxLit$ or $\pLab(g_i(l)) = \AxTop$:} 
In this case we may simply label the node $l$ with the corresponding axiom, 
while apart from this, we do not change $\bbT_{i}$, $S_{i}$ of $g_{i}$.
Note that $l$ will remain an (axiomatic) leaf of the tableau $\bbT$.
\smallskip

\textit{Case $\pLab(g_i(l)) = \RuOr$:} 
If the rule applied at $g_i(l)$ is \RuOr with principal formula, say,  
$(\phi\lor\psi)^{a}$, then this application of $\RuOr$ is followed by a 
(possibly empty) series of applications of weakening until a descendant $t$
of $g_{i(l)}$ is reached that is labeled with a thin sequent.

By condition~\ref{i:tra} the formula $\phi \lor \psi$ occurs at $l$, as it 
occurs in $g_i(l)$, so that we may label $l$ with the disjunction rule as 
well.
We extend $\bbT_i$, $S_i$ and $g_i$ accordingly, meaning that $T_{i + 1}$ is
$T_i$ extended with one node $v$ that is labelled with the premise of the 
application of the disjunction rule, $S_{i + 1}$ is $S_i$ extended to contain 
$v$ and $g_{i + 1}$ is just like $g_i$ but additionally maps $v$ to $t$. 
It is easy to check that with these definitions, the conditions
\ref{i:order preserving} -- \ref{i:trb} are satisfied. 
For condion \ref{i:tra} we need the fact that the formula 
$(\phi\lor\psi)^{\ol{a}}$ does not occur as a side formula in $\Si_{g_{i}(l)}$
since the latter sequent is thin, so that, as $\Pi$ is also progressive, the 
formula $\phi\lor\psi$ does not appear in the premisse of the rule at all, and
hence not in $\Si_{t}$ either.
\smallskip

\textit{Case $\pLab(g_i(l)) = \RuAnd$:} 
In the case where \RuAnd is applied at $g_i(l)$ with principal formula 
$(\phi \land \psi)^a$ it follows that $g_i(l)$ has a child $s_\phi$ for $\phi^a$
and a child $s_\psi$ for $\psi^a$, and that these nodes have thin descendants
$t_{\phi}$ and $t_{\psi}$, respectively, each of which is reached by a possibly
empty series of weakenings.

By condition~\ref{i:tra} it follows that $\phi \land \psi \in \Phi_l$. 
We can then apply the conjunction rule at $l$ to the formula $\phi \land \psi$
and obtain two new premises $v_\phi$ and $v_\psi$ for each of the conjuncts.
$\bbT_{i + 1}$ is defined to extend $\bbT_i$ with these additional two children. 
We let $S_{i + 1}$ include both nodes $v_\phi$ and $v_\psi$ as the 
conjunction rule belongs to \Refuter in the tableaux game. 
Moreover, $g_{i + 1}$ is the same as $g_i$ on the domain of $g_i$, while it 
maps $v_\phi$ to $t_\phi$ and $v_\psi$ to $t_\psi$. 
It is easy to check that the conditions \ref{i:order preserving} -- \ref{i:trb} 
are satisfied, where for condition condion \ref{i:tra} we use the thinness and 
progressivity of $\Pi$ as in the case for $\RuOr$.
\smallskip

\textit{Case $\pLab(g_i(l)) = \RuBox$:} 
We want to match this application of \RuBox in $\Pi$ with an application of the 
rule \RuMod in the tableau system. 
To make this work, however, two difficulties need to be addressed.
Let $s$ be the successor of $g_{i}(l)$ in $\Pi$, and, as before, assume that 
$\RuBox$ is followed by a possibly empty series of weakenings until a descendant
$t$ of $s$ is reached that is labelled with a thin sequent.

The first issue is that to apply the rule \RuMod in the tableau system, every
formula in the consequent must be either atomic or modal, whereas the sequent
$\Phi_{l}$ may contain boolean or fixpoint formulas.
The second difficulty is that the rule \RuBox in the focus proof system has 
only one premise, whereas the tableau rule \RuMod has one premise for each box 
formula in the conclusion.

To address the first difficulty we step by step apply the Boolean rules (\RuOr 
and \RuAnd) to break down all the Boolean formulas in $\Phi_l$ and the fixpoint 
rules (\RuMu and \RuNu) to unfold all fixpoint formulas.
Because the rule \RuAnd is branching this process generates a subtree $\bbT_{l}$
at $l$ such that all leaves of $\bbT_{l}$ contain literals and modal formulas 
only.
Moreover, any modal formula from $\Phi_l$ is still present in $\Phi_m$, for any
such leaf $m$, because modal formulas are not affected by the application 
of Boolean or fixpoint rules.

We add all nodes of $\bbT_{l}$ to the strategy $S$, and we define $g_{i + 1}(u) 
\isdef g_i(u)$ for any $u$ in this subtree.
To see that this does not violate condition~\ref{i:tra} or \ref{i:trb}, note 
that all formulas in $\Sigma_{g_i(l)}$ are modal and so, as we saw, remain
present throughout the subtree.

Note that $\bbT_{l}$ may contain leaves $m$ such that $\Phi_{m}$ does not meet
the side condition (\dag) of the modal rule $\RuMod$; this means, however, that 
$\Phi_{m}$ is axiomatic, so that we may label such a leaf $m$ with either 
$\AxLit$ or $\AxTop$.
We then want to expand any remaining leaf in $\bbT_{l}$ by applying the modal 
rule \RuMod.
To see how this is done, fix such a leaf $m$.
Applying the modal rule of the tableau system at $m$ generates a new child
$n_\chi$ for every box formula $\Box \chi \in \Phi_l$. 
At this point we have to solve our second difficulty mentioned above, which is
to select one child $n_\chi$ to add into $S_{i + 1}$ and finish the construction
of the tableau for all other children.

To select the appropriate child of $m$, consider the unique box formula $\Box 
\phi$ such that $\Box \phi^a \in \Sigma_{g_i(l)}$ for some $a \in \{f,u\}$
--- such a formula exists because \RuBox is applied at $g_i(l)$. 
By condition~\ref{i:tra} we then have $\Box \phi \in \Phi_l$ and from this it 
follows, as we saw already, that $\Box \phi \in \Phi_m$.
We select the child $n_\phi$ of $m$ to be added to $S_{i + 1}$ and set 
$g_{i+1}(n_\phi) = t$, where $t$ is defined before. 
It is not hard to see that this definition satisfies the conditions~\ref{i:tra} 
and~\ref{i:trb}, because all diamond formulas in $\Sigma_{g_i(l)}$ are also in 
$\Phi_l$ and thus still present in $\Phi_m$.

We still need to deal with the other children of $m$, since these are still 
undetermined but not in $S_{i + 1}$, something we do not allow in our iterative
construction.
To solve this issue we simply use Proposition~\ref{p:tableau exists} to obtain
a new tree-shaped tableau $\bbT_k$ for any such child $k$ of $m$ with $k \neq 
n_\phi$. 
For the definition of $\bbT_{i + 1}$ we append $\bbT_k$ above the child $k$. 
Hence, the only undetermined leaf that is left above $m$ in $\bbT_{i+1}$ is the 
node $n_\phi$, which belongs to $S_{i + 1}$.
\smallskip

\textit{Case $\pLab(g_i(l)) = \RuMu$ or $\pLab(g_i(l)) = \RuNu$:} 
The case for the fixpoint rules is similar to the case for $\RuOr$, we just 
apply the corresponding fixpoint rule on the tableau side.
\smallskip

\textit{Case $\pLab(g_i(l)) = \RuWeak$:} 
Note that in this case the sequent $\Sigma_{t}$, associated with the successor 
node $t$ of $g(l)$, being the premise of an application of the weakening rule, 
is a (proper) subset of the consequent sequent $\Sigma_{g_i(l)}$.
In this case we simply define $T_{i+1} \isdef \bbT_i$ and $S_{i+1} \isdef
S_i$, but we modify $g_i$ so that $g_{i+1} : S_{i+1} \to T$ maps 
$g_{i+1}(l) = t$ and $g_{i+1}(k) = g_i(k)$ for all $k \neq l$. 
This clearly satisfies condition~\ref{i:order preserving}. 
To see that it satisfies the other two conditions we use the facts that 
$\Sigma_t \subseteq \Sigma_{g_i(l)}$, and that the trail relation for 
the weakening rule is trivial.

However, after applying this step we still have that $l$ is an undetermined leaf
of $\bbT_{i + 1}$. Thus the construction does not really make progress
in this step and one might worry that not all undetermined leaves get
eventually. We address this matter further below.
\smallskip

\textit{Case $\pLab(g_i(l)) = \RuFocus$:} The case for the focus change rule
\RuFocus is analogous to the previous case for the weakening rule
\RuWeak. The fact that the annotations of formulas change has no bearing
on the conditions.
\smallskip

We now address the problem that in the cases for \RuWeak and \RuFocus,
we do not extend $\bbT_i$ at its undetermined leaf $l$. Thus, without further
arguments it would seem possible that the construction loops through these
cases without ever making progress at the undetermined leaf $l$. 
To see that this can not happen note first that in each of these cases we are 
moving on in the proof $\Pi$ in the sense that $g_{i + 1}(l) \neq g_i(l)$ and
$(g_i(l),g_{i + 1}(l)) \in P$. 
Thus, if we would never make progress at $l$ this means that we would need to
follow an infinite path in $\Pi$ of which every node is labelled with either
\RuWeak or with \RuFocus. 
However, this would contradict Proposition~\ref{p:nu-trail} because every 
infinite branch in $\Pi$ is successful.
\medskip

It remains to be seen that $S$ is a winning strategy for \Prover. 
It is clear that \Prover wins all finite matches that are played according to
$S$ because by construction all leaves in $S$ are axioms. 
To show that all infinite matches are winning, consider an infinite path 
$\beta = (v_{n})_{n\in\om}$ in $S$. 
We need to show that $\beta$ contains a $\nu$-trail. 
Using condition~\ref{i:order preserving} it follows that there is an infinite 
path $\alpha = (t_{n})_{n\in\om}$ in $\Pi$ such that for every $i \in \omega$
we have that $g(v_i) = t_{k_i}$ for some $k_i \in \omega$, and, moreover, 
$k_i \leq k_j$ if $i \leq j$. 
By Proposition~\ref{p:nu-trail} the infinite path $\alpha$ contains a 
$\nu$-trail $\tau = \phi_0^{a_0} \phi_1^{a_1} \cdots$. 
With condition~\ref{i:trb} it follows that $\tau' \isdef \phi_{k_0} \phi_{k_1} \phi_{k_2} 
\cdots$ is a trail on $\beta$.
By Proposition~\ref{p:af4}, $\tau$ contains only finitely many $\mu$-formulas;
from this it is immediate that $\tau'$ also features at most finitely many 
$\mu$-formulas.
Thus, using Proposition~\ref{p:af4} a second time, we find that $\tau'$ is
a $\nu$-trail, as required.
\end{proofof}

\section{Completeness}
\label{s:completeness}

In this section we show that the focus systems are complete, that is,
every valid sequent is provable in either $\Focus$ or $\Focusinf$. 
As for the soundness argument in the previous section, we rely on 
Theorem~\ref{t:adequacy} which states that \Prover has a winning strategy
in any tableau for a given valid formula, and on Theorem~\ref{t:same}
which claims that every formula that is provable in \Focusinf is also provable
in \Focus. 
Thus, it suffices to show that winning strategies for \Prover in the tableau
game can be transformed into \Focusinf-proofs.

\begin{theorem} \label{t:completeness}
If \Prover has a winning strategy in some tableau game for a sequent $\Phi$
then $\Phi$ is provable in \Focusinf.
\end{theorem}
\begin{proof}
Let $\bbT = (V,E,\Phi,\tLab,v_I)$ be a tableau for $\Phi$ and let $S$ be
a winning strategy for \Prover in $\game{\bbT}$. 
Because of Proposition~\ref{p:tableau exists}, Corollary~\ref{cor:invariant} 
and Remark~\ref{r:treestrat} of we may assume that $\bbT$ is tree based, with 
root $v_I$, and that $S \subseteq V$ is a subtree of $\bbT$. 
We will construct a \Focusinf-proof $\Pi = (T,P,\Sigma,\pLab)$ for $\Phi^f$.

Applications of the focus rules in $\Pi$ will be very restricted.
To start with, the unfocus rule $\RuUnfocus$ will not be used at all, and the 
focus rule $\RuFocus$ will only occur in series of successive applications, 
with the effect of transforming an annotated sequent of the form $\Psi^{u}$ 
into its totally focused companion $\Psi^{f}$.
It will be convenient to think of this series of applications of $\RuFocus$ as
a \emph{single} proof rule, which we shall refer to as the total focus rule 
$\RuFocustot$:
\begin{prooftree}
 \AxiomC{$\Phi^f$}
 \RightLabel{\RuFocustot}
 \UnaryInfC{$\Phi^u$}
\end{prooftree}

We construct the pre-proof $\Pi$ of $\Phi^f$ together with a function $g : S
\to T$ in such a way that the following conditions are satisfied:
\begin{enumerate}
\item \label{i:first} \label{i:ordpres} 
If $E v u$ then $P^+ g(v) g(u)$.
\item \label{i:path lifting} 
For every $v \in S$ and every infinite branch $\beta = (v_{n})_{n\in\om}$
in $\Pi$ with $v_0 = g(v)$ there is some $i \in \omega$ and some $u \in S$ 
such that $Evu$ and $g(u) = v_i$.
\item \label{i:push formulas down} 
$\Sigma_{g(v)}$ is thin.
\item \label{i:trace preserving} 
If $Evu$ and $(\varphi,\psi) \in \gtrail_{v,u}$ then $(\varphi^{a_\varphi},
\psi^{a_\psi}) \in \gtrail_{g(v),g(u)}$.
\item \label{i:unfocus reflects mu}
If $Evu$, and $s$ and $t$ are nodes on the path from $g(v)$ to $g(u)$ such that
$P^+st$, $(\chi^a,\varphi^f) \in \gtrail_{g(v),s}$ for some $a \in \{f,u\}$ and 
$(\varphi^f,\psi^u) \in \gtrail_{s,t}$, then $\chi = \varphi$ and $\chi$ is a 
$\mu$-formula.
\item \label{i:last} \label{i:eventually focus} 
If $\alpha$ is an infinite branch of $\Pi$ and \RuFocustot is applicable at some
node on $\alpha$, then $\RuFocustot$ is applied at some later node on $\alpha$.
\end{enumerate}
The purpose of these conditions is that they allow us to prove later
that every branch in $\Pi$ is successful.

We construct $\Pi$ and $g$ as the limit of finite stages, where at stage
$i$ we have constructed a finite pre-proof $\Pi_i$ and a partial
function $g_i : S \to \Pi_i$. At every stage we make sure that $g_i$ and
$\Pi_i$ satisfy the following conditions:
\begin{enumerate}[resume]
\item \label{i:open leaves in range} 
All open leaves of $\Pi_i$ are in the range of $g_i$.
\item \label{i:same sequent} 
All nodes $v \in S$ for which $g_i(v)$ is defined satisfy
$\Phi_v = \uls{\Si}_{g_{i}(v)}$.
\end{enumerate}

In the base case we define $\Pi_0$ to consist of just one node $r$ that is
labelled with the sequent $\Phi^f$. 
The partial function $g_0$ maps $r$ to $v_I$. 
Clearly, this satisfies the conditions \ref{i:open leaves in range} and 
\ref{i:same sequent}.

In the inductive step we consider any open leaf $m$ of $\Pi_i$, which
has a minimal distance from the root of $\Pi_i$. 
This ensures that in the limit every open leaf is eventually treated, so that 
$\Pi$ will not have any open leaves. 
By condition~\ref{i:open leaves in range} there is a $u \in S$ such that $g(u) 
= m$.

Our plan is to extend the proof $\Pi_i$ at the open leaf $m$ to mirror the
rule that is applied at $u$ in $\bbT$. In general this is possible because by 
condition~\ref{i:same sequent} the formulas in the annotated sequent at $m = 
g_i(u)$ are the same as the formulas at $u$.
All children of $u$ that are in $S$ should then be mapped by $g_{i+1}$ to new 
open leaves in $\Pi_{i + 1}$. 
This guarantees that condition~\ref{i:open leaves in range} is satisfied at 
step $i+1$ and because we are going to simulate the rule in the tableau by rules
in the focus system we ensure that condition~\ref{i:same sequent} holds at
these children as well. 
Clearly, the precise definition of $\Pi_{i+1}$ depends on the rule applied at 
$u$. 
Before going into the details we address two technical issues that feature in
all the cases.

First, to ensure that condition~\ref{i:eventually focus} is satisfied by our
construction we will apply \RuFocustot at $m$, whenever it is applicable. 
Thus, we need to check whether all formulas in the sequent of $m$ are annotated 
with $u$. 
If this is the case then we apply the total focus rule and proceed with the 
premise $n$ of this application of the focus rule. 
Otherwise we just proceed with $n = m$. 
Note that in either case the sequent at $n$ contains the same formulas as the 
sequent at $m$ and if $n \neq m$ then the trace relation relates the formulas 
at $n$ in an obvious way to those at $m$.

The second technical issue is that to ensure condition~\ref{i:push formulas 
down} we may need to apply \RuWeak to the new leaves of $\Pi_{i + 1}$. 
To see how this is done assume we have already extended $\Pi_i$ and obtained 
a new leaf $v$ which we would like to add into the range of $g_{i + 1}$. 
The annotated sequent at $v$, however, might contain both instances $\varphi^f$
and $\varphi^{u}$ of some formula $\varphi$, which would violate 
condition~\ref{i:push formulas down}. 
To take care of this we apply \RuWeak to get rid of the unfocused occurrence 
$\varphi^u$. 
in fact, we might need to apply \RuWeak multiple times to get rid of all 
unfocused duplicates of formulas. 
In the following we will refer to the node of the proof, that is obtained by
repeatedly applying \RuWeak in this way at an open leaf $l$, as the 
\emph{thin normalisation} of $l$.
\medskip

We are now ready to discuss the main part of the construction, which is based
on a case distinction depending on the rule $\tLab(u)$ that is applied at $u$.
\smallskip

\textit{Case $\tLab(u) = \AxLit$ or $\tLab(u) = \AxTop$:}
In this case we can just apply the corresponding rule at $m = g(u)$.
We might need to apply \RuWeak to get rid of side formulas that were present 
in the tableau. 
There is no need to extend $g_i$.

\textit{Case $\tLab(u) = \RuOr$:} 
In this case we can just apply \RuOr at $m$. 
This generates a new open leaf $l$ which corresponds to the successor node $v$
of $u$ in the tableau. 
We define $g_{i + 1}$ such that it maps $v$ to the thin normalisation of $l$.

\textit{Case $\tLab(u) = \RuAnd$:} 
In this case we also apply \RuAnd in the focus system at $m$. 
This generates two successors which we can associate with the two children of 
$u$, both of which must be in $S$.
Thus, $g_{i+1}$ will map the children of $u$ to the thin normalisations of the 
successors we have added to $m$.

\textit{Case $\tLab(u) = \RuMod$:} 
In this case we want to apply the rule \RuBox in the focus system. 
However, the sequent $\Sigma_m$ might contain multiple box formulas, whereas
\RuBox can only be applied to one of those.
To select the proper formula $\Box \varphi^{a} \in \Sigma_m$ we use the fact 
that the successors of $u$ are indexed by the box formulas in $\Phi_{u}$, and 
that the strategy $S$ contains precisely one of these successors.
That is, let $\Box \varphi^{a} \in \Sigma_m$ be such that its associated 
successor $v_\varphi$ of $u$ belongs to $S$. 
We then apply \RuWeak at $m$ until we have removed all formulas from the sequent
that are not diamond formulas and that are distinct from $\Box \varphi$. 
Once this is done the sequent only contains annotated versions of the diamond 
formulas from $\Phi_u$ plus an annotated version of the formula $\Box \varphi$.
We can then apply \RuBox and obtain a new node $l$ and we define 
$g_{i+1}(v_\varphi)$ to be the thin normalisation of $l$.

\textit{Case $\tLab(u) = \RuMu$ or $\tLab(u) = \RuNu$:} 
This is analogous to the case for \RuOr. 
Note, however, that the application of the fixpoint rules in the focus system 
has an effect on the annotation. 
\medskip

We define the function $g : S \to T$ as the limit of the maps $g_i$.
To see that $g$ is actually a total function, first observe that for every 
$v \in S$ and $i \in \omega$ either $v$ is already in the domain of $g_i$, 
in which case it is in the domain of $g$, or there is some node $u$ on the branch
leading to $v$ that is mapped by $g_i$ to an open leaf of $\Pi_i$.
Eventually, the proof is extended at this leaf because in every step we treat
an open leaf that is maximally close to the root. 
It is easy to check that in every step, when we extend the proof $\Pi_j$ at some
open leaf, we also move forward on the branches of $\bbT$ that run through $v$. 
Iterating this reasoning shows that eventually $v$ must be added to the domain 
of some $g_j$.
\medskip

We now show that $g$, together with $\Pi$, satisfies the
conditions~\ref{i:first}--\ref{i:last}. To start with, it is clear from
the step-wise construction of $g$ and $\Pi$ that
condition~\ref{i:ordpres} is satisfied.

Condition~\ref{i:path lifting} holds because all trees $\Pi_i$
are finite. Thus, on every infinite branch of $\Pi$ there are infinitely
many nodes that are a leaf in some $\Pi_i$ and by condition~\ref{i:open
leaves in range} each of these nodes is in the range of $g_i$ and thus of
$g$. 

Condition~\ref{i:push formulas down} is obviously satisfied at the root
of $\Pi$. It is satisfied at all other nodes because of
condition~\ref{i:same sequent} and because we make sure that we only add nodes
to the domain of $g$ that are normalized, using the procedure described above.

To see that condition~\ref{i:trace preserving} is satisfied by $\Pi$ and $g$
one has to carefully inspect each case of the inductive definition of $\Pi$. 
This is tedious but does not give rise to any technical difficulties.

To check condition~\ref{i:unfocus reflects mu}, note that if $(\varphi^f,\psi^u)
\in \gtrail_{s,t}$ then the trace from $\varphi^{f}$ to $\psi^u$ must lose its 
focus at some point on the path from $s$ to $t$. 
Since we do not use the unfocus rule in $\Pi$, the only case of the inductive 
construction of $\Pi$ where this is possible is the case where $\tLab(u) = 
\RuMu$.
In this case the formula that loses its focus is the principal formula, which
is then a $\mu$-formula and already present at the open leaf that we are
extending.

For condition~\ref{i:eventually focus} first observe that
if \RuFocustot is applicable at some node that is an open leaf of some
$\Pi_i$ then it will be applied immediately when this open leaf is taken
care of. 
Moreover, it is not hard to see that if \RuFocustot becomes applicable at 
some node $v$ during some stage $i$ of the construction of $\Pi$, then it 
will remain applicable at every node that is added above $v$ at this stage. 
This applies in particular to the new open leaves that get added above $v$, 
and so the total focus rule will be applied to each of these at a later stage
of the construction.
\medskip

It remains to show that every infinite branch in $\Pi$ is successful.
Let $\beta = (v_{n})_{n\in\om}$ be such a branch.
We claim that
\begin{equation}
\label{eq:cpl1}
\text{from some moment on, every sequent on $\be$ contains a formula in focus},
\end{equation}
and to prove \eqref{eq:cpl1} we will link $\be$ to a match in $S$.
Observe that because of condition~\ref{i:path lifting} we can `lift' $\beta$
to a branch $\alpha = (t_{n})_{n\in\om}$ in $S$ such that there are 
$0 = k_0 < k_1 < k_2 < \cdots$ with $g(t_i) = v_{k_i}$
for all $i < \omega$. 
Because $\alpha$, as a match of the tableau game, is won by \Prover, it contains
a $\nu$-trail $(\phi_{n})_{n\in\om}$. 
This trail being a $\nu$-trail means that there is some $m \in \omega$ such that 
$\varphi_h$ is a $\mu$-formula for no $h \geq m$. 
We then use condition~\ref{i:trace preserving} to obtain a trace $\psi_0^{a_0}
\psi_1^{a_1} \cdots$ in $\beta$ such that $\varphi_i = \psi_{k_i}$.
Now distinguish cases.

First assume that there is an application of the total focus rule at some $v_l$,
with $l \geq k_m$.
Then at $v_{l + 1}$ all formulas are in focus and thus in particular the
annotation $a_{l+1}$ of the formula $\psi_{l+1}$ must be equal to $f$.
We show that
\begin{equation}
\label{eq:cpl2}
a_n = f \text{ for all } n > l. 
\end{equation}
Assume for contradiction that this is not the case and let $t$ be the smallest
number larger than $l$ such that $a_t = u$; since $a_{l+1} = f$ we find that 
$n > l+1$, and by assumption on $n$ we have $a_{t - 1} = f$.
Now let $h$ be such that $v_{n-1}$ and $v_{n}$ are on the path between 
$g(t_h) = v_{k_h}$ and $g(t_{h+1}) = v_{k_{h+1}}$; since 
$k_m \leq l \leq n-1$ it follows that $h \geq m$. 
But then by condition~\ref{i:unfocus reflects mu} $\varphi_h$ must be a 
$\mu$-formula, which contradicts our observation above that $\varphi_h$ is 
\emph{not} a $\mu$-formula for any $h \geq m$.
This proves \eqref{eq:cpl2}, which means that for every $n > l$, the formula 
$\psi_{n}$ is in focus at $v_{n}$.
From this \eqref{eq:cpl1} is immediate.

If, on the other hand, there is \emph{no} application of the total focus rule 
on $v_{k_m} v_{k_m+1} \cdots$ then it follows by condition~\ref{i:eventually 
focus} that the total focus rule is not \emph{applicable} at any sequent 
$v_{l}$ with $l \geq k_{m}$.
In other words, all these sequents contain a formula in focus, which
proves~\eqref{eq:cpl1} indeed.
\end{proof}

\section{Interpolation}
\label{sec-itp}

In this section we will show that the alternation-free fragment of the modal
$\mu$-calculus enjoys the Craig interpolation property.
To introduce the actual statement that we will prove, consider an implication
of the form $\phi \to \psi$, with $\phi,\psi \in \AFMC$.
First of all, we may without loss of generality assume that $\phi$ and $\psi$ 
are guarded, so that we may indeed take a proof-theoretic approach using the
$\Focus$ system.
Given our interpretation of sequents, we represent the implication $\phi \to 
\psi$ as the sequent $\ol{\phi},\psi$, and similarly, the implications involving
the interpolant $\theta$ can be represented as, respectively, the sequents
$\ol{\phi},\theta$ 
and $\ol{\theta},\psi$.
What we will prove below is that for an arbitrary derivable sequent
$\Gamma$,
and an arbitrary partition $\Gamma^{L},\Gamma^{R}$ of $\Gamma$, there is an 
interpolant $\theta$ such that the sequents $\Gamma^{L},\theta$ and 
$\Gamma^{R},\ol{\theta}$ are both provable.

Before we can formulate and prove our result, we need some preparation.
First of all, we will assume that in our $\Focus$ proofs every application of the
discharge rule discharges at least one assumption, i.e., every node in the proof
that is labelled with the discharge rule is the companion of at least one leaf.
It is easy to see that we can make this assumption without loss of generality
--- we leave the details to the reader.

Furthermore, it will be convenient for us to fine-tune the notion of 
a partition in the following way.

\begin{definition}
A \emph{partition} of a set $A$ is a non-empty finite tuple $(A_{1},\ldots,
A_{n})$ of pairwise disjoint subsets of $A$ such that $\bigcup_{i=1}^{n} A_{i} 
= A$.
A binary partition of $A$ may be denoted as $A^{L} \mid A^{R}$; in this setting
we may refer to the members of $A^{L}$ and $A^{R}$ as being \emph{left} and 
\emph{right} elements of $A$, respectively.
\end{definition}

Finally, to formulate the condition on an interpolant, note that we may identify
the \emph{vocabulary} of a sequent $\Si$ simply with the set $\FV(\Si)$ of free 
variables occurring in $\Si$. Our interpolation result can then be
stated as follows:

\begin{theorem}[{\bf Interpolation}]
\label{t:itp}
Let $\Pi$ be a \Focus-proof of some sequent $\Gamma$, and let $\Gamma^{L}\mid\Gamma^{R}$ be a 
partition of $\Gamma$.
Then there are a formula $\theta$ with $\FV(\theta) \sse \FV(\Gamma^{L}) \cap 
\FV(\Gamma^{R})$, and \Focus-proofs $\Pi^{L}$, $\Pi^{R}$, all effectively 
obtainable from $\Pi, \Gamma^{L}$ and $\Gamma^{R}$, such that $\Pi^{L}$ derives the
sequent $\Gamma^{L},\theta$ and $\Pi^{R}$ derives the sequent $\Gamma^{R},
\ol{\theta}$.
\end{theorem}

The remainder of this section contains the proof of this theorem. 
We first consider the definition of interpolants for the conclusion of a
single proof rule, under the assumption that we already have
interpolants for the premises. We then show in
Proposition~\ref{p:locitp} that this definition is well-behaved. We need
some additional auxiliary definitions.

In this section it will be convenient to define the negation of $\theta$
in a slightly simpler manner than in section~\ref{s:prel}. This is
possible since the bound variables of $\theta$ will be taken from the
set $\Tokens$ of discharge tokens, which is disjoint from the collection
of variables used in the formulas featuring in $\Pi$.

\begin{definition}
\label{d:altneg}
Given a formula $\phi$ such that $\BV(\phi) \sse \Tokens$, we define the 
formula $\sneg{\phi}$ as follows.
For atomic $\phi$ we define
\[
\sneg{\phi} \isdef 
   \left\{\begin{array}{ll}
      x         & \text{if } \phi = x \in \Tokens
   \\ \ol{\phi} & \text{otherwise},
   \end{array}\right.
\]
and then we inductively we continue with
\[\begin{array}{lllclll}
   \sneg{\phi \land \psi} & \isdef & \sneg{\phi} \lor \sneg{\psi} 
   &&
   \sneg{\phi \lor \psi} & \isdef & \sneg{\phi} \land \sneg{\psi}
\\ \sneg{\Box \phi} & \isdef & \Diamond \sneg{\phi}  
   &&
   \sneg{\Diamond \phi} & \isdef & \Box \sneg{\phi} 
\\ \sneg{\mu x. \phi} & \isdef & \nu x. \sneg{\phi} 
   &&
   \sneg{\nu x. \phi} & \isdef & \mu x. \sneg{\phi}
\end{array}\]
\end{definition}
It is not hard to see that $\sneg{\theta} = \ol{\theta}$ precisely if
$\FV(\theta)$ does not contain any discharge token from $\Tokens$ as a free
variable.
For atomic formulas $\phi$ that are not of the form $x \in \Tokens$ we will 
continue to write $\ol{\phi}$ rather than $\sneg{\phi}$.

\begin{definition}
A formula is \emph{basic} if it is either atomic, or of the form $x$, $x_{0} \land x_{1}$, 
$x_{0} \lor x_{1}$, $\dia x$ or $\Box x$, where $x$, $x_{0}$ and $x_{1}$ are 
discharge tokens.
\end{definition}

\begin{definition}
\label{d:locitp}
Let $\Ru$ be some derivation rule, let 
\begin{prooftree}
\AXC{$\Si_{0} \quad\ldots\quad \Si_{n-1}$}
\UIC{$\Si$}
\end{prooftree}
be an instance of $\Ru$, and let $\Si^{L}\mid\Si^{R}$ be a partition of $\Si$.
By a case distinction as to the nature of the rule $\Ru$ we define a \emph{basic}
formula $\chi(x_{0},\ldots,x_{n-1})$, together with a partition $\Si_{i}^{L}\mid
\Si_{i}^{R}$ for each $\Si_{i}$.
Here the variables $x_{0},\ldots,x_{n-1}$ correspond to the premises of the
rule.

\begin{description}

\item[Case $\Ru = \AxLit$.]
Let $\Si$ be of the form $\Si = \{ p, \atneg{p} \}$, and observe that 
since there are no premises, we only need to define the formula $\chi$.
For this purpose we make a further case distinction as to the exact nature of
the partition.

If $\Si^{L}\mid \Si^{R} = p^a \mid \atneg{p}^{b}$, define $\chi \isdef \atneg{p}$.

If $\Si^{L}\mid \Si^{R} = \atneg{p}^a \mid p^{b}$, define $\chi \isdef p$.

If $\Si^{L}\mid \Si^{R} = p^a, \atneg{p}^{b} \mid \nada$, define $\chi \isdef \bot$.

If $\Si^{L}\mid \Si^{R} = \nada \mid p^a, \atneg{p}^{b}$, define $\chi \isdef \top$.

\item[Case $\Ru = \AxTop$.]
Here $\Si$ must be of the form $\Si = \{\top \}$, and, as in the case of the 
other axiom, we only need to define the formula $\chi$ since there are no 
premises.
We make a further case distinction.

If $\Si^{L}\mid \Si^{R} = \top^a \mid \nada$, define $\chi \isdef \bot$.

If $\Si^{L}\mid \Si^{R} = \nada \mid \top^a$, define $\chi \isdef \top$.

\item[Case $\Ru = \RuAnd$.]
We distinguish cases, as to which side the active formula $(\phi_{0}\land
\phi_{1})^{a}$ belongs to.

   \begin{description}

   \item[Subcase $(\phi_{0}\land \phi_{1})^{a} \in \Si^{L}$.]
   We may then represent the partition of $\Si$ as 
   $(\phi_{0}\land\phi_{1})^{a},\Si_{0} \mid \Si_{1}$.
   Here we define $\chi(x_{0},x_{1}) \isdef x_{0} \lor x_{1}$, and we partition
   the premises of $\RuAnd$ as, respectively, 
   $\phi_{0}^{a},\Si_{0} \mid \Si_{1} \setminus \{\phi_{0}^{a}\}$ and
   $\phi_{1}^{a},\Si_{0} \mid \Si_{1} \setminus \{\phi_{1}^{a}\}$.
    
   \item[Subcase $(\phi_{0}\land \phi_{1})^{a} \in \Si^{R}$.]
   We may now represent the partition of $\Si$ as 
   $\Si_{0} \mid \Si_{1}, (\phi_{0}\land\phi_{1})^{a}$.
   Now we define $\chi(x_{0},x_{1}) \isdef x_{0} \land x_{1}$, and we partition
   the premises of $\RuAnd$ as, respectively, 
   $\Si_{0}\setminus \{\phi_{0}^{a}\} \mid \Si_{1}, \phi_{0}^{a}$ and
   $\Si_{0}\setminus \{\phi_{1}^{a}\} \mid \Si_{1}, \phi_{1}^{a} $.

   \end{description}

\item[Case $\Ru = \RuOr$.]
We only consider the case where the active formula $(\phi_{0}\lor\phi_{1})^{a}$
belongs to $\Si^{L}$ (the other case is symmetric).
We may then represent the partition of $\Si$ as 
$(\phi_{0}\lor\phi_{1})^{a},\Si_{0} \mid \Si_{1}$.
Here we define $\chi(x_{0}) \isdef x_{0}$, and 
we partition the premise of $\RuOr$ as 
$\phi_{0}^{a},\phi_{1}^{a},\Si_{0} \mid \Si_{1} \setminus 
   \{\phi_{0}^{a},\phi_{1}^{a}\}$.

\item[Case $\Ru = \RuBox$.]
We distinguish cases, as to whether the active formula $\Box\phi^{a}$ belongs
to $\Si^{L}$ or to $\Si^{R}$.

    \begin{description}

    \item[Subcase $\Box\phi^{a} \in \Si^{L}$.]
	We may then represent the partition of $\Si$ as $\Box\phi^{a},\dia\Si_{0}
	\mid \dia\Si_{1}$.
	We define $\chi \isdef \dia x_{0}$ and we partition the premise 
	of $\RuBox$ as $\phi, \Si_{0} \mid \Si_{1} \setminus \{ \phi \}$.

    \item[Subcase $\Box\phi^{a} \in \Si^{R}$.]
	We may then represent the partition of $\Si$ as 
	$\Si_{0} \mid \Si_{1}, \Box\phi^{a}$.
	Now we define $\chi \isdef \Box x_{0}$ and we partition the premise 
	of $\RuBox$ as $\Si_{0} \setminus \{ \phi \} \mid \Si_{1}, \phi$.
	\end{description}

\item[Case $\Ru = \RuMu$.]
We only consider the case where the active formula $\mu x.\phi^{a}$ belongs to
$\Si^{L}$ (the other case is symmetric).
We may then represent the partition of $\Si$ as $\mu x.\phi^{a},\Si_{0} \mid 
\Si_{1}$.
Here we define $\chi(x_{0}) \isdef x_{0}$, and we partition the premise
of $\RuMu$ as $\phi(\mu x.\phi)^{u},\Si_{0} \mid \Si_{1} \setminus \{ 
\phi(\mu x.\phi)^{u} \}$.

\item[Case $\Ru = \RuNu$.]
The definitions are analogous to the case of $\RuMu$. 

\item[Case $\Ru = \RuWeak$.]
We only consider the case where the active formula $\phi^{a}$ belongs to 
$\Si^{L}$ (the other case is symmetric).
We may then represent the partition of $\Si$ as 
$\phi^{a},\Si_{0} \mid \Si_{1}$.
Here we define $\chi(x_{0}) \isdef x_{0}$, and 
we partition the premise of $\RuWeak$ as 
$\Si_{0} \mid \Si_{1}$.

\item[Case $\Ru = \RuFocus$.]
We only consider the case where the active formula $\phi^{u}$ belongs to 
$\Si^{L}$ (the other case is symmetric).
We may then represent the partition of $\Si$ as 
$\phi^u, \Si_{0} \mid \Si_{1}$.
In this case we define $\chi(x_{0}) \isdef x_{0}$, and 
we partition the premise of $\RuFocus$ as $\phi^f, \Si_0 \mid \Si_{1}
\setminus \{\phi^f\}$.

\item[Case $\Ru = \RuUnfocus$.]
This case is analogous to the case for \RuFocus, just swapping the
annotations of $\phi$.

\item[Case $\Ru = \RuDischarge{}$.]
In this case the premise and the conclusions are the same, and so we also 
partition the premise in the same way as the conclusion.
Furthermore, we define $\chi \isdef x_{0}$.
\end{description}
\end{definition}

\begin{proposition}[{\bf Interpolation Transfer}]
\label{p:locitp}
Let
\begin{prooftree}
\AXC{$\Si_{0} \quad\ldots\quad \Si_{n-1}$}
\UIC{$\Si$}
\end{prooftree}
be an instance of some derivation rule $\Ru \neq \RuDischarge{}$, let $\Si^{L}
\mid\Si^{R}$ be a partition of $\Si$, and let $\chi$ and $\Si_{i}^{L}\mid
\Si_{i}^{R}$, for $i = 0,\ldots,n-1$ be as in Definition~\ref{d:locitp}.
Then the following hold:
\begin{urlist}
\item \label{i:coh}
$\FV(\Si_{i}^{K}) \sse \FV(\Si^{K})$ where $K \in \{ L, R \}$;

\item \label{i:itptrf}
For any sequence $\theta_{0},\ldots,\theta_{n-1}$ of formulas and any 
$b \in \{ u,f \}$ there are derivations $\Xi^{L}$ and $\Xi^{R}$:

\begin{minipage}{.40\textwidth}
\begin{prooftree}
\AXC{$\Si_{0}^{L},\theta_{0}^{b} \quad \ldots \quad 
   \Si_{n-1}^{L},\theta_{n-1}^{b}$}
\noLine\UIC{$\vdots$}
\noLine\UIC{$\Xi^L$}
\noLine\UIC{$\vdots$}
\UIC{$\Si^{L}, \chi(\theta_{0},\ldots,\theta_{n-1})^{b}$}
\end{prooftree}
\end{minipage}
\quad and \quad 
\begin{minipage}{.40\textwidth}
\begin{prooftree}
\AXC{$\Si_{0}^{R},\sneg{\theta_{0}}^{b} \quad \ldots \quad 
   \Si_{n-1}^{R},\sneg{\theta_{n-1}}^{b}$}
\noLine\UIC{$\vdots$}
\noLine\UIC{$\Xi^R$}
\noLine\UIC{$\vdots$}
\UIC{$\Si^{R}, \sneg{\chi(\theta_{0},\ldots,\theta_{n-1})}^{b}$}
\end{prooftree}
\end{minipage}

\noindent
Provided that $\Ru \notin \{\RuFocus,\RuUnfocus\}$, these derivations satisfy the following conditions:
   \begin{urlist}
   \item[a)] $\Xi^{L}$ and $\Xi^{R}$ do not involve
      the rules \RuFocus or \RuUnfocus.
   \item[b)] 
   If, for some $i$, the assumption $\Si_{i}^{L},\theta_{i}^{b}$ contains a
   formula in focus, then so does every sequent in $\Xi^{L}$ on the path to
   this assumption.
   \item[c)] 
   If, for some $i$, the assumption $\Si_{i}^{R},\sneg{\theta_{i}}^{b}$ contains 
   a formula in focus, then so does every sequent in $\Xi^{L}$ on the path to
   this assumption.
   \item[d)] If $\Ru = \RuBox$ then there is an applications of $\RuBox$ at
the root of $\Xi^L$ and $\Xi^R$.
\end{urlist}
\end{urlist}
\end{proposition}
\begin{proof}
The proof of both parts proceeds via a case distinction depending on the 
proof rule $\Ru$, following the case distinction in Definition~\ref{d:locitp}.
Part~(\ref{i:coh} easily follows from a direct inspection.
For part~(\ref{i:itptrf} we restrict attention to some representative cases.

Below we use $\RuWeak^*$ as a `proof rule' in the sense that, in a proof,
we draw the configuration
\AXC{$\Gamma_{t}$}
\RightLabel{$\RuWeak^*$}
\UIC{$\Gamma_{s}$}
\DisplayProof
to indicate that either $\Gamma_{t}$ is a proper subset of $\Gamma_{s}$,
in which case we are using repeated applications of the weakening rule
at node $s$, or else there is only one single node $s = t$ labelled with
$\Gamma_{s} = \Gamma_{t}$.

\begin{description}
\item[Case $\Ru = \AxLit$.]
As an example consider the case where the partition is such that
$\Sigma^L \mid \Sigma^R = p^a \mid \atneg{p}^{c}$. Then we have by
definition that $\chi = \atneg{p}$ and hence we need to supply proofs for
the annotated sequents $\Sigma^L,\chi^b = p^a,\atneg{p}^b$ and
$\Sigma^R,\sneg{\chi}^b = \atneg{p}^{c},p^b$. 
Both of these can easily be proved with the axiom \AxLit.

As a second example consider the case where the partition is such that
$\Sigma^L \mid \Sigma^R = p^a, \atneg{p}^{c} \mid \nada$. Then we have that
$\chi = \bot$ and hence need to provide proofs for the sequents
$\Sigma^L,\chi^b = p^a,\atneg{p}^{c},\bot^b$ and $\Sigma^R,\sneg{\chi}^b =
\top^b$. The latter is proved with \AxTop and for the former we use the
proof:
\begin{prooftree}
\AXC{\phantom{X}}
\RightLabel{\AxLit}
\UIC{$p^a,\atneg{p}^{c}$}
\RightLabel{\RuWeak}
\UIC{$p^a,\atneg{p}^{c},\bot^b$}
\end{prooftree}

\item[Case $\Ru = \RuAnd$.]
First assume that the active formula $(\phi_{0}\land\phi_{1})^{a}$
belongs to $\Si^{L}$. We may then represent the partition of $\Si$ as
$(\phi_{0}\land\phi_{1})^{a},\Si_{0} \mid \Si_{1}$.
For the claim of the proposition, the following derivations suffice:

\begin{minipage}{.45\textwidth}
\begin{prooftree}
\AXC{$\Si_{0},\phi_{0}^{a},\theta_{0}^{b}$}
\RightLabel{$\RuWeak$}
\UIC{$\Si_{0},\phi_{0}^{a},\theta_{0}^{b},\theta_{1}^{b}$}
\RightLabel{$\RuOr$}
\UIC{$\Si_{0},\phi_{0}^{a},(\theta_{0}\lor\theta_{1})^{b}$}
\AXC{$\Si_{0},\phi_{1}^{a},\theta_{1}^{b}$}
\RightLabel{$\RuWeak$}
\UIC{$\Si_{0},\phi_{1}^{a},\theta_{0}^{b},\theta_{1}^{b}$}
\RightLabel{$\RuOr$}
\UIC{$\Si_{0},\phi_{1}^{a},(\theta_{0}\lor\theta_{1})^{b}$}
\RightLabel{$\RuAnd$}
\BIC{$\Si_{0},(\phi_{0}\land\phi_{1})^{a},(\theta_{0}\lor\theta_{1})^{b}$}
\end{prooftree}
\end{minipage}
\begin{minipage}{.45\textwidth}
\begin{prooftree}
\AXC{$\Si_{1}\setminus\{\phi_{0}^{a}\},\sneg{\theta_{0}}^{b}$}
\RightLabel{$\RuWeak^*$}
\UIC{$\Si_{1},\sneg{\theta_{0}}^{b}$}
\AXC{$\Si_{1}\setminus\{\phi_{1}^{a}\},\sneg{\theta_{1}}^{b}$}
\RightLabel{$\RuWeak^*$}
\UIC{$\Si_{1},\sneg{\theta_{1}}^{b}$}
\RightLabel{$\RuAnd$}
\BIC{$\Si_{1}, (\sneg{\theta_{0}}\land\sneg{\theta_{1}})^{b}$}
\end{prooftree}
\end{minipage}

We then consider the other possibility, where the active formula
$(\phi_{0}\land\phi_{1})^{a}$ belongs to $\Si^{R}$. 
We may represent the partition of $\Si$ as $\Si_{0} \mid
(\phi_{0}\land\phi_{1})^{a},\Si_{1}$.
Now the following derivations suffice:
\begin{prooftree}
\AXC{$\Si_{0}\setminus\{\phi_{1}^{a}\},\theta_{0}^{b}$}
\RightLabel{$\RuWeak^*$}
\UIC{$\Si_{0},\theta_{0}^{b}$}
\AXC{$\Si_{0}\setminus\{\phi_{0}^{a}\},\theta_{1}^{b}$}
\RightLabel{$\RuWeak^*$}
\UIC{$\Si_{0},\theta_{1}^{b}$}
\RightLabel{$\RuAnd$}
\BIC{$\Si_{0}, (\theta_{0}\land\theta_{1})^{b}$}
\end{prooftree}

\begin{prooftree}
\AXC{$\Si_{1},\phi_{0}^{a},\sneg{\theta_{0}}^{b}$}
\RightLabel{$\RuWeak$}
\UIC{$\Si_{1},\phi_{0}^{a},\sneg{\theta_{0}}^{b},\sneg{\theta_{1}}^{b}$}
\RightLabel{$\RuOr$}
\UIC{$\Si_{1},\phi_{0}^{a},(\sneg{\theta_{0}}\lor\sneg{\theta_{1}})^{b}$}
\AXC{$\Si_{1},\phi_{1}^{a},\sneg{\theta_{1}}^{b}$}
\RightLabel{$\RuWeak$}
\UIC{$\Si_{1},\phi_{1}^{a},\sneg{\theta_{0}}^{b},\sneg{\theta_{1}}^{b}$}
\RightLabel{$\RuOr$}
\UIC{$\Si_{1},\phi_{1}^{a},(\sneg{\theta_{0}}\lor\sneg{\theta_{1}})^{b}$}
\RightLabel{$\RuAnd$}
\BIC{$\Si_{1},(\phi_{0}\land\phi_{1})^{a},(\sneg{\theta_{0}}\lor\sneg{\theta_{1}})^{b}$}
\end{prooftree}

\item[Case $\Ru = \RuOr$.]
We only consider the case where the active formula $(\phi_0 \lor
\phi_1)^a$ belongs to $\Si^{L}$ (the other case is similar). We may then
represent the partition of $\Si$ as $(\phi_0 \lor \phi_1)^a,\Si_{0}
\mid \Si_{1}$.
The two derivations below then suffice to prove the proposition:

\begin{prooftree}
\AXC{$\phi_0^a,\phi_1^a, \Si_{0}, \theta_{0}^{b}$}
\RightLabel{$\RuOr$}
\UIC{$(\phi_0 \lor \phi_1)^a, \Si_{0},\theta_{0}^{b}$}
\end{prooftree}

\begin{prooftree}
\AXC{$\Si_{1}\setminus\{\phi_0^{a},\phi_1^a\},\sneg{\theta_{1}}^{b}$}
\RightLabel{$\RuWeak^*$}
\UIC{$\Si_{1}\setminus\{\phi_0^{a}\},\sneg{\theta_{1}}^{b}$}
\RightLabel{$\RuWeak^*$}
\UIC{$\Si_{1}, \sneg{\theta_{0}}^{b}$}
\end{prooftree}

\item[Case $\Ru = \RuBox$.]
We only consider the case where the active formula $\Box\phi^{a}$ belongs to
$\Si^{L}$ (the other case is similar).
We may then represent the partition of $\Si$ as 
$\Box\phi^{a},\dia\Si_{0} \mid \dia\Si_{1}$.
The two derivations below then suffice to prove the proposition:

\begin{prooftree}
\AXC{$\phi^a, \Si_{0}, \theta_{0}^{b}$}
\RightLabel{$\RuBox$}
\UIC{$\Box\phi^a, \dia\Si_{0},\dia\theta_{0}^{b}$}
\end{prooftree}

\begin{prooftree}
\AXC{$\Si_{1}\setminus\{\phi^{a}\},\sneg{\theta_{1}}^{b}$}
\RightLabel{$\RuWeak^*$}
\UIC{$\Si_{1}, \sneg{\theta_{0}}^{b}$}
\RightLabel{$\RuBox$}
\UIC{$\dia\Si_{1}, \Box\sneg{\theta_{0}}^{b}$}
\end{prooftree}

\item[Case $\Ru = \RuMu$.]
We only consider the case where the principal formula $\mu x. \phi^a$ belongs to $\Si^{L}$ (the other case is similar). We may then
represent the partition of $\Si$ as $\mu x. \phi^a,\Si_{0}
\mid \Si_{1}$. 
The two derivations below then suffice to prove the proposition:

\begin{prooftree}
\AXC{$\phi(\mu x. \phi)^u, \Si_{0}, \theta_{0}^{b}$}
\RightLabel{$\RuMu$}
\UIC{$\mu x. \phi^a, \Si_{0},\theta_{0}^{b}$}
\end{prooftree}

\begin{prooftree}
\AXC{$\Si_{1}\setminus\{\phi(\mu x. \phi)^{u}\},\sneg{\theta_{1}}^{b}$}
\RightLabel{$\RuWeak^*$}
\UIC{$\Si_{1}, \sneg{\theta_{0}}^{b}$}
\end{prooftree}

\item[Case $\Ru = \RuNu$.]
This case is analogous to the case of \RuMu, simply keeping the
annotation of the principal formula, instead of unfocusing.

\item[Case $\Ru = \RuWeak$.]
We only consider the case where the weakened formula $\phi^a$ belongs to $\Si^{L}$ (the other case is similar). We may then
represent the partition of $\Si$ as $\phi^a,\Si_{0}
\mid \Si_{1}$. 
For $\Xi^L$ we can use the derivation
\begin{prooftree}
\AXC{$\Si_{0}, \theta_{0}^{b}$}
\RightLabel{$\RuWeak$}
\UIC{$\phi^a, \Si_{0},\theta_{0}^{b}$}
\end{prooftree}

The derivation $\Xi^R$ consists of the single sequent
$\Si_1,\sneg{\theta_0}^b$, without any rules being applied.

\item[Case $\Ru = \RuFocus$.]
Again, only consider the case where the principal formula is on the
left. We can write the partition of $\Si$ as $\phi^u, \Si_0 \mid \Si_1$
and use the proofs
\begin{prooftree}
\AXC{$\phi^f,\Si_{0}, \theta_{0}^{b}$}
\RightLabel{$\RuFocus$}
\UIC{$\phi^u, \Si_{0},\theta_{0}^{b}$}
\end{prooftree}
and
\begin{prooftree}
\AXC{$\Si_{1}\setminus\{\phi^f\},\sneg{\theta_{1}}^{b}$}
\RightLabel{$\RuWeak^*$}
\UIC{$\Si_{1}, \sneg{\theta_{0}}^{b}$}
\end{prooftree}

\item[Case $\Ru = \RuUnfocus$.]
This case is analogous to the case for \RuFocus.
\end{description}

To finish the proof of Proposition~\ref{p:locitp}, we need to check that each 
of the proofs given above satisfies the conditions (a) - (c).
Condition (a) can be verified by a direct inspection.
One may also verify the conditions (b) and (c) directly, using the observation
that for any node $t$ in the pre-proofs $\Pi^{L}$ and $\Pi^{R}$, if some 
formula occurring at a child of $t$ is annotated with $f$, then also some
formula at $t$ is annotated with $f$. Lastly, one can check in the case
for \RuBox that the constructed proof also contains an application of
\RuBox at its root.
\end{proof}

To prove Theorem~\ref{t:itp} we assemble the interpolant $\theta$ by an
induction on the tree that underlies the proof $\Pi$, where most cases
of the inductive step are covered by Definition~\ref{d:locitp} and
Proposition~\ref{p:locitp}. The main difficulty is treating the cases
for discharged leafs and the discharge rule. The idea is to introduce a
fresh variable as the interpolant of a discharged leaf and to then bind
the variable with a fixpoint operator at the step that corresponds to
the application of the discharge rule at the companion of the leaf. We
need to ensure that this can be done in such that the interpolant stays
alternation-free. The key notion that allows us to organize the
introduction of fixpoint operators to the interpolant are the fixpoint
colourings from Definition~\ref{d:coloring} below. The fixpoint colouring
specifies for every node in $\Pi$ whether the application of the
discharge rule at the node should be either a least fixpoint $\mu$ or a
greatest fixpoint $\nu$. Before we can discuss this notion we need to
show that the partition of $\Phi^L \mid \Phi^R$ of the root of $\Pi$
can be extended in a well-behaved way to all nodes of the proof.

\begin{definition}
Let $\Pi = (T,P,R,\Si)$ be a proof.
A \emph{nodewise partition} of $\Pi$ is a pair $(\Si^{L},\Si^{R})$ of labellings  
such that, for every $t \in T$, the pair $\Si^{L}_{t}\mid \Si^{R}_{t}$ is a 
partition of $\Si_{t}$.
Such a partition is \emph{coherent} if it agrees with the derivation rules 
applied in the proof, as expressed by Definition~\ref{d:locitp}.
\end{definition}

\begin{proposition}
\label{p:itp1}
Let $\Pi$ be a proof of some sequent $\Ga$ and let $(\Ga^{L},\Ga^{R})$ be 
a partition of $\Ga$.
Then there is a unique coherent nodewise partition $(\Si^{L},\Si^{R})$ of $\Pi$ 
such that $\Si^{L}_{r} = \Ga^{L}$ and $\Si^{R}_{r} = \Ga^{R}$, where $r$ is
the root of $\Pi$.
\end{proposition}

\begin{proof}
Immediate by the definitions.
\end{proof}

We shall refer to the nodewise partition given in Proposition~\ref{p:itp1} as
being \emph{induced} by the partition of the root sequent.

\begin{definition}
Let $\Pi = (T,P,R,\Si)$ be a proof and let $(\Si^{L}, \Si^{R})$ be a coherent
nodewise partition of $\Pi$.
This partition is called \emph{balanced} if $\Si^{L}_{l} = \Si^{L}_{c(l)}$ and 
$\Si^{R}_{l} = \Si^{R}_{c(l)}$, for every discharged leaves $l$ of $\Pi$.
\end{definition}

In words, a coherent nodewise partition is balanced if it splits the sequents
of any discharged leaf in exactly the same manner as it splits the leaf's
companion node.
As a corollary of the following proposition, for every partition $(\Ga^{L},
\Ga^{R})$ of a provable sequent $\Ga$ we can find a proof on which the induced
partition is balanced.

\begin{proposition}
\label{p:itp2}
Let $\Pi$ be a proof of some sequent $\Ga$, and let $(\Ga^{L},\Ga^{R})$
be a partition of $\Ga$. 
Then there is some finite proof $\Pi'$ of $\Ga$ such that the nodewise 
partition on $\Pi'$, induced by $(\Ga^{L},\Ga^{R})$, is balanced.
\end{proposition}

\begin{proof}
Let $\vec{\Pi}$ be the full unravelling of $\Pi$ into a \Focusinf-proof
according to Proposition~\ref{p:fintoinf}, and extend the nodewise
partition of $\Pi$ to $\vec{\Pi}$ in the obvious way.
Using the same strategy as in the proof of Proposition~\ref{p:ppp} we may `cut
off' $\vec{\Pi}$ to a balanced proof $\Pi'$.
%
\end{proof}

\begin{definition} \label{d:coloring}
Let $\Pi = (T,P,R,\Si)$ be a proof of some sequent $\Ga$, and let 
$(\Si^{L},\Si^{R})$ be a nodewise partition of $\Ga$.
A \emph{fixpoint colouring} for $(\Si^{L},\Si^{R})$ is a map $\eta: T \to \{ \mu,
\nu, \ntr \}$, satisfying the conditions below (where we write 
$T_{\mu} \isdef \eta^{-1}(\mu)$, etc.):
\begin{urlist}
\item 
$T_{\ntr}$ consists of those nodes that belong to \emph{no} set of the form 
$\itv{c(l)}{l}$;
\item
for every discharged leaf $l$ of $\Pi$ we have either 
$\itv{c(l)}{l} \sse T_{\mu}$ or $\itv{c(l)}{l} \sse T_{\nu}$;
\item \label{i:in focus}
if $t \in T_{\mu}$ then $\Si^{L}_{t}$ contains a focused formula, 
and if $t \in T_{\nu}$ then $\Si^{R}_{t}$ contains a focused formula.
\end{urlist}
We usually write $\eta_{t}$ rather than $\eta(t)$ and refer to $\eta_{t}$ as 
the \emph{fixpoint type} of $t$.
Nodes in $T_{\ntr}, T_{\mu}$ and $T_{\nu}$ will sometimes be called 
\emph{transparent}, \emph{magenta} and \emph{navy}, respectively.
\end{definition}

\begin{proposition}
\label{p:itp3}
Let $(\Si^{L},\Si^{R})$ be a balanced nodewise partition of some proof $\Pi$.
Then there is a fixpoint colouring $\eta$ for $(\Si^{L}, \Si^{R})$.
\end{proposition}

For a proof of Proposition~\ref{p:itp3}, we need the following definition and 
auxiliary proposition.

\begin{definition}
\label{d:conn}
Let $u_{0}$ and $u_{1}$ be two nodes of some proof $\Pi$.
We call $u_{0}$ and $u_{1}$ \emph{closely connected} if there is a non-axiomatic
leaf $l$ such that $u_{0},u_{1} \in \itv{c(l)}{l}$.
The relation of being \emph{connected} is the reflexive/transitive closure of 
that of being closely connected.
\end{definition}

The relation of being connected is easily seen to be an equivalence relation, 
which refines the partition induced by the fixpoint colouring; note that 
transparent nodes are only connected to themselves.
Furthermore, as we will see, the partition induced by the connectedness relation
refines the fixpoint colouring mentioned in Proposition~\ref{p:itp3}.
Here is the key observation that makes this possible.

\begin{proposition}
\label{p:itp4}
Let $(\Si^{L},\Si^{R})$ be a balanced nodewise partition of some proof $\Pi = 
(T,P,\Ru,\Si)$,
and let $u$ and $v$ be connected nodes of $\Pi$.
Then, for $K \in \{ L, R \}$, we have
\begin{equation}
\label{eq:lr1}
\Si^{K}_{u} \text{ contains a formula in focus iff } \Si^{K}_{v}
\text{ contains a formula in focus}.
\end{equation}
\end{proposition}

\begin{proof}
Fix $K \in \{ L, R \}$.
We first consider one direction of the equivalence in \eqref{eq:lr1}, for
a special case.
\begin{claimfirst} \label{cl:lr2}
Let $u$ and $v$ be nodes in $\Pi$ such that $v$ is a discharged leaf and
$u \in \itv{c(v)}{v}$.
Then $u$ and $v$ satisfy \eqref{eq:lr1}.
\end{claimfirst}
\begin{pfclaim}
Assume first that $\Si^K_u$ contains a formula in focus. Note that the
discharge rule is never applied on the path $\itv{c(v)}{v}$. We can thus
iteratively apply Proposition~\ref{p:lr1} backwards along the path
$\itv{c(v)}{u}$ to find that $\Si^K_{c(v)}$ contains a formula in focus.
But then the same applies to $\Si^K_v$: since $(\Si^L,\Si^R)$ is
balanced we have $\Si^K_v = \Si^K_{c(v)}$. For the other direction
assume that $\Si^K_v$ contains a formula in focus. Again with
Proposition~\ref{p:lr1} applied iteratively, now backwards along the
path $\itv{u}{v}$, we show that $\Si^K_u$ must contain a formula in
focus as well.
\end{pfclaim}

Finally, it is immediate by Claim~\ref{cl:lr2} and the definitions that
\eqref{eq:lr1} holds in case $u$ and $v$ are closely connected, and from
this an easy induction shows that \eqref{eq:lr1} holds as well if $u$
and $v$ are merely connected.
\end{proof}

\begin{proofof}{Proposition~\ref{p:itp3}}
Let $(\Si^{L},\Si^{R})$ be a balanced nodewise partition of some proof $\Pi$.
First define $\eta_{u} = \ntr$ for every node $u$ that does not lie on any path
to a discharged leaf from its companion node.

Then, consider any equivalence class $C$ of the connectedness relation defined 
in Definition~\ref{d:conn} such that $C \cap T_{\ntr} = \nada$, and make a case 
distinction.
If every node $u$ in $C$ is such that $\Si^{L}_{u}$ contains a formula in focus,
then we map all $C$-nodes to $\mu$.

If, on the other hand, some node $u$ in $C$ is such that $\Si^{L}_{u}$ contains
\emph{no} formula in focus, we reason as follows.
Since $\eta_{u} \neq \ntr$, $u$ must lie on some path to a non-axiomatic leaf 
$l$ from its companion node $c(l)$.
By the conditions on a successful proof, $\Si_{u}$ must contain \emph{some} 
formula in focus, and so this formula must belong to $\Si^{R}_{u}$.
It then follows from Proposition~\ref{p:itp4} that \emph{every} node in $C$
has a right formula in focus.
In this case we map all $C$-nodes to $\nu$.

With this definition it is straightforward to verify that $\eta$ is a fixpoint
colouring for $\Si^{L}\mid \Si^{R}$.
\end{proofof}

We will now see how we can read off interpolants from a balanced nodewise
partition and an associated fixpoint colouring.
Basically, the idea is that with every node of the proof we will associate a 
formula that can be seen as some kind of `preliminary' interpolant for the 
partition of the sequent of that node.

\begin{definition}
\label{d:itp}
Let $(\Si^{L},\Si^{R})$ be a balanced nodewise partition of some proof $\Pi$, 
and let $\eta$ be some fixpoint colouring for $(\Si^{L},\Si^{R})$. 
By induction on the depth of nodes we will associate a formula $\ip{s}$ with
every node $s$ of $\Pi$.
The bound variables of these formulas, if any, will be supplied by the 
discharge tokens used in $\Pi$.

For the definition of $\ip{s}$, inductively assume that $\ip{t}$ has already 
been defined for all proper descendants of $s$.
We distinguish cases depending on whether $s \in \Ran(c)$ and on whether $s$ 
is a discharged leaf:
\begin{description}
\item[Case $s \in \Dom(c)$.]
In this case we consider the discharge token $\dx_{c(s)}$ associated
with the companion of $s$ as a variable and define 
\[
 \ip{s} \isdef \dx_{c(s)}.
\]

\item[Case $s \not\in \Dom(c)$ and $s \not\in \Ran(c)$.]
Note that this case includes the situation where $s$ is an axiomatic leaf,
which is one of the base cases of the induction.

Let $\Ru = \Ru_{s}$ be the derivation rule applied at the node $s$, and 
assume that $s$ has successors $v_{0},\ldots,v_{n-1}$.
Let $\chi_{s}(x_{0},\ldots,x_{n-1})$ be the basic formula provided by 
Definition~\ref{d:locitp}.
Inductively we assume formulas $\ip{v_{i}}$ for all $i<n$, and so 
we may define
\[
 \ip{s} \isdef \chi_{s}(\ip{v_{0}},\ldots,
   \ip{v_{n-1}}).
\]

\item[Case] $s \in \Ran(c)$.
In this case the rule applied at $s$ is the discharge rule, with
discharge token $\dx_{s}$, $s$ has a unique child $s'$, and, obviously,
we have $\eta_{s} \in \{ \mu, \nu \}$.
We define
\[
 \ip{s} \isdef \eta_{s}\dx_{s} . \ip{s'}.
\]
In this case we bind the variable $\dx_{s}$, which was introduced at the leaves
discharged by $s$.
\end{description}
Finally we define
\[
 \fitp{\Pi} \isdef \ip{r},
\]
where $r$ is the root of $\Pi$.
\end{definition}

We will prove a number of statements about these interpolants $\ip{s}$, for
which we need some auxiliary definitions. 
We call a node $u$ a \emph{proper connected ancestor of $s$}, notation: 
$P_{c}^{+}us$, if $u$ is both connected to and a proper ancestor of $s$. 
For a node $s$ in $\Pi$ we then define
\[
\Vitp(s) \isdef \{ \dx_{u} \mid u \in \Ran(c) \text{ and } P^{+}_{c} us \}.
\]
Intuitively, $\Vitp(s)$ can be seen as the set of discharge tokens that may 
occur as free variables in the interpolant $\ip{s}$.
Furthermore, we call a node \emph{special} if it is not connected to its parent,
or if has no parent at all (that is, it is the root of $\Pi$).
Observe that in particular all nodes in $T_{\ntr}$ are special.

\begin{proposition}
\label{p:itp5}
The following hold for every node $s$ in $\Pi$:
\begin{urlist}
\item \label{it:clitp5-1}
if $\Ru_{s} \neq \RuDischarge{}$ then $\Vitp(s) = \Vitp(v)$ for every $v 
\in P(s)$ that is connected to $s$;
\item \label{it:clitp5-2}
if $\Ru_{s} = \RuDischarge{}$ then $\Vitp(s) = \Vitp(s')\setminus\{\dx_{s}\}$,
where $s'$ is the unique child of $s$;
\item \label{it:clitp5-3}
if $s$ is special then $\Vitp(s) = \nada$.
\end{urlist}
\end{proposition}

\begin{proof}
For item~\ref{it:clitp5-1}, the key observation is that if $\Ru_{s} \neq 
\RuDischarge{}$, and $v$ is connected to $s$, then $s$ and $v$ have exactly 
the same connected strict ancestors.
From this it is immediate that $\Vitp(s) = \Vitp(v)$.

In case $\Ru_{s} = \RuDischarge{}$, then $s$ is connected to its unique child
$s'$ --- here we use the fact that every application of the discharge rule
discharges at least one leaf, so that $s'$ actually lies on some path from
$s$ to a leaf of which $s$ is the companion.
But if $s$ and $s'$ are connected, then they have the same connected strict 
ancestors, with the obvious exception of $s$ itself.
From this item~\ref{it:clitp5-2} follows directly.

Item~\ref{it:clitp5-3} follows from the definition of $\Vitp(s)$ and the 
observation that if $s$ is special then it has no proper connected ancestors.
\end{proof}

Our next claim is that the interpolant $\fitp{\Pi}$ is of the right syntactic
shape, in that it is alternation free and only contains free variables that 
occur in both $\Sigma^L_r$ and $\Sigma^R_r$, where $r$ is the root of $\Pi$.

\begin{proposition}
\label{p:itp6}
The following hold for every node $s$ in $\Pi$:
\begin{urlist}
\item \label{it:clitp1-1}
$\FV(\ip{s}) \sse \Big(\FV(\Si^{L}_{s}) \cap \FV(\Si^{R}_{s})\Big)\cup\Vitp(s)$;
\item \label{it:clitp1-2}
$\ip{s} \in \nth{\eta_{s}}{\Vitp(s)}$ if $\eta_s \in \{\mu,\nu\}$;
\item \label{it:clitp1-3}
$\ip{s} \in \nth{\nu}{\nada} = \nth{\mu}{\nada}$ if $s$ is special.
\end{urlist}
\end{proposition}

\begin{proof}
We prove the first two items by induction on the depth of $s$ in $\Pi$, making 
the same case distinction as in Definition~\ref{d:itp}. 

\begin{description}
\item[Case] $s \in \Dom(c)$.
In this case $s$ is a discharged leaf, and we have $\ip{s} = \dx_{c(s)}$, so 
that $\FV(\ip{s}) = \{\dx_{c(s)}\} \subseteq \Vitp(s)$ because the companion 
$c(s)$ of $s$ must be a proper ancestor of $s$ and by definition $c(s)$ is 
connected to $s$. 
Moreover, we clearly find $\ip{s} \in \nth{\eta_{s}}{\Vitp(s)}$.

\item[Case] $s \not\in \Dom(c)$ and $s \not\in \Ran(c)$.
Assume that $t$ has children $v_{0},\ldots,v_{n-1}$, then we have 
$\ip{s} = \chi_{s}(\ip{v_{0}},\ldots,\ip{v_{n-1}})$, where 
$\chi_{s}(x_{0},\ldots,x_{n-1})$ is the basic formula provided by 
Definition~\ref{d:locitp}.

For item~\ref{it:clitp1-1} we now reason as follows:
   \begin{align*}
   \FV(\ip{s}) & = \bigcup_{i} \FV(\ip{v_{i}}) 
   & \text{(definition $\ip{s}$)}
   \\ & \sse \bigcup_{i} \Big(
        \big(\FV(\Si^{L}_{v_{i}}) \cap \FV(\Si^{R}_{v_{i}})\big)
        \cup \Vitp(v_{i}) \Big)
   & \text{(induction hypothesis)}
   \\ & \sse \bigcup_{i} \Big(
        \big(\FV(\Si^{L}_{v_{i}}) \cap \FV(\Si^{R}_{v_{i}})\big)
        \cup \Vitp(s) \Big)
   & \text{(Proposition~\ref{p:itp5}(\ref{it:clitp5-1})}
   \\ & \sse 
        \big(\FV(\Si^{L}_{s}) \cap \FV(\Si^{R}_{s})\big)
        \cup \Vitp(s) 
   & \text{(Proposition~\ref{p:locitp}(\ref{i:coh})},
   \end{align*}
which suffices to prove item \ref{it:clitp1-1}.

For item~\ref{it:clitp1-2} we first show that if $\eta_s \in \{\mu,\nu\}$ then
$\ip{s} \in \nth{\eta_{s}}{\Vitp(s)}$. 
Assume that $\eta_s \in \{\mu,\nu\}$. 
We claim that 
\begin{equation} \label{eq:thekidsarefine}
\ip{v_i} \in \nth{\eta_s}{\Vitp(s)} \quad \mbox{for all } i < n.
\end{equation}
To see that this is the case fix $i$ and distinguish cases depending on whether
$v_i$ is special or not. 
If $v_i$ is special then we reason as follows:
\begin{align*}
\FV(\ip{v_i}) & \sse 
\Big(\FV(\Si^{L}_{v_{i}}) \cap \FV(\Si^{R}_{v_{i}})\Big)\cup\Vitp(s)
   & \text{(induction hypothesis)}
\\ & = \FV(\Si^{L}_{v_{i}}) \cap \FV(\Si^{R}_{v_{i}})
   & \text{(Proposition~\ref{p:itp5}(\ref{it:clitp5-3})}
\\ & \subseteq \FV(\Si^{L}_{s}) \cap \FV(\Si^{R}_{s}),
   & \text{(Proposition~\ref{p:locitp}(\ref{i:coh})}
\end{align*}
so that $\FV(\ip{v_i}) \cap \Vitp(s) = \nada$.
From this \eqref{eq:thekidsarefine} is immediate by the definitions.

On the other hand, if $v_i$ is not special then by definition it is connected
to $s$. 
It follows that $\eta_{v_i} = \eta_s \in \{\mu,\nu\}$ and thus we obtain by
the inductive hypothesis that $\ip{v_i} \in \nth{\eta_s}{\Vitp(v_i)}$. 
But since $s \notin \Ran(c)$ we have $\Ru_s \neq \RuDischarge{}$ and so by
Proposition~\ref{p:itp5}(\ref{it:clitp5-2} we find $\Vitp(v_i) = \Vitp(s)$. 
This finishes the proof of \eqref{eq:thekidsarefine}.

To show that $\ip{s} \in \nth{\eta_s}{\Vitp(s)}$ recall that $\ip{s} =
\chi_s(\ip{v_0},\dots,\ip{v_{n-1}})$. 
Because of \eqref{eq:thekidsarefine} it suffices to check that
$\nth{\eta_s}{\Vitp(s)}$ is closed under the schema $\chi_s$. 
But since $\chi_{s}$ is a basic formula, this is immediate by the definitions.

\item[Case $s \in \Ran(c)$.]
In this case the rule applied at $s$ is the discharge rule, with discharge 
token $\dx_{s}$, $s$ has a unique child $s'$, $\eta_{s} \in \{
\mu, \nu \}$ and by definition $\ip{s} = \eta_{s}\dx_{s} . \ip{s'}$. 
To prove item~\ref{it:clitp1-1} we can then reason as follows:
   \begin{align*}
   \FV(\ip{s}) & = \FV(\ip{s'}) \setminus \{ \dx_{s} \}
   & \text{(definition $\ip{s}$)}
   \\ & \sse \Big(
        \big(\FV(\Si^{L}_{s'}) \cap \FV(\Si^{R}_{s'})\big)
        \cup \Vitp(s') \Big)
		\setminus \{ \dx_{s} \}
   & \text{(induction hypothesis)}
   \\ & \subseteq \Big(
        \big(\FV(\Si^{L}_{s}) \cap \FV(\Si^{R}_{s})\big)
        \cup \Vitp(s') \Big)
		\setminus \{ \dx_{s} \}
   & \text{($\Sigma^L_{s'} = \Sigma^L_s$ and $\Sigma^R_{s'} = \Sigma^R_s$)}
   \\ & \sse 
        \big(\FV(\Si^{L}_{s}) \cap \FV(\Si^{R}_{s})\big) 
		\cup \big(\Vitp(s') \setminus \{ \dx_{s} \} \big)
   & \text{(basic set theory)}
   \\ & =
        \big(\FV(\Si^{L}_{s}) \cap \FV(\Si^{R}_{s})\big) 
		\cup \Vitp(s) 
   & \text{(Proposition~\ref{p:itp5}(\ref{it:clitp5-2})}
   \end{align*}

To check item~\ref{it:clitp1-2}, note that $\eta_s \in \{\mu,\nu\}$, because 
$s$ itself is on the path from $s$ to any of the leaves that it discharges, and
that $\eta_{s'} = \eta_s$ because $s'$ is connected to $s$.
By the inductive hypothesis we find that $\ip{s'} \in \nth{\eta_s}{\Vitp(s')}$,
so that it is clear from the definitions that $\ip{s} \in \nth{\eta_s}{\Vitp(s') 
\setminus \{\dx_s\}}$. 
It follows that $\ip{s} \in \nth{\eta_s}{\Vitp(s)}$, since $\Vitp(s) = \Vitp(s')
\setminus \{\dx_s\}$ by Proposition~\ref{p:itp5}(\ref{it:clitp5-2}.
\end{description}
This finishes the proof of the first two items of the proposition.
\medskip

For item~\ref{it:clitp1-3}, let $s$ be special.
It is then immediate from item~\ref{it:clitp1-2} and 
Proposition~\ref{p:itp5}(\ref{it:clitp5-3} that $\ip{s} \in 
\nth{\eta_s}{\nada}$.
The statement then follows by the observation of 
Proposition~\ref{p:af2}(\ref{it:af2-2} that $\nth{\mu}{\nada} = \AFMC = 
\nth{\nu}{\nada}$.
\end{proof}

Proposition~\ref{p:itp7} is the key technical result of our proof.
In its formulation we need the following.

\begin{definition}
Let $\Pi = (T,P,\Si,\Ru)$ be some proof.
A \emph{global annotation} for $\Pi$ is a map $\ann: T \to \{ u,f \}$; the dual 
of the global annotation $\ann$ is the map $\dann$ given by  
\[
\dann(t) \isdef 
\left\{\begin{array}{ll}
   f & \text{if } \ann(t) = u 
\\ u & \text{if } \ann(t) = f.
\end{array}\right.
\]
A global annotation $\ann$ is \emph{consistent} with a fixpoint colouring
$\eta$ if it satisfies $\ann(t) = u$ if $\eta_t = \mu$ and $\ann(t) = f$ if 
$\eta_t = \nu$.
\end{definition}

Note that the conditions on an annotation $a$ to be consistent with a fixpoint
colouring $\eta$ only mentions the nodes in $T_{\mu}$ and $T_{\nu}$; the 
annotation $\ann(t)$ can be arbitrary for $t \in T_{\ntr}$.
\medskip

For the final part of the interpolation argument we need a general observation
about the result of applying a substitution to (all formulas in a) \emph{proof}.
First we need some definitions.

\begin{definition}
Let $\Sigma$ be an annotated sequent.
We define $\BV(\Sigma) = \bigcup\{\BV(\psi) \mid \psi^a \in \Sigma\}$, and, for
any formula $\phi$ such that $\FV(\phi) \cap \BV(\Sigma) = \nada$,
we set 
\[
 \Sigma[\phi / x] \isdef \{(\psi[\phi / x])^a \mid \psi^a \in
\Sigma\}.
\]
Furthermore, where $\Pi = (T,P,R,\Si)$ is some proof, we let $\Pi[\phi/x]$ 
denote the labelled tree $\Pi[\phi/x] \isdef (T,P,R,\Si')$ which is obtained
from $\Pi$ by replacing every annotated sequent $\Si_{t}$ with 
$\Sigma_{t}[\phi/x]$.
\end{definition}

\begin{proposition} \label{p:psubst}
Let $\Pi$ be a \Focus-proof of a sequent $\Sigma$ with open assumptions 
$\{\Gamma_i \mid i \in I\}$, and let $\phi$ be a formula such that $\FV(\phi) 
\cap \BV(\Sigma) = \nada$. 
Then $\Pi[\phi/x]$ is a well-formed \Focus-proof of the sequent 
$\Sigma[\phi/x]$, with open assumptions $\{\Gamma_i [\phi/x] \mid i \in I\}$.
\end{proposition}
\begin{proof} (Sketch)
One may show that $\BV(\chi) \subseteq \BV(\psi)$ for every $\chi \in 
\Clos(\psi)$, by an induction on the length of the trace from $\psi$ to $\chi$ 
witnessing that $\chi \in \Clos(\psi)$.
Because every formula $\chi$ that occurs in one of the sequents of $\Pi$ belongs
to the closure of $\Sigma$ it follows that $\BV(\chi) \subseteq \BV(\Sigma)$ 
and hence all the substitutions are well-defined. 
Moreover, one can check that all the proof rules remain valid if one performs 
the same substitution uniformly on all the formulas in the conclusion and the
premises. 
It should also be clear that the global conditions on proofs are not 
affected by the substitution.
\end{proof}

\begin{proposition}
\label{p:itp7}
Let $(\Si^{L},\Si^{R})$ be a balanced nodewise partition of some proof $\Pi$,
let $\eta$ be some fixpoint colouring for $(\Si^{L},\Si^{R})$, and let $\ann: 
T \to \{ u, f \}$ be a global annotation that is consistent with $\eta$.
Then we can effectively construct \Focus-proofs $\Pi^{L}$ and $\Pi^{R}$ of the sequents
$\Si^{L}_{r},(\fitp{\Pi})^{\ann(s)}$ and 
$\Si^{R}_{r},(\sneg{\fitp{\Pi}})^{\dann(s)}$, respectively, where $r$ is the root
of $\Pi$.
\end{proposition}

\begin{proof}
For every node $s$ of $\Pi$ we will construct two proofs with open assumptions,
$\Pi^{L}_{s}$ and $\Pi^{R}_{s}$, for the sequents
$\Si^{L}_{s},\ip{s}^{\ann(s)}$ and $\Si^{R}_{s},\sneg{\ip{s}}^{\dann(s)}$, 
respectively. 
We will make sure that the only
open assumptions of these proofs will be associated with leaves $l$ of
which the companion node $c(l)$ is a proper connected ancestor of $s$.
We define $\Pi^{L}_s$ and $\Pi^{R}_s$ as labelled trees that satisfy
conditions \ref{i:local condition}~and~\ref{i:leaf condition} from Definition~\ref{d:proof}. We
check the other conditions in subsequent claims. 
The definition of $\Pi^{L}_{s}$ and $\Pi^{R}_{s}$ proceeds by induction on the 
depth of $s$ in the tree $\Pi$, where we make the same case distinction as in 
Definition~\ref{d:itp}.

\begin{description}
\item[Case $s \in \Dom(c)$.]
In this case we let $\Pi^L_s$ and $\Pi^R_s$ be the leaves that are labelled
with the discharge variable $\dx_{c(s)}$ and the sequents $\Sigma^L_s,\ip{s} 
= \Sigma^L_s, \dx_{c(l)}^{\ann(l)}$ and $\Sigma^R_s,\sneg{\ip{s}} = 
\Sigma^R_s, \dx_{c(l)}^{\dann(l)}$, respectively. Note that here we are
creating an open assumption that is labelled with a discharge token and
not with $\star$. This open assumption will be discharged later when the
induction is at the node $c(s)$.

\item[Case $s \not\in \Dom(c)$ and $s \not\in \Ran(c)$.]
The basic strategy in this case is to use Proposition~\ref{p:locitp} to
extend the proofs $\Pi^L_s$ and $\Pi^R_s$. The details depend on the
global annotation $\ann$. We only consider the subcases where
$\ann(s)$ is distinct from $\ann(v)$ for at least one child $v$ of $s$.
The case where $\ann(s) = \ann(v)$ for all $v \in P(s)$ is similar, but
easier.

\begin{description}
\item[Subcase $\ann(s) = u$, but $\ann(v) = f$, for some $v \in P(s)$.]
As a representative example of this, consider the situation where
$\Ru_{s}$ is binary, and $\ann(s) = \ann(v_{0}) = u$, while $\ann(v_{1}) = f$,
where $v_{0}$ and $v_{1}$ are the two successors of $s$.

We first consider the proof $\Pi^L_{s}$.
Inductively we assume labelled trees $\Pi^{L}_{v_{0}}$ and $\Pi^{L}_{v_{1}}$ 
for, respectively, the sequents $\Si^{L}_{v_{0}},\ip{v_{0}}^{u}$ and 
$\Si^{L}_{v_{1}},\ip{v_{1}}^{f}$.
Combining these with the proof with assumptions $\Xi^{L}$ from 
Proposition~\ref{p:locitp}, we then define $\Pi^{L}_{s}$ to be the following    labelled tree:
   \begin{prooftree}
   \AXC{$\Pi^{L}_{v_{0}}$}
   \UIC{$\Si^{L}_{v_{0}},\ip{v_{0}}^{u}$}
   \AXC{$\Pi^{L}_{v_{1}}$}
   \UIC{$\Si^{L}_{v_{1}},\ip{v_{1}}^{f}$}
   \RightLabel{$\RuFocus$}
   \UIC{$\Si^{L}_{v_{1}},\ip{v_{1}}^{u}$}
   \BIC{$\Xi^{L}$}
   \UIC{$\Si^{L}_{s},\chi_{s}(\ip{v_{0}},\ip{v_{1}})^{u}$}
   \end{prooftree}

A similar construction works for $\Pi^R_{s}$: Inductively we are
given proofs $\Pi^{R}_{v_{0}}$ and $\Pi^{R}_{v_{1}}$ for, respectively,
the sequents $\Si^{R}_{v_{0}},\sneg{\ip{v_{0}}}^{f}$ and
$\Si^{R}_{v_{1}},\sneg{\ip{v_{1}}}^{u}$. Together with the proof $\Xi^{R}$
that we obtain from Proposition~\ref{p:locitp} we can define $\Pi^R_{s}$
as follows:
   \begin{prooftree}
   \AXC{$\Pi^{R}_{v_{0}}$}
   \UIC{$\Si^{R}_{v_{0}},\sneg{\ip{v_{0}}}^{f}$}
   \AXC{$\Pi^{R}_{v_{1}}$}
   \UIC{$\Si^{R}_{v_{1}},\sneg{\ip{v_{1}}}^{u}$}
   \RightLabel{$\RuUnfocus$}
   \UIC{$\Si^{R}_{v_{1}},\sneg{\ip{v_{1}}}^{f}$}
   \BIC{$\Xi^{R}$}
   \UIC{$\Si^{R}_{s},\sneg{\chi_{s}(\ip{v_{0}},\ip{v_{1}})}^{f}$}
   \end{prooftree}

\item[Subcase $\ann(s) = f$, but $\ann(v) = u$, for some $v \in P(s)$.]
Similarly as in the previous subcase, we consider a representative example 
where $s$ has two successors, $v_{0}$ and $v_{1}$, but now $\ann(s) = 
\ann(v_{0}) = f$, while $\ann(v_{1}) = u$.
Inductively we are provided with labelled trees $\Pi^{L}_{v_{0}}$ and 
$\Pi^{L}_{v_{1}}$ for, respectively, the sequents 
$\Si^{L}_{v_{0}},\ip{v_{0}}^{f}$ and 
$\Si^{L}_{v_{1}},\ip{v_{1}}^{u}$.
Combining these with the proof with assumptions $\Xi^{L}$, which we obtain by 
Proposition~\ref{p:locitp}, we then define $\Pi^{L}_{s}$ to be the following 
labelled tree:
	 
   \begin{prooftree}
   \AXC{$\Pi^{L}_{v_{0}}$}
   \UIC{$\Si^{L}_{v_{0}},\ip{v_{0}}^{f}$}
   \AXC{$\Pi^{L}_{v_{1}}$}
   \UIC{$\Si^{L}_{v_{1}},\ip{v_{1}}^{u}$}
   \RightLabel{$\RuUnfocus$}
   \UIC{$\Si^{L}_{v_{1}},\ip{v_{1}}^{f}$}
   \BIC{$\Xi^{L}$}
   \UIC{$\Si^{L}_{s},\chi_{s}(\ip{v_{0}},\ip{v_{1}})^{f}$}
   \end{prooftree}
Again, a similar construction works for $\Pi^R_{s}$.
\end{description}

\item[Case $s \in \Ran(c)$.]
In this case the rule applied at $s$ is the discharge rule; let $\dx_{s}$,
$s'$ and $\eta_{x}$ be as in the corresponding case in Definition~\ref{d:itp}.

Note that by the assumption on $\ann$ we have that $\ann(s) = \ann(s')$ and
$\ann(s) = \ann(l)$ for any discharged leaf $l$ such that $c(l) = s$.
Furthermore, there are only two possibilities: either $\ann(s) = u$ and 
$\eta_{s} = \mu$, or $\ann(s) = f$ and $\eta_{s} = \nu$.
We cover both cases at once but first only consider the definition of $\Pi^L_{s}$.
Inductively we have a proof $\Pi^L_{s'}$ of $\Sigma^L_{s'}, \ip{s'}^{\ann(s')}$.
Note that $\Sigma^L_{s'} = \Sigma^L_s$, because the discharge rule is applied at
$s$.
   
Let $(\Pi')^{L} \isdef \Pi^{L}_{s'}[\eta_{s}\dx_{s}.\ip{s'}/\dx_{s}]$;
that is, $(\Pi')^{L}$ is the labelled tree $\Pi^{L}_{s'}$, with all
occurrences of $\dx_{s}$ replaced by the formula
$\eta_{s}\dx_{s}.\ip{s'}$. That this is a well-defined operation on
proofs follows from Proposition~\ref{p:psubst}. However, we need to make
sure that $\FV(\eta_{s} \dx_{s}.\ip{s'}) \cap \BV(\Sigma^L_{s'},
\ip{s'}^{\ann(s')}) = \emptyset$. This follows with
item~\ref{it:clitp1-1} of Proposition~\ref{p:itp6} and the observations
that the variables in $\Vitp(s')$ do not occur as bound variables in any
of the formulas in $\Sigma^L_{s'}$ nor in $\ip{s'}$. Note that
$(\Pi')^{L}$ has the open assumption $\Sigma^L_{s}, (\eta_{s}\dx_{s}.
\ip{s'})^{\ann(s)}$ instead of $\Sigma^L_{s}, \dx_{s}^{\ann(s)}$.

To obtain $\Pi^{L}_{s}$ from $(\Pi')^{L}$, add one application of the fixpoint
rule for $\eta_{s}\dx_{s}.\ip{s'}$, followed by an application of the discharge
rule for the discharge token $\dx_{s}$:
   
   \begin{prooftree}
   \AxiomC{$[\Sigma^{L}_{s},\big(\eta_{s}\dx_{s}.\ip{s'}\big)^{\ann(s)}]^{\dx_{s}}$}
   \UnaryInfC{$(\Pi')^L$}
   \UnaryInfC{$\Sigma^L_{s}, 
   \big(\ip{s'}[\eta_{s} \dx_{s}.\ip{s'}/\dx_{s}]\big)^{\ann(s)}$}
   \RightLabel{$\Ru_{\eta_s}$}
   \UnaryInfC{$\Sigma^L_{s}, \big(\eta_{s} \dx_{s}.\ip{s'}\big)^{\ann(s)}$}
   \RightLabel{\RuDischarge{\dx_{s}}}
   \UnaryInfC{$\Sigma^L_s, \big(\eta_{s} \dx_{s}.\ip{s'}\big)^{\ann(s)}$}
   \end{prooftree}
The application of the rule $\Ru_{\eta_s}$ is correct because if $\eta_s
= \mu$ then $\ann(s) = u$. Thus, the unfolded fixpoint formula in the
premise of the application of $\Ru_{\eta_s}$ is still annotated with
$\ann(s)$. If $\eta_s = \nu$ then the unfolded fixpoint stays annotated
with $\ann(s)$ because $\Ru_\nu$ does not change the annotation of its
principal formula. Also note that the proof $\Pi^L_s$ no longer contains
open assumptions that are labelled with the token $\dx_{s}$.

A similar construction can be used to define $\Pi^R_{s}$.
By induction there is a proof $\Pi^R_{s'}$ of $\Sigma^R_{s'},
\sneg{\ip{s'}}^{\dann(s')}$. 
As before we use Proposition~\ref{p:psubst} to substitute all occurrences of 
$\dx_{s}$ with $\ol{\eta_{s}}\dx_{s}.\sneg{\ip{s'}}$ in the proof $\Pi^R_{s'}$ 
to obtain a proof $(\Pi')^{R} \isdef
\Pi^{R}_{s'}[\ol{\eta_{s}}\dx_{s}.\sneg{\ip{s'}}/\dx_{s}]$. Note that $(\Pi')^{R}$
has the open assumption
$\Sigma^R_{s},(\ol{\eta_{s}}\dx_{s}.\sneg{\ip{s'}})^{\dann(s)}$ instead
of $\Sigma^R_{s}, \dx_{s}^{\dann(s)}$. We then construct the proof
$\Pi^R_s$ as follows:
   \begin{prooftree}
   \AxiomC{$[\Sigma^{R}_{s},\big(\ol{\eta_{s}}\dx_{s}.\sneg{\ip{s'}}\big)^{\dann(s)}]^{\dx_{s}}$}
   \UnaryInfC{$(\Pi')^R$}
   \UnaryInfC{$\Sigma^R_{s}, 
   \big(\sneg{\ip{s'}}[\ol{\eta_{s}}
\dx_{s}.\sneg{\ip{s'}}/\dx_{s}]\big)^{\dann(s)}$}
   \RightLabel{$\Ru_{\ol{\eta_s}}$}
   \UnaryInfC{$\Sigma^R_{s}, \big(\ol{\eta_{s}}
\dx_{s}.\sneg{\ip{s'}}\big)^{\dann(s)}$}
   \RightLabel{\RuDischarge{\dx_{s}}}
   \UnaryInfC{$\Sigma^R_s, \big(\ol{\eta_{s}}
\dx_{s}.\sneg{\ip{s'}}\big)^{\dann(s)}$}
   \end{prooftree}
Note that if $\ol{\eta_s} = \mu$ then $\eta_s = \nu$, $\ann(s) = f$ and
$\dann(s) = u$. Therefore, the application of the rule $\Ru_\mu$ above
has the right annotation at the unfolded fixpoint.
\end{description}

We now check that $\Pi^L_r$ and $\Pi^R_r$ are indeed \Focus-proofs of,
respectively, the sequents $\Sigma^L_r,\ip{r}^{\ann(s)}$ and 
$\Sigma^R_r,\sneg{\ip{r}}^{\dann(r)}$, where $r$ is the root of $\Pi$.
Note that whereas we are proving statements about $\Pi^L_r$ and $\Pi^R_r$,
our proof is by induction on the complexity of the original proof $\Pi$.
In the formulation of the inductive hypothesis it is convenient to allow for
proofs in which some open assumptions are already labelled with a
discharge token instead of with $\star$. 
(In the end of the induction this makes no difference because $\Pi^L_r$ and 
$\Pi^R_r$ do not have any open assumption.) 
With this adaptation we will establish the claim below.

Before going into the details we observe that, given the inductive definition 
of the proof $\Pi^{L} = \Pi^{L}_{r}$, it contains, for every node $s$ in $\Pi$,
some substitution instance of $\Pi^{L}_{s}$ as a subproof.
In particular, we may assume the existence of an injection $f^L$ mapping 
$\Pi$-nodes to $\Pi^{L}$-nodes, in such a way that $f^L(s)$ is the root of the 
proof tree $\Pi^{L}_{s}$, for every node $s$ of $\Pi$.
A similar observation holds for the proof $\Pi^{R}$.

\begin{claimfirst}
\label{cl:itp2}
For all nodes $s$ in $\Pi$ the following hold.
\begin{urlist}
\item 
$\Pi^{L}_{s}$ is a \Focus-proof for the sequent $\Si^{L}_{s},
\ip{s}^{\ann(s)}$, with assumptions $\{ \Si^{L}_{l},\dx_{c(l)}^{\ann(l)}
\mid P^{+}c(l)s \text{ and } P^{*}sl \}$ such that additionally for
every node $t'$ that is on a path from the root $f^L(s)$ of $\Pi^L_s$ to one of
its open assumptions the following hold:
\begin{enumerate}
\item \label{it:itpcl-1}
the annotated sequent at $t'$ contains at least one formula that is in focus;
\item \label{it:itpcl-2}
the rule applied at $t'$ is not $\RuFocus$ or $\RuUnfocus$;
\item \label{it:itpcl-3}
if $t' = f^L(s')$ and \RuBox is applied at $s'$ then \RuBox is applied at
$t'$.
\end{enumerate}

\item
$\Pi^{R}_{s}$ is a \Focus-proof for the sequent $\Si^{R}_{s},
\sneg{\ip{s}}^{\dann(s)}$, with assumptions $\{
\Si^{R}_{l},\dx_{c(l)}^{\dann(l)} \mid P^{+}c(l)s \text{ and } P^{*}sl
\}$ such that additionally for
every node $t'$ that is on a path from the root of $\Pi^R_s$ to one of
its open assumptions it holds that:
\begin{enumerate}
 \item the annotated sequent at $t'$ contains at least one
formula that is in focus;
 \item the rule applied at $t'$ is not $\RuFocus$ or
$\RuUnfocus$;
\item 
if $t' = f^R(s')$ and \RuBox is applied at $s'$ then \RuBox is applied at
$t'$.
\end{enumerate}
\end{urlist}
\end{claimfirst}

\begin{pfclaim}
As mentioned, our argument proceeds by induction on the complexity of the proof
$\Pi$, or, to be somewhat more precise, by induction on the depth of $s$ in
$\Pi$.
Here we will use the same case distinction as the construction of $\Pi^L_s$ 
and $\Pi^R_s$. 
We focus on the proof $\Pi^{L}$, the case of $\Pi^{R}$ being similar.

First we make an auxiliary observation that will be helpful for understanding 
our proof:
\begin{equation} \label{eq:core}
\text{if } s \in T_{\mu} \cup T_{\nu} \text{ then $f^L(s)$ contains a formula 
in focus in } \Pi^{L}_{s}.
\end{equation}
For a proof of this, first assume that $s \in T_{\mu}$, i.e., $\eta_{s} = \mu$.
Then $\Si^{L}_{s}$ contains a formula in focus by item~\ref{i:in focus} of
Definition~\ref{d:coloring}.
On the other hand, if $s \in T_{\nu}$, then since the annotation $\ann$ is 
consistent with $\eta$, we have $\ann(s) = f$, so that the formula 
$\ip{s}^{\ann(s)}$ is in focus.
\medskip

Now we turn to the inductive proof of the claim proper.
It is obvious from the construction that the root $f(s)$ of $\Pi^L_s$ is labelled
with the annotated sequent $\Si^{L}_{s}, \ip{s}^{\ann(s)}$, and it is not 
hard to see that the open assumptions of this proof are indeed of the form
claimed above. 
To show that $\Pi^L_s$ is indeed a \Focus-proof we need to check the 
conditions from Definition~\ref{d:proof}.

Condition~\ref{i:local condition}, which requires the annotated sequents
to match the applied proof rule at every node, can be easily verified by
inspecting the nodes that are added in each step of the construction of
$\Pi^L_s$.
Similarly, it is clear that only leaves get labelled
with discharge tokens and thus condition~\ref{i:leaf condition} is satisfied.

It is also not too hard to see that all non-axiomatic leaves that are not
open assumptions are discharged. 
This is just our (already established) claim that all open assumptions of
$\Pi^L_s$ are in the set $\{ \Si^{L}_{l},\dx_{c(l)}^{\ann(l)} \mid P^{+}c(l)s 
\text{ and } P^{*}sl \}$. 
This means that condition~\ref{i:discharge condition} is satisfied. 
(Note that it is here where we conveniently allow for open leaves that are 
labelled with a discharge token rather than with $\star$.)

It is left to consider condition~\ref{i:path condition}. We have to
consider any path between a leaf $l'$ and its companion $c(l')$ in
$\Pi^{L}_s$. We can focus on the case, where $c(l')$ is the root
$f^L(s)$ of $\Pi^L_s$; in later steps of the induction the labels of the
node only get changed by substitutions of formulas for the open fixpoint
variables, which by Proposition~\ref{p:psubst} does not affect
condition~\ref{i:path condition}. Note then that $l' = f^L(l)$ for
some leaf $l$ of $\Pi$ with $c(l) = s$ and $c(l') = f^L(s)$. The path
from $s$ to $l$ in $\Pi$ satisfies condition~\ref{i:path condition}
because $\Pi$ is a \Focus-proof. That the path from $f^L(s)$ to $l'$
satisfies condition~\ref{i:path condition} follows from the
statements~\eqref{it:itpcl-1}, \eqref{it:itpcl-2} and \eqref{it:itpcl-3}
that we are about to prove.

\medskip

To prove the parts~\eqref{it:itpcl-1}, \eqref{it:itpcl-3} and
\eqref{it:itpcl-3} of the inductive statement, let $t'$ be a node on a
path from the root $f(s)$ of $\Pi^L_s$ to one of its open assumptions.
We now make our case distinction.

\begin{description}
\item[Case $s \in \Dom(c)$.]
In this case $\Pi^{L}_{s}$ contains $f^L(s)$ as its single node, and so
\eqref{it:itpcl-1} follows by \eqref{eq:core}, while \eqref{it:itpcl-2}
and \eqref{it:itpcl-3} are obvious by construction.

\item[Case $s \not\in \Dom(c)$ and $s \not\in \Ran(c)$.]
Let $v_{0},\ldots,v_{n-1}$ be the children of $s$ (in $\Pi$).
Then by construction $\Pi^{L}_{s}$ consists of the pre-proofs $\Pi^{L}_{v_{0}},
\ldots,\Pi^{L}_{v_{n-1}}$, linked to the root $f^L(s)$ via an instance $\Xi^{L}$ 
of Proposition~\ref{p:locitp}, in such a way that (i) all open leafs of 
$\Pi^{L}_{s}$ belong to one of the $\Pi^{L}_{v_{i}}$ where $s$ and $v_{i}$ are
connected, and (ii) $\Pi^{L}_{v_{i}}$ is \emph{directly} pasted to the 
corresponding leaf of $\Xi^{L}$ in case $s$ and $v_{i}$ are connected
(that is, no focus or unfocus rule are needed).
Concerning the position of the node $t'$ in $\Pi^{L}_{s}$, it follows from (i)
and (ii) that there is a child $v = v_{i}$ of $s$, which is connected to $s$ 
and such that $t'$ either lies (in the $\Pi^{L}_{v}$-part of $\Pi^{L}_{s}$) on
the path from $f^L(v)$ to an open leaf, or on the path in $\Pi^{L}_{s}$ from
$f^L(s)$ to $f^L(v)$.
Since the first case is easily taken care of by the inductive hypothesis, we 
focus on the latter.
It follows from (ii) that the full path from $f^L(s)$ to $f^L(v)$ is taken from 
the pre-proof $\Xi^{L}$ as provided by Proposition~\ref{p:locitp}.
But then \eqref{it:itpcl-1}, \eqref{it:itpcl-2} and \eqref{it:itpcl-3} are immediate by
item~\ref{i:itptrf}(a), (b) and (c) from mentioned proposition, given the fact that 
by \eqref{eq:core} the node $f^L(v)$ features a formula in focus.
(Note that the rule applied at $s$ in $\Pi^{L}_{s}$ is not the focus rule since 
$s \in T_{\mu} \cup T_{\nu}$ and thus $\Si_{s}$ contains a formula in 
focus.)

\item[Case $s \in \Ran(c)$.]
Let $s^{+}$ be the unique successor of $s$ in $\Pi$. Then by
construction $\Pi^{L}_{s}$ consists of a substitution instance of
$\Pi^{L}_{s^{+}}$, connected to $f^L(s)$ via the application of the
rules $\Ru_{\eta_s}$ (at the unique successor of $f^L(s)$) and
\RuDischarge{\dx_{s}} (at $f^L(s)$ itself). Clearly then there are two
possible locations for the node $t'$. If $t'$ is situated in the subtree
rooted at $f^L(s^{+})$, then \eqref{it:itpcl-1} and \eqref{it:itpcl-2}
follow from the inductive hypothesis (note that when we apply a
substitutions to the derivation $\Pi^{L}_{s^{+}}$ we do not change the
proof rules or alter the annotations). On the other hand, the only two
nodes of $\Pi^{L}_{s}$ that do not belong to mentioned subtree are
$f^L(s)$ itself and its unique child. These nodes carry the same sequent
label, and so in this case \eqref{it:itpcl-1} follows from
\eqref{eq:core}. Finally, \eqref{it:itpcl-2} and \eqref{it:itpcl-3} are
obvious since we already saw that the rules applied in $\Pi^{L}_{s}$ at
$f^L(s)$ and its successor are $\RuDischarge{\dx_{s}}$ and
$\Ru_{\eta_s}$, respectively.
\end{description}
This finishes the proof of the claim.
\end{pfclaim}
 
Finally, the proof of the Proposition is immediate by these claims if we
consider the case $s = r$, where $r$ denotes the root of the tree.
\end{proof}

We close this section with an example that illustrates the computation of
the interpolant:

\newcounter{nodecounter}
\renewcommand{\thenodecounter}{(\alph{nodecounter})}
\newcommand{\node}{\refstepcounter{nodecounter}\thenodecounter\ }
\newcommand{\se}{,\,}
\newcommand{\interpol}[1]{}

\newcommand{\eq}[1]{}

\begin{example}
In this part of the appendix we discuss an example in which we compute
an interpolant by induction on the complexity of a \Focus-proof. The
example is the interpolant for the implication
\begin{equation} \label{eq:implication}
(\alpha(p) \rightarrow p) \rightarrow (\alpha(q) \lor q),
\end{equation}
 where $\alpha(p)$ is the following formula:
\begin{align*}
\alpha(p) & = \mu x . \psi_1(p) \lor \psi_2(p) \lor \psi_3(p) \lor
\varphi \lor \Diamond x \\
\psi_1(p) & = (p \land \Diamond p) \lor (\atneg{r} \land \Diamond p)
\lor (\atneg{p} \land r \land \Box \atneg{p}) \\
\psi_2(p) & = p \land \atneg{r} \\
\psi_3(p) & = \Diamond \atneg{p} \land \Diamond{p} \\
\varphi & = \nu x . \Box (r \land x)
\end{align*}
This example is based on the example provided in \cite{stud:ckbp09},
which is in turn based on an earlier example by \cite{maks:temp91}, to
show that epistemic logic with common knowledge does not have Craig
interpolation. If substitutes the formula $\mu x . \Diamond(\atneg{s}
\land x)$ for the propositional letter $r$ in the definition of $\alpha$
then one obtains the translations of the formulas from
\cite{stud:ckbp09} to the alternation-free $\mu$-calculus. We will see
that the interpolant of \eqref{eq:implication} can be expressed in the
alternation-free $\mu$-calculus.

\begin{figure}
\begin{prooftree}
\def\fCenter{\: \mid \:}

\Axiom$\node \label{topmost} [\varphi^f \se \atneg{p}\se \Diamond p\se \atneg{r}\se
\Diamond \alpha(p) \fCenter q \se \alpha(q)]^\dx
\interpol{\dx}$

\RightLabel{\RuAnd, \AxLit}
\UnaryInf$\node \label{second} (r \land \varphi)^f \se \atneg{p}\se \Diamond p\se
\atneg{r}\se \Diamond \alpha(p) \fCenter q \se \alpha(q)
\interpol{\bot \lor \dx \eq{\dx}}$

\RightLabel{\RuAnd, \AxLit}
\UnaryInf$\node \label{third} (r \land \varphi)^f \se \atneg{p} \se p \land \Diamond p \se p \land
\atneg{r} \se \Diamond \alpha(p) \fCenter q \se \alpha(q)
\interpol{\bot \lor (\bot \lor \dx) \eq{\dx}}$

\RightLabel{\RuWeak}
\UnaryInf$(r \land \varphi)^f \se \atneg{p}\se \psi_1(p)\se \psi_2(p)\se
\psi_3(p) \se \Diamond \alpha(p) \fCenter q \se \alpha(q)
\interpol{\dx}$

\RightLabel{\RuMu, \RuOr}
\UnaryInf$(r \land \varphi)^f \se \atneg{p} \se \alpha(p) \fCenter q \se
\alpha(q)
\interpol{\dx}$

\RightLabel{\RuBox}
\UnaryInf$\node \label{firstbox} \Box (r \land \varphi)^f \se \Diamond \atneg{p} \se \Diamond
\alpha(p) \fCenter \Diamond q \se \Diamond \alpha(q)
\interpol{\Diamond \dx}$

\RightLabel{\RuWeak}
\UnaryInf$p \se \Box (r \land \varphi)^f \se  \Box \atneg{p} \se \Diamond
\atneg{p} \se \Diamond \alpha(p) \fCenter \atneg{q} \se
\Diamond q \se \atneg{r} \se \Diamond \alpha(q)
\interpol{\Diamond \dx}$

\RightLabel{\RuNu}
\UnaryInf$p \se \varphi^f \se  \Box \atneg{p} \se
\Diamond \atneg{p} \se \Diamond \alpha(p) \fCenter \atneg{q} \se
\Diamond q \se \atneg{r} \se \Diamond \alpha(q)
\interpol{\Diamond \dx}$

\RightLabel{\RuAnd, \AxLit}
\UnaryInf$\node \label{addnegr} p \se (r \land \varphi)^f \se  \Box \atneg{p} \se \Diamond
\atneg{p} \se \Diamond \alpha(p) \fCenter \atneg{q} \se
\Diamond q \se \atneg{r} \se \Diamond \alpha(q)
\interpol{\atneg{r} \land \Diamond \dx}$

\RightLabel{\RuAnd, \RuWeak, \RuBox, \AxLit}
\UnaryInf$\node \label{complexand} p \se (r \land \varphi)^f \se  \Box \atneg{p} \se \Diamond
\atneg{p} \land \Diamond p \se \Diamond \alpha(p) \fCenter \atneg{q} \se
\Diamond q \se \atneg{r} \se \Diamond \alpha(q)
\interpol{(\atneg{r} \land \Diamond \dx) \lor \Diamond \bot
\eq{\atneg{r} \land \Diamond \dx}}$

\RightLabel{\RuAnd, \AxLit}
\UnaryInf$\node \label{morenegr} p \se (r \land \varphi)^f \se \atneg{p} \land r \land \Box
\atneg{p} \se \Diamond \atneg{p} \land \Diamond p \se \Diamond \alpha(p)
\fCenter \atneg{q} \se \Diamond q \se \atneg{r} \se \Diamond \alpha(q)
\interpol{\bot \lor (\atneg{r} \lor (\atneg{r} \land \Diamond \dx))
\eq{\atneg{r} \lor \Diamond \dx}}$

\RightLabel{\RuWeak}
\UnaryInf$p \se (r \land \varphi)^f \se \psi_1(p)\se \psi_2(p)\se \psi_3(p)\se
\Diamond \alpha(p) \fCenter \atneg{q} \se \Diamond q \se \atneg{r} \se
\Diamond \alpha(q)
\interpol{\atneg{r} \lor \Diamond \dx}$

\RightLabel{\RuMu, \RuOr}
\UnaryInf$p \se (r \land \varphi)^f \se \alpha(p) \fCenter \atneg{q} \se
\Diamond q \se \atneg{r} \se \Diamond \alpha(q)
\interpol{\atneg{r} \lor \Diamond \dx}$

\RightLabel{\RuAnd, \AxLit}
\UnaryInf$p \se (r \land \varphi)^f \se \alpha(p) \fCenter \atneg{q} \se
q \land \Diamond q \se q \land \atneg{r} \se \Diamond \alpha(q)
\interpol{\top \land (\top \land (\atneg{r} \lor \Diamond \dx))
\eq{\atneg{r} \lor \Diamond \dx}}$

\RightLabel{\RuWeak}
\UnaryInf$p \se (r \land \varphi)^f \se \alpha(p) \fCenter \atneg{q} \se
\psi_1(q)\se \psi_2(q)\se \psi_3(q)\se \Diamond \alpha(q)
\interpol{\atneg{r} \lor \Diamond \dx}$

\RightLabel{\RuMu, \RuOr}
\UnaryInf$p \se (r \land \varphi)^f \se \alpha(p) \fCenter \atneg{q} \se
\alpha(q)
\interpol{\atneg{r} \lor \Diamond \dx}$

\RightLabel{\RuBox}
\UnaryInf$\node \label{secondbox} \Diamond p \se \Box (r \land \varphi)^f \se \Diamond
\alpha(p) \fCenter \Diamond \atneg{q} \se \Diamond \alpha(q)
\interpol{\Diamond(\atneg{r} \lor \Diamond \dx)}$

\RightLabel{\RuWeak}
\UnaryInf$\atneg{p} \se \Diamond p\se \atneg{r}\se \Box(r \land
\varphi)^f \se \Diamond \alpha(p) \fCenter q \se \Box \atneg{q} \se \Diamond
\atneg{q} \se \Diamond \alpha(q)
\interpol{\Diamond(\atneg{r} \lor \Diamond \dx)}$

\RightLabel{\RuNu}
\UnaryInf$\atneg{p} \se \Diamond p\se \atneg{r}\se \varphi^f \se \Diamond
\alpha(p) \fCenter q \se \Box \atneg{q} \se \Diamond \atneg{q} \se \Diamond \alpha(q)
\interpol{\Diamond(\atneg{r} \lor \Diamond \dx)}$

\RightLabel{\RuAnd, \RuWeak, \RuBox, \AxLit}
\UnaryInf$\node \label{morecomplexand} \atneg{p} \se \Diamond p\se \atneg{r}\se \varphi^f \se \Diamond
\alpha(p) \fCenter q \se \Box \atneg{q} \se \Diamond \atneg{q} \land
\Diamond q \se \Diamond \alpha(q)
\interpol{\Diamond (\atneg{r} \lor \Diamond \dx) \land \Box \top
\eq{\Diamond (\atneg{r} \lor \Diamond \dx)}}$

\RightLabel{\RuAnd, \AxLit}
\UnaryInf$\node \label{addr} \atneg{p} \se \Diamond p\se \atneg{r}\se \varphi^f \se
\Diamond \alpha(p) \fCenter q \se \atneg{q} \land r \land \Box \atneg{q}
\se \Diamond \atneg{q} \land \Diamond q \se \Diamond \alpha(q)
\interpol{\top \land (r \land \Diamond (\atneg{r} \lor \Diamond \dx))
\eq{r \land \Diamond (\atneg{r} \lor \Diamond \dx)}}$

\RightLabel{\RuWeak}
\UnaryInf$\atneg{p} \se \Diamond p\se \atneg{r}\se \varphi^f \se \Diamond
\alpha(p) \fCenter q \se \psi_1(q)\se \psi_2(q)\se \psi_3(q)\se
\Diamond \alpha(q)
\interpol{r \land \Diamond (\atneg{r} \lor \Diamond \dx)}$

\RightLabel{\RuMu, \RuOr}
\UnaryInf$\atneg{p}\se \Diamond p\se \atneg{r}\se \varphi^f \se \Diamond
\alpha(p) \fCenter q \se \alpha(q)
\interpol{r \land \Diamond (\atneg{r} \lor \Diamond \dx)}$

\RightLabel{\RuDischarge{\dx}}
\UnaryInf$\node \label{discharge} \atneg{p}\se \Diamond p\se \atneg{r}\se \varphi^f \se \Diamond
\alpha(p) \fCenter q \se \alpha(q)
\interpol{\mu \dx . r \land \Diamond (\atneg{r} \lor \Diamond \dx)}$

\RightLabel{\RuFocus,\RuUnfocus}
\UnaryInf$\atneg{p}^f\se \Diamond p\se \atneg{r}\se \varphi\se \Diamond
\alpha(p) \fCenter q^f \se \alpha(q)^f
\interpol{\mu \dx . r \land \Diamond (\atneg{r} \lor \Diamond \dx)}$

\RightLabel{\RuAnd, \AxLit}
\UnaryInf$\atneg{p}^f\se p \land \Diamond p\se p \land \atneg{r} \se \Diamond \alpha(p) \fCenter q^f \se \alpha(q)^f
\interpol{\bot \lor (\bot \lor \mu \dx . r \land \Diamond (\atneg{r}
\lor \Diamond \dx)) \eq{\mu \dx . r \land \Diamond (\atneg{r} \lor
\Diamond \dx)}}$

\RightLabel{\RuWeak}
\UnaryInf$\atneg{p}^f\se \psi_1(p)\se \psi_2(p)\se \psi_3(p)\se \varphi\se
\Diamond \alpha(p) \fCenter q^f \se \alpha(q)^f
\interpol{\mu \dx . r \land \Diamond (\atneg{r} \lor \Diamond \dx)}$

\RightLabel{\RuMu, \RuOr}
\UnaryInf$\atneg{p}^f \se \alpha(p)^f \fCenter q^f \se \alpha(q)^f
\interpol{\mu \dx . r \land \Diamond (\atneg{r} \lor \Diamond \dx)}$

\RightLabel{\RuOr}
\UnaryInf$\node \label{root} (\atneg{p} \lor \alpha(p))^f \fCenter (q \lor \alpha(q))^f
\interpol{\mu \dx . r \land \Diamond (\atneg{r} \lor \Diamond \dx)}$
\end{prooftree}
\caption{A \Focus-proof of \eqref{eq:implication}}
\label{fig:example}
\end{figure}

Figure~\ref{fig:example} contains a \Focus-proof of the implication from
$\alpha(p) \rightarrow p$ to $\alpha(q) \lor q$. All the sequents in
this proof are already partitioned. At many steps we apply multiple
proof rules or apply the same rules multiple times. For instance at the
node labelled with \ref{complexand}, moving toward the node labeled with
\ref{addnegr}, we first apply the rule \RuAnd to the formula $\Diamond
\atneg{p} \land \Diamond p$. This splits the proof into two branches.
The left branch for the residual formula $\Diamond \atneg{p}$ is the
node labeled with \ref{addnegr}. The right branch for the residual
formula $\Diamond p$ is not written out. It continues with an
application of weakening to reduce the sequent to $\Box \atneg{p},
\Diamond p \mid$. On this branch the proof continues with an application
of \RuBox follows by \AxLit. We leave it to the reader to reconstruct
these details for all other nodes of the proof in
Figure~\ref{fig:example}.

\begin{figure}
\begin{center}
\begin{tabular}{rrl}
node & interpolant  & \phantom{$\equiv$} simplification \\
\ref{topmost} & $\dx$ & \\
\ref{second} & $\bot \lor \dx$ & $\equiv \dx$ \\
\ref{third} & $\bot \lor (\bot \lor \dx)$ & $\equiv \dx$ \\
\ref{firstbox} & $\Diamond \dx$ & \\
\ref{addnegr} & $\atneg{r} \land \Diamond \dx$ & \\
\ref{complexand} & $(\atneg{r} \land \Diamond \dx) \lor \Diamond \bot$
& $\equiv \atneg{r} \land \Diamond \dx$ \\
\ref{morenegr} & $\bot \lor (\atneg{r} \lor (\atneg{r} \land \Diamond
\dx))$ & $\equiv \atneg{r} \lor \Diamond \dx$ \\
\ref{secondbox} & $\Diamond(\atneg{r} \lor \Diamond \dx)$ & \\
\ref{morecomplexand} & $\Diamond (\atneg{r} \lor \Diamond \dx) \land \Box \top$
& $\equiv \Diamond (\atneg{r} \lor \Diamond \dx)$ \\
\ref{addr} & $\top \land (r \land \Diamond (\atneg{r} \lor \Diamond
\dx))$ & $\equiv r \land \Diamond (\atneg{r} \lor \Diamond \dx)$ \\
\ref{discharge} & $\mu \dx . r \land \Diamond (\atneg{r} \lor \Diamond
\dx)$ & \\
\ref{root} & $\mu \dx . r \land \Diamond (\atneg{r} \lor \Diamond
\dx)$ & \\
\end{tabular}
\end{center}
\caption{Interpolant computed from the proof in Figure~\ref{fig:example}}
\label{fig:table}
\end{figure}
Following Definitions \ref{d:locitp}~and~\ref{d:itp}, we can compute the
interpolant of \eqref{eq:implication} by induction over the proof in
Figure~\ref{fig:example}. The most important steps of this computation
are in the table of Figure~\ref{fig:table}.
At some nodes we rewrite the interpolant into a simpler equivalent
formula, and then continue the computation with the simplified version
of the interpolant. The formula $\mu \dx . r \land \Diamond (\atneg{r}
\lor \Diamond \dx)$ at the root node \ref{root} is the interpolant of
$\alpha(p) \rightarrow p$ and $\alpha(q) \lor q$.
\end{example}

\section{Conclusion \& Questions}

In this paper we saw that the idea of placing formulas in \emph{focus} can be
extended from the setting of logics like \textsc{ltl} and 
\textsc{ctl}~\cite{lang:focu01} to that of the alternation-free modal 
$\mu$-calculus: we designed a very simple and natural, cut-free sequent system
which is sound and complete for all validities in the language 
consisting of all (guarded) formulas in the alternation-free fragment $\AFMC$ 
of the modal $\mu$-calculus.
We then used this proof system $\Focus$ to show that the alternation-free 
fragment enjoys the Craig Interpolation Theorem.
Clearly, both results add credibility to the claim that $\AFMC$ is an interesting 
logic with good meta-logical properties.
\medskip

\noindent
Below we list some directions for future research.

\begin{enumerate}
\item
Probably the most obvious question is whether the restriction to guarded formulas
can be lifted. 
In fact, we believe that the focus proof system, possibly with some minor 
modifications in the definition of a proof, is also sound and complete for the 
full alternation-free fragment.
To prove this observation, one may bring ideas from Friedmann \& 
Lange~\cite{frie:deci13} into our definition of tableaux and tableau games.
\item
Another question is whether we may tidy up the focus proof system, in the same 
way that Afshari \& Leigh did with the Jungteerapanich-Stirling 
system~\cite{afsh:cutf17,jung:tabl10,stir:tabl14}.
As a corollary of this it should be possible to obtain an annotation-free
sequent system for the alternation-free fragment of the $\mu$-calculus, and to 
prove completeness of Kozen's (Hilbert-style) axiomatisation for $\AFMC$.
\item
Moving in a somewhat different direction, we are interested to see to which 
degree the focus system can serve as a basis for sound and complete derivation
systems for the alternation-free validities in classes of frames satisfying 
various kinds of frame conditions.
\item
We think it is of interest to see which other fragments of the modal 
$\mu$-calculus enjoy Craig interpolation.
A very recent result by L.~Zenger~\cite{zeng:proo21} shows that the fragments
$\Sigma^{\mu}_{1}$ and $\Pi^{\mu}_{1}$ consisting of, respectively, the 
$\mu$-calculus formulas that \emph{only} contain least- or greatest fixpoint
operators, each have Craig interpolation.
Clearly, a particular interesting question would be whether our focus system can 
be used to shed some light on the interpolation problem for propositional dynamic
logic (see the introduction for some more information) and other fragments of the
alternation-free $\mu$-calculus.
Looking at fragments of the modal $\mu$-calculus that are \emph{more} expressive
than $\AFMC$, an obvious question is whether \emph{every} bounded level of the
alternation hierarchy admits Craig interpolation.
\item
Finally, the original (uniform) interpolation proof for the full $\mu$-calculus 
is based on a direct automata-theoretic construction~\cite{dago:logi00}. 
Is something like this possible here as well? 
That is, given two modal automata $\bbA_{\phi}$ and $\bbA_{\psi}$ corresponding
to $\AFMC$-formulas $\phi$ and $\psi$, can we directly construct a modal 
automaton $\bbB$ which serves as an interpolant for $\bbA_{\phi}$ and 
$\bbA_{\psi}$ (so that we may obtain an $\AFMC$-interpolant for $\phi$ and $\psi$
by translating the automaton $\bbB$ back into $\AFMC$)?
Recall that the automata corresponding to the alternation-free $\mu$-calculus are 
so-called \emph{weak} modal parity 
automata~\cite{mull:alte92,carr:powe20}.
\end{enumerate}

\bibliographystyle{plain}
\bibliography{references,extra}

\appendix
\section{Infinite games}
\label{sec:games}

In this brief appendix we give the basic definitions of infinite two-player 
games.
We fix two players that we shall refer to as $\eloi$ (female) and $\abel$
(male).

A {\em two-player game} is a quadruple $\bbG = (V,E,\Own,\WC)$ where $(V,E)$ 
is a graph, $\Own$ is a map $\Own: V \to \{ \eloi, \abel \}$, and $\WC$ is a 
set of infinite paths in $(V,E)$.
We denote $G_{\Pi} \isdef \Own^{-1}(\Pi)$.
An \emph{initialised game} is a pair consisting of a game $\bbG$ and an element
$v$ of $V$; such a pair is usually denoted as $\bbG@v$.

We will refer to $(V,E)$ as the \emph{board} or \emph{arena} of the game. 
Elements of $V$ will be called \emph{positions}, and $\Own(v)$ is the 
\emph{owner} of $v$.
Given a position $v$ for player $\Pi \in \{ \eloi, \abel\}$, the set $E[v]$ 
denotes the set of \emph{moves} that are \emph{legitimate} or \emph{admissible
to} $\Pi$ at $v$.
The set $\WC$ is called the \emph{winning condition} of the game.

A \emph{match} of an initialised game consists of the two players moving a token
from one position to another, starting at the initial position, and following 
the edge relation $E$.
Formally, a \emph{match} or \emph{play} of the game $\bbG = (V,E,\Own,\WC)$
starting at position $v_{I}$ is simply a path $\pi$ through the graph $(V,E)$
such that $\first(\pi) = v_{I}$.
Such a match $\pi$ is \emph{full} if it is maximal as a path, that is, either
finite with $E[\last(\pi)] = \nada$, or infinite.
The owner of a position is responsible for moving the token from that position 
to an adjacent one (that is, an $E$-successor); in case this is impossible 
because the node has no $E$-successors, the player \emph{gets stuck} and 
immediately loses the match.
If neither player gets stuck, the resulting match is infinite; we declare 
$\eloi$ to be its winner if the match, as an $E$-path, belongs to the set $\WC$.
Full matches that are not won by $\eloi$ are won by $\abel$.

Given these definitions, it should be clear that it does not matter which player
owns a state that has a unique successor; for this reason we often take $\Own$ to
be a \emph{partial} map, provided $\Own(v)$ is defined whenever $\sz{E[v]} 
\neq 1$.

A position $v$ is a winning position for a player if they have a way of playing
the game that guarantees they win the resulting match, no matter how their 
opponent plays. 
To formalise this, we let $\PM_{\Pi}$ denote the collection of partial matches 
$\pi$ ending in a position $\last(\pi) \in V_{\Pi}$, and define $\PM_{\Pi}@v$ as
the set of partial matches in $\PM_{\Pi}$ starting at position $v$.
A \emph{strategy for a player} $P$ is a function $f: \PM_{P} \to V$; if $f(\pi)
\not\in E[\last(\pi)]$, for some $\pi \in \PM_{P}$, we say that $f$ prescribes 
an \emph{illegitimate move} in $\pi$.
A match $\pi = (v_{i})_{i<\kappa}$ is \emph{guided} by a $P$-strategy $f$ if 
$f(v_{0}v_{1}\cdots v_{n-1}) = v_{n}$ for all $n<\kappa$ 
such that $v_{0}\cdots v_{n-1}\in \PM_{P}$.
A position $v$ is \emph{reachable} by a strategy $f$ is there is an $f$-guided
match $\pi$ with $v = \last(\pi)$.
A $P$-strategy $f$ is \emph{legitimate from a position $v$} if the moves that it
prescribes to $f$-guided partial matches in $\PM_{P}@v$ are always legitimate, 
and \emph{winning for $P$ from $v$} if in addition $P$ wins all $f$-guided full
matches starting at $v$.
When defining a strategy $f$ for one of the players in a board game, we can 
and in practice will confine ourselves to defining $f$ for partial matches 
that are themselves guided by $f$.
A position $v$ is a \emph{winning position} for player $P \in \{ \eloi, \abel\}$
if $P$ has a winning strategy in the game $\bbG@v$; the set of these positions 
is denoted as $\Win_{P}(\bbG)$.
The game $\bbG$ is \emph{determined} if every position is winning for either 
$\eloi$ or $\abel$.

A strategy is \emph{positional} if it only depends on the last position of a 
partial match, i.e., if $f(\pi) = f(\pi')$  whenever $\last(\pi) = \last(\pi')$;
such a strategy can and will be presented as a map $f: V_{P} \to V$.

A \emph{priority map} on the board $V$ is a map $\Om: V \to \om$ with finite 
range.
A \emph{parity game} is a board game $\bbG = (V,E,\Own,\WC_{\Om})$ in which the
winning condition $\WC_{\Om}$ is given as follows.
Given an infinite match $\pi$, let $\Inf(\pi)$ be the set of positions that
occur infinitely often in $\pi$; then $\WC_{\Om}$ consists of those infinite 
paths $\pi$ such that $\max\big(\Om[\Inf(\pi)]\big)$ is even.
Such a parity game is usually denoted as $\bbG = (V,E,\Own,\Om)$.
The following fact is independently due to Emerson \& Jutla~\cite{emer:tree91}
and Mostowski~\cite{most:game91}.

\begin{fact}[Positional Determinacy]
\label{f:pdpg}
Let $\bbG = (G,E,\Own,\Om)$ be a parity game.
Then $\bbG$ is determined, and both players have positional winning strategies.
\end{fact}

\end{document}